\documentclass[11pt,notitlepage,tightenlines,footinbib,prx,reprint,superscriptaddress]{revtex4-2}

\usepackage{newpxtext,newpxmath}

\let\coloneqq\relax
\let\eqqcolon\relax

\usepackage[latin1]{inputenc}
\usepackage{amsthm}
\usepackage{amssymb}
\usepackage{amsmath}
\usepackage{bbold}
\usepackage{bbm}
\usepackage[pdftex, backref=page]{hyperref}
\usepackage{braket}
\usepackage{dsfont}
\usepackage{mathdots}
\usepackage{mathtools}
\usepackage{enumerate}
\usepackage[shortlabels]{enumitem}
\usepackage{csquotes}
\usepackage{stmaryrd}
\usepackage[cal=boondox]{mathalfa}
\usepackage{graphicx}
\usepackage{stackengine}
\usepackage{scalerel}
\usepackage{tensor}       %\tensor[_n]{\braket{j|\psi}}{_{1\ldots n}}
\usepackage{array}
\usepackage{makecell}
\newcolumntype{x}[1]{>{\centering\arraybackslash}p{#1}}
\usepackage{tikz}
\usepackage{pgfplots}
\usetikzlibrary{shapes.geometric, shapes.misc, positioning, arrows, arrows.meta, decorations.pathreplacing, decorations.pathmorphing, patterns, angles, quotes, calc}
\usepackage{booktabs}
\usepackage{xfrac}
\usepackage{siunitx}
\usepackage{centernot}
\usepackage{comment}
\usepackage{chngcntr}
\usepackage{caption}
\usepackage{subcaption}

\newtheorem{thm}{Theorem}
\newtheorem*{thm*}{Theorem}
\newtheorem{prop}[thm]{Proposition}
\newtheorem*{prop*}{Proposition}
\newtheorem{lemma}[thm]{Lemma}
\newtheorem*{lemma*}{Lemma}

\newtheorem*{cor*}{Corollary}

\newtheorem*{cj*}{Conjecture}
\newtheorem{Def}[thm]{Definition}
\newtheorem*{Def*}{Definition}

\newtheorem*{question*}{Question}

\newtheorem*{problem*}{Problem}

\makeatletter
\def\thmhead@plain#1#2#3{%
  \thmname{#1}\thmnumber{\@ifnotempty{#1}{ }\@upn{#2}}%
  \thmnote{ {\the\thm@notefont#3}}}
\let\thmhead\thmhead@plain
\makeatother

\theoremstyle{definition}
\newtheorem{rem}[thm]{Remark}

\newtheorem{ex}[thm]{Example}

\newcommand{\bb}{\begin{equation}\begin{aligned}\hspace{0pt}}
\newcommand{\bbb}{\begin{equation*}\begin{aligned}}
\newcommand{\ee}{\end{aligned}\end{equation}}
\newcommand{\eee}{\end{aligned}\end{equation*}}
\newcommand*{\coloneqq}{\mathrel{\vcenter{\baselineskip0.5ex \lineskiplimit0pt \hbox{\scriptsize.}\hbox{\scriptsize.}}} =}
\newcommand*{\eqqcolon}{= \mathrel{\vcenter{\baselineskip0.5ex \lineskiplimit0pt \hbox{\scriptsize.}\hbox{\scriptsize.}}}}
\newcommand\floor[1]{\left\lfloor#1\right\rfloor}

\newcommand{\eqt}[1]{\stackrel{\mathclap{\scriptsize \mbox{#1}}}{=}}
\newcommand{\leqt}[1]{\stackrel{\mathclap{\scriptsize \mbox{#1}}}{\leq}}

\newcommand{\ketbra}[1]{\ket{#1}\!\!\bra{#1}}
\newcommand{\Ketbra}[1]{\Ket{#1}\!\!\Bra{#1}}
\newcommand{\ketbraa}[2]{\ket{#1}\!\!\bra{#2}}

\newcommand{\sumno}{\sum\nolimits}

\newcommand{\dd}{\mathrm{d}}

\newcommand{\tcr}[1]{{\color{red} #1}}

\newcommand{\id}{\mathds{1}}
\newcommand{\R}{\mathds{R}}
\newcommand{\N}{\mathds{N}}

\newcommand{\C}{\mathds{C}}

\newcommand{\locc}{\mathrm{LOCC}}

\DeclareMathOperator{\Tr}{Tr}

\DeclareMathAlphabet{\pazocal}{OMS}{zplm}{m}{n}

\DeclareMathOperator{\supp}{supp}
\DeclareMathOperator{\spec}{sp}

\DeclareMathOperator{\vol}{vol}

\newcommand{\HH}{\pazocal{H}}
\newcommand{\T}{\pazocal{T}}

\newcommand{\K}{\pazocal{K}}

\newcommand{\B}{\pazocal{B}}

\newcommand{\EE}{\mathcal{E}}

\newcommand{\lsmatrix}{\left(\begin{smallmatrix}}
\newcommand{\rsmatrix}{\end{smallmatrix}\right)}

\stackMath
\newcommand\xxrightarrow[2][]{\mathrel{%
  \setbox2=\hbox{\stackon{\scriptstyle#1}{\scriptstyle#2}}%
  \stackunder[5pt]{%
    \xrightarrow{\makebox[\dimexpr\wd2\relax]{$\scriptstyle#2$}}%
  }{%
   \scriptstyle#1\,%
  }%
}}

\newcommand{\tends}[2]{\xxrightarrow[\! #2 \!]{\mathrm{#1}}}

\stackMath

\makeatletter
\newcommand*\rel@kern[1]{\kern#1\dimexpr\macc@kerna}
\newcommand*\widebar[1]{%
  \begingroup
  \def\mathaccent##1##2{%
    \rel@kern{0.8}%
    \overline{\rel@kern{-0.8}\macc@nucleus\rel@kern{0.2}}%
    \rel@kern{-0.2}%
  }%
  \macc@depth\@ne
  \let\math@bgroup\@empty \let\math@egroup\macc@set@skewchar
  \mathsurround\z@ \frozen@everymath{\mathgroup\macc@group\relax}%
  \macc@set@skewchar\relax
  \let\mathaccentV\macc@nested@a
  \macc@nested@a\relax111{#1}%
  \endgroup
}

\counterwithin*{equation}{part}
\counterwithin*{thm}{part}
\counterwithin*{figure}{part}

\tikzset{meter/.append style={draw, inner sep=10, rectangle, font=\vphantom{A}, minimum width=30, line width=.8, path picture={\draw[black] ([shift={(.1,.3)}]path picture bounding box.south west) to[bend left=50] ([shift={(-.1,.3)}]path picture bounding box.south east);\draw[black,-latex] ([shift={(0,.1)}]path picture bounding box.south) -- ([shift={(.3,-.1)}]path picture bounding box.north);}}}
\tikzset{roundnode/.append style={circle, draw=black, fill=gray!20, thick, minimum size=10mm}}
\tikzset{squarenode/.style={rectangle, draw=black, fill=none, thick, minimum size=10mm}}

\definecolor{Blues5seq1}{RGB}{239,243,255}
\definecolor{Blues5seq2}{RGB}{189,215,231}
\definecolor{Blues5seq3}{RGB}{107,174,214}
\definecolor{Blues5seq4}{RGB}{49,130,189}
\definecolor{Blues5seq5}{RGB}{8,81,156}

\definecolor{Greens5seq1}{RGB}{237,248,233}
\definecolor{Greens5seq2}{RGB}{186,228,179}
\definecolor{Greens5seq3}{RGB}{116,196,118}
\definecolor{Greens5seq4}{RGB}{49,163,84}
\definecolor{Greens5seq5}{RGB}{0,109,44}

\definecolor{Reds5seq1}{RGB}{254,229,217}
\definecolor{Reds5seq2}{RGB}{252,174,145}
\definecolor{Reds5seq3}{RGB}{251,106,74}
\definecolor{Reds5seq4}{RGB}{222,45,38}
\definecolor{Reds5seq5}{RGB}{165,15,21}

\usepackage{lipsum}

\renewcommand{\mathcal}{\pazocal}

\newcommand{\HS}{\pazocal{H\!S}}

\newcommand{\cls}{\hspace{-0.8pt}c\hspace{-0.8pt}\ell}
\newcommand{\vre}{\varepsilon}
\newcommand{\deff}[1]{\textbf{\emph{#1}}}
\newcommand{\dmin}{d}

\makeatletter
\newcommand{\raisemath}[1]{\mathpalette{\raisem@th{#1}}}
\newcommand{\raisem@th}[3]{\raisebox{#1}{$#2#3$}}
\makeatother

\usepackage{textgreek}
\tikzset{cross/.style={cross out, draw=black, thick, minimum size=2*(#1-\pgflinewidth), inner sep=0pt, outer sep=0pt}, cross/.default={1mm}}
\usepackage[most]{tcolorbox}

\tcbset{colback=Blues5seq1, colframe=Blues5seq5, highlight math style= {enhanced, colframe=Blues5seq5,colback=Blues5seq1,boxsep=0pt}}
\tikzset{roundnode/.append style={circle, draw=black, fill=gray!20, thick, minimum size=10mm}}
\tikzset{redsquarenode/.style={rectangle, draw=black, fill=Reds5seq3, opacity=1, thick, minimum size=10mm}}

\allowdisplaybreaks

\let\tcr\relax

\begin{document}

\title{Testing the quantum nature of gravity without entanglement}

\author{Ludovico Lami}
\email{ludovico.lami@gmail.com}
\affiliation{QuSoft, Science Park 123, 1098 XG Amsterdam, The Netherlands}
\affiliation{Korteweg--de Vries Institute for Mathematics, University of Amsterdam, Science Park 105-107, 1098 XG Amsterdam, The Netherlands}
\affiliation{Institute for Theoretical Physics, University of Amsterdam, Science Park 904, 1098 XH Amsterdam, The Netherlands}
\affiliation{Institut f\"{u}r Theoretische Physik und IQST, Universit\"{a}t Ulm, Albert-Einstein-Allee 11, D-89069 Ulm, Germany}

\author{Julen S. Pedernales}
\email{julen.pedernales@uni-ulm.de}
\affiliation{Institut f\"{u}r Theoretische Physik und IQST, Universit\"{a}t Ulm, Albert-Einstein-Allee 11, D-89069 Ulm, Germany}

\author{Martin B. Plenio}
\email{martin.plenio@uni-ulm.de}
\affiliation{Institut f\"{u}r Theoretische Physik und IQST, Universit\"{a}t Ulm, Albert-Einstein-Allee 11, D-89069 Ulm, Germany}

\begin{abstract}
Given a unitary evolution $U$ on a multi-partite quantum system and an ensemble of initial states, how well can $U$ be simulated by local operations and classical communication (LOCC) on that ensemble? We answer this question by establishing a general, efficiently computable upper bound on the maximal LOCC simulation fidelity --- what we call an `LOCC inequality'. We then apply our findings to the fundamental setting where $U$ implements a quantum Newtonian Hamiltonian over a gravitationally interacting system. Violation of our LOCC inequality can rule out the LOCCness of the underlying evolution, thereby establishing the non-classicality of the gravitational dynamics, which can no longer be explained by a local classical field. 
%\mbp{Comment: The next half sentence is maybe a litle early here and also repeats the last sentence to some extent.} , even when no entanglement is ever generated at any point in the process. 
%In other words, if an experiment violates our inequality then a description of the gravitational interaction in terms of a local classical field is no longer possible. 
As a prominent application of this scheme we study systems of quantum harmonic oscillators initialised in coherent states following a normal distribution and interacting via Newtonian gravity, and discuss a possible physical implementation with torsion pendula. One of our main technical contributions is the analytical calculation of the above LOCC inequality for this family of systems. As opposed to existing tests based on the detection of gravitationally mediated entanglement, our proposal works with coherent states alone, and thus it does not require the generation of largely delocalised states of motion nor the detection of entanglement, which is never created at any point in the process.
\end{abstract}

\maketitle
%\tableofcontents

\section{Introduction} \label{sec_intro}
The weakest force of Nature, gravity, is also the cruellest. The desire to understand its elusive character has often stirred the very roots of science, breeding so many of the greatest achievements of human thought%~\cite{PRINCIPIA, Einstein1916}
; at the same time, the extreme weakness of the gravitational interaction, compared to all of the other forces, has placed us in the bitter situation of not, yet, being able to investigate its ultimate nature. The question of whether gravity is fundamentally a classical or a quantum field, or perhaps an entity of an altogether different type, remains so far unanswered. %~\footnote{The reader who wished to be entertained might enjoy reading about a recent bet between Jonathan Oppenheim, Geoff Penington, and Carlo Rovelli on this very question~\cite{gravity_bet}.}.

While probing phenomena where the \tcr{quantum degrees of freedom of gravity could directly be tested %quantisation of the gravitational field could emerge 
remains prohibitive with current or near-term technologies, determining whether spacetime could exist in a superposition of distinct classical configurations} %the gravitational field lives in a Hilbert space --- and is thus able to experience superposition --- 
is a substantially easier, if still challenging, endeavour. A first proposal in that direction was put forth by Feynman~\cite{Feynman1957}, who presented a Gedankenexperiment in which a mass capable of generating a gravitational field is placed in a superposition of different spatial locations; such a mass could induce --- via the gravitational interaction --- a superposition in a second test mass, and the physical existence of such quantum coherences (mediated by gravity) can be tested, at least in principle.

In quantum as well as post-quantum theories, correlated superpositions naturally lead to the notion of entanglement~\cite{cones-3}. %~\cite{cones-1,cones-2,cones-3}. 
Therefore, further developing the above idea leads to the natural prediction that no classical gravitational field can entangle two quantum mechanical systems~\cite{Kafri-Taylor-1, Krisnanda2017, Galley2020}. 
Recently, proposals seeking to detect gravitationally generated entanglement~\cite{lindner2005testing, Kafri-Taylor-1, Schmoele2016, Bose2017, Krisnanda2020, Miao2020, Chevalier2020, pedernales2020motional, cosco2021enhanced, weiss2021large, Pedernales2022}
have prompted a discussion of the conclusions that could be drawn from these experiments regarding the quantum nature of gravity~\cite{Danielson2022, Christodoulou2022, Christodoulou2022-2, Sidajaya2022, Fragkos2022, Doener2022, Carney2022, Streltsov2022, Spaventa2023}. 
The concept underlying these proposals is the creation of coherently delocalised quantum states of two masses, in such a way that different branches of the wavefunction of one mass have different distances to those of the other mass. As a result of the gravitational interaction, evolution in time will assign diverging phases to those components, effectively entangling the two systems. Such entanglement can then be detected experimentally in various ways. For an introduction to these topics, we refer the reader to the recent review~\cite{tabletop}.

It is worth noting that besides the impossibility of entanglement generation, a classical gravity would be subjected to several other constraints that may also have observable consequences. For example, the continuous measurement to which a quantum source should be subjected by hand of a classical gravity should produce a decoherence effect~\cite{Diosi-thesis, Diosi2011, Kafri-Taylor-1, Kafri-Taylor-2, Kafri-Taylor-3, Pfister2016, Oppenheim2022b}. Other proposed experiments seek to exploit the spatial superposition of source masses~\cite{lindner2005testing, Schmoele2016, Carlesso2017, Carlesso2019, Haine2021}, or aim to detect the possible non-Gaussian signatures of a genuine quantum theory of gravity~\cite{Howl2021}.

For the sake of this discussion, let us focus our attention on the gravitational entanglement proposals. In layman's terms, the main idea underlying these proposals is to `challenge' gravity to create an entangled state between separated systems, starting from a situation where no entanglement is present. \emph{If the gravitational interaction --- assumed to be local --- were to be mediated by a classical field, the time evolution operator it could induce on the system would necessarily~\cite{Kafri-Taylor-1, Kafri-Taylor-2, Kafri-Taylor-3, Bose2017, Marletto2017, Krisnanda2017} be representable as a sequence of local operations assisted by classical communication (LOCC)~\cite{LOCC}.} This is sometimes known as the \emph{LOCC theorem}. Since LOCC operations cannot turn an unentangled state into an entangled one~\cite{Plenio-Virmani,Horodecki-review}, the above task is impossible to accomplish for a classical gravity. To generate that entanglement, instead, gravity will have to resort to its quantum nature, assuming it has one, and by doing so it will reveal it to us. \tcr{Note that here we have implicitly assumed a very specific meaning for the term \emph{classical gravity}, namely, any gravitational interaction between quantum objects that can be effectively described by LOCC operations. This, in turn, also sets the meaning of the term \emph{quantum gravity} as any gravitational interaction that cannot be described by LOCC operations. This delineation, which is a matter of choice, establishes the nomenclature we will adhere to for the remainder of this article.}

The main issue with the experimental implementation of the above proposals is that placing in a coherently delocalised quantum state masses that are sufficiently large to generate an appreciable gravitational field has proved to be a formidable task. To get an idea of the challenges associated with these proposals, consider that the heaviest mass for which spatial superposition, in the form of matter-wave interference, has been observed is a large molecule of mass $\sim \SI{4e-23}{\kg}$~\cite{Fein2019}. By contrast, the smallest source mass whose gravitational field has been measured
is just below $\SI{100}{\mg}$~\cite{westphal2021measurement}.

\subsection{The idea} \label{subsec_idea}

As we have seen, the gravitational entanglement proposal rests on the LOCC theorem. However, it does not exploit its full power: on the contrary, the only consequence drawn from the theorem is that no entanglement can be generated by a classical gravitational field when the system is initially prepared in some unentangled state. Even if applied to all unentangled inputs --- which is often not considered --- this amounts, in quantum information parlance, to the statement that LOCC operations are non-entangling. %~\cite{BrandaoPlenio1}. 
However, one key insight from the theory of entanglement manipulation is that the set of non-entangling operations is much larger than that of LOCCs~\cite{LOCC, irreversibility}. 
For an extreme example of this difference, consider the swap operator, acting on a bipartite system $AB$ as $F_{AB} \ket{\psi}_A\ket{\phi}_B = \ket{\phi}_A\ket{\psi}_B$: this operation cannot generate any entanglement when acting on an unentangled state, yet it is highly non-LOCC --- indeed, it involves the perfect exchange of quantum information between the two parties, which cannot be accomplished with a classical communication channel~\cite{eisert2000optimal}. 
%More generally, the difference between the set of LOCCs and that of non-entangling operations is so radical that it leads to spectacularly different entanglement distillation capabilities~\footnote{Indeed, while there exist states that are \emph{bound entangled} under LOCCs~\cite{HorodeckiBound}, i.e.\ such that no pure entanglement can be extracted, or \emph{distilled}, from them when only LOCC operations are applied, it has been recently shown that all entangled states are distillable under non-entangling operations~\cite[Proposition~7]{gap}.}. 
In other words, there is no way for two separate parties Alice and Bob who do not have access to a quantum communication channel --- e.g.\ an optical fibre --- to successfully exchange two unknown quantum states $\ket{\psi}_A$ and $\ket{\phi}_B$.

What the above discussion teaches us is that there are so many more constraints that a quantum evolution should obey, besides the mere absence of entanglement generation, in order for it to be a genuine LOCC. As a result, the failure to observe entanglement generation does not imply the classicality of the gravitational interaction, in the same way as observing an entirely entanglement-free swap $\ket{\psi}_A\ket{\phi}_B \mapsto \ket{\phi}_A\ket{\psi}_B$ should not lead us to conclude that the underlying process is purely classical --- on the contrary, there \emph{must} be some exchange of quantum information going on between $A$ and $B$. The fundamental general question we will tackle here is therefore:
\begin{center}
\emph{how do we decide whether a particular set of observations of a process is incompatible with an LOCC dynamics?}
\end{center}
In a nutshell, our main contribution is to answer the above question by figuring out some of the constraints that every LOCC-model should obey. We dub these constraints \emph{LOCC inequalities}, because the role they play here is somewhat similar to the role of Bell inequalities in the theory of quantum non-locality~\cite{Brunner-review}: as the violation of a Bell inequality certifies that the underlying process, whose inner workings we ignore, is unequivocally non-local, in the same way the violation of one of our LOCC inequalities certifies that the unknown evolution undergone by the system cannot be 
%is most certainly not 
LOCC, i.e.\ it cannot be explained by introducing a mediator which is a local classical field. We then proceed to show how these LOCC inequalities can be applied to rule out the LOCCness of a time evolution associated with a gravitational Hamiltonian, even when no entanglement whatsoever is generated at any point in time during the process. 
%To achieve this, we develop and employ techniques from quantum information, and in particular from entanglement theory.
%Developing techniques from entanglement theory in quantum information, we answer the above question by proving some LOCC inequalities, i.e.\ inequalities that any LOCC evolution should obey.

%Put differently, we sharpen the analysis of the consequences entailed by a local classical theory of gravity by looking not only at the preparation of a single output state. 
Another way to look at our conceptual step forward is that while analysing gravitationally interacting quantum systems we do not look only at the preparation of a single output state. Rather, we consider an ensemble of possible input states and their resulting output. In this way, we can better characterise the channel that describes the full dynamics under gravity after a set time. This is a substantial conceptual improvement, because quantum channels, modelling the result of the time evolution to which open quantum systems are subjected, are more general and thus harder to simulate than states. It is also a practically relevant improvement, because, as we have seen with the swap example above, there are several dynamical processes that never involve any entanglement but are still impossible to reproduce by classical means. Resorting to these dynamical tests has the potential to by-pass the many problems related with entanglement generation and maintenance in realistic experimental environments, while retaining the capability of testing the quantum nature %quantumness 
of the gravitational interaction.

To understand how this is possible at all, let us start by looking at a simple toy model. Imagine two identical quantum harmonic oscillators with masses $m$ and frequencies $\omega$, oscillating along the same straight line, and whose centres are separated by a distance $d$. The local Hamiltonian takes the form $H_1+H_2 = \sum_{i=1}^2 \left(\frac12 m\omega^2 x_i^2 + \frac{p_i^2}{2m} \right)$. At this point, we would like to include the gravitational interaction into the picture. \emph{If gravity is to be quantised}, general arguments~\cite{tabletop} suggest that this can be done, in the low-energy limit, by simply adding the Newtonian potential to the local Hamiltonian 
and treating the resulting expression
\bb
H_1+H_2 +H_{G} = \sum_{i=1}^2 \left(\frac12 m\omega^2 x_i^2 + \frac{p_i^2}{2m} \right) - \frac{Gm^2}{|d-x_1+x_2|}\, ,
\label{H_tot_two}
\ee
where the gravitational constant is given by $G = \SI{6.6743e-11}{\m^3\kg^{-1}\s^{-2}}$, as a \emph{quantum} Hamiltonian acting on the bipartite system~\cite{tabletop}. %~\cite{Weinberg1964, Weinberg1965, Boulware1975, Hamber1995, tabletop}. 
For a fully consistent derivation of this claim, we refer the reader to~\cite{Christodoulou2022}. If the position variance of each oscillator is much smaller than $d$, we can perform a Taylor expansion of $H_{G}$. Furthermore, although this is not strictly necessary, we can also assume that the gravitational interaction is much weaker than the local harmonic potential, i.e.\ $\frac{Gm}{d^3 \omega^2}\ll 1$; in this situation, we can safely make the rotating-wave approximation as well. This yields the total effective Hamiltonian~\cite{tabletop}
\bb
H_{\text{eff}} &= \sum_{i=1}^2 \left(\frac12 m \omega^2 x_i^2 + \frac{p_i^2}{2m} \right) + \frac{Gm^2}{d^3} \left(x_1x_2+\frac{p_1p_2}{(m\omega)^2} \right) .
\label{H_eff_two}
\ee
Let us assume that the two-oscillator system starts off in the product coherent state $\ket{\alpha} \otimes \ket{\beta}$, where $\alpha,\beta\in \C$~\cite{BUCCO}. %~\cite{Schroedinger1926-coherent, Klauder1960, Glauber1963, Sudarshan1963}. 
After some time $t$, it will have evolved to the state
\bb
e^{-\frac{it}{\hbar} H_{\text{eff}} } \left(\ket{\alpha}\! \otimes\! \ket{\beta}\right) &= \ket{e^{i\omega t} \left(\cos (\gamma t) \alpha - i\sin (\gamma t) \beta \right) } \\
&\quad \otimes \ket{e^{i\omega t} \left(\cos (\gamma t) \beta - i\sin (\gamma t) \alpha \right) } ,
\label{elementary_evolution}
\ee
where $\gamma = \frac{Gm}{d^3 \omega}$. Note that this is a product state, i.e.\ it contains no entanglement; the same can be said of the initial state. However, there exists no $(\alpha,\beta)$-independent LOCC protocol implementing the transformation $\ket{\alpha}\otimes \ket{\beta} \to e^{-\frac{it}{\hbar} H_{\text{eff}}} \left(\ket{\alpha}\! \otimes\! \ket{\beta}\right)$ for all $\alpha,\beta$~\footnote{Such a protocol does exist, however, if we allow dependence on the input coherent states. Also, it is worth observing that the particular dynamics we consider \emph{can} generate entanglement if applied to a different unentangled input state --- for example squeezed states. We will return to these points later.}. To see why this must be the case, consider a time $t_0$ such that $\gamma t_0 = \pi/2$. Then
\bb
e^{-\frac{it_0}{\hbar} H_{\text{eff}}} \left(\ket{\alpha}\! \otimes\! \ket{\beta}\right) &= \ket{-i e^{i \omega t_0} \beta} \otimes \ket{-i e^{i \omega t_0} \alpha} .
\label{even_more_elementary_evolution}
\ee
Up to an immaterial local phase $-ie^{i\omega t_0}$, we have effectively swapped the two systems. As we discussed above, this is not possible by means of local operations and classical communication alone, without the assistance of some quantum communication channel and without knowing the values of $\alpha$ and $\beta$. Therefore, if we could ascertain that the quantum systems evolve as in~\eqref{elementary_evolution} for arbitrary coherent state inputs, we could immediately rule out the classicality of the gravitational field!

Although we take this qualitative observation as the starting point of our discussion, more work is needed to turn it into quantitative --- and testable --- predictions. Indeed, in any realistic scenario uncontrolled sources of noise will always alter the ideal time evolution~\eqref{elementary_evolution}. What we can observe, at best, is some approximate version of~\eqref{elementary_evolution}, and our intuition tells us that if the approximation is sufficiently good, the impossibility of implementing the dynamics via LOCC should nevertheless hold. And indeed, this intuition can be rigorously validated by employing our LOCC inequalities.

Although motivated by the quest for detection of the quantum nature %quantumness 
of gravity, the scope of %the above question 
our analysis is actually broader. In fact, the scheme we present here is general enough to be applicable to any interaction whose quantum nature %quantumness 
we wish to certify. Our main conceptual contribution in answering this question is that we do away with any assumption on what particular LOCC operation the unknown interaction may implement. When applied to gravity, this is especially crucial: indeed, on the one hand it is perfectly conceivable that the nature of the gravitational field is classical instead of quantum~\cite{gravitization-1}, 
%--- according to this school of thought, ultimately quantum mechanics will be `gravitised' instead of gravity quantised~\cite{gravitization-1, gravitization-2} --- 
while on the other there is currently no consensus as to what a classical theory of the gravitational field interacting with quantum sources should look like. Proposals that are more or less ambitious range from %the Schr\"odinger--Newton equation~\cite{Moller1962, Ruffini1969, Kibble1980} to 
collapse models~\cite{Karolyhazy1966, Diosi1987, Diosi1989, Penrose1996, Kafri-Taylor-2, Tilloy2016} to fully-fledged theories~\cite{Oppenheim2018, Oppenheim2022, Oppenheim2023}. Because of all these theoretical alternatives, our treatment should and will encompass any possible LOCC implemented on a multi-partite quantum system. In this conceptually comprehensive setting, the above question becomes technically challenging, because the set of LOCCs is well known to be exceedingly difficult to characterise mathematically~\cite{LOCC}.

\subsection{Summary of results} \label{subsec_summary_results}

Our first contribution is to identify a relevant quantum information theoretic task, simulation of a unitary dynamics on an ensemble of states via LOCCs, and to define a quantitative figure of merit for it, the \emph{LOCC} (or \emph{classical}) \emph{simulation fidelity}, that in principle can be measured experimentally. This is described in abstract terms in \S~\ref{subsec_unitary_simulation_LOCCs}, while in \S~\ref{subsec_experiments} we explore how this scheme could be applied to design a very general class of experiments to determine the quantum nature %quantumness 
of gravity.

Our first main theoretical contribution is Theorem~\ref{general_bound_thm} in \S~\ref{subsec_general_benchmarks}, where we state our most general LOCC inequality. This takes the form of an upper bound on the maximal fidelity that an LOCC channel can achieve when attempting to simulate a given global unitary $U$ on a multi-partite quantum system initially prepared in an unknown pure state taken from an ensemble $\{p_\alpha, \psi_\alpha\}_\alpha$. By applying it, we can \emph{benchmark} the aforementioned experiments, that is, we can determine quantitatively the threshold above which a successful detection of the quantum nature %quantumness 
of gravity can be claimed (under our assumptions). To prove Theorem~\ref{general_bound_thm} we employ mathematical tools taken from quantum information science, mainly the min-entropy formalism, the Choi--Jamio{\l}kowski isomorphism, the theory of LOCC operations, and the positive partial transpose (PPT) entanglement criterion. As a consequence, our approach actually yields upper bounds on the simulation fidelity achievable with PPT preserving operations, which are a strictly larger and more powerful class than LOCC operations.

In \S~\ref{sec_proposed_implementation} we describe 
%within the broad family sketched out in the previous \S~\ref{sec_dynamical_experiments}, 
a particular family of thought experiments involving gravitationally interacting quantum harmonic oscillators, initially prepared in random coherent states. A detailed exposition can be found in \S~\ref{subsec_proposal}. Applying the general results of Theorem~\ref{general_bound_thm} to this setting is itself technically challenging, and the associated calculation constitutes one of our main mathematical contributions. In \S~\ref{subsec_quantitative_analysis} we solve this problem by means of Theorems~\ref{general_symplectic_bound_thm} and~\ref{recap_thm}; a simplified expression is obtained by looking at experiments taking place on small enough time scales (Theorem~\ref{recap_small_times_thm}). 
%We discuss what these results mean in \S~\ref{subsec_discussion}. 
In \S~\ref{subsec_examples} we look at some concrete examples, while in \S~\ref{subsec_experimental_considerations} we consider possible implementations of our proposal with torsion pendula, touching upon some of the associated experimental challenges. Finally, in \S~\ref{sec_conclusions} we summarise our work, comparing our proposal with some existing ones and discussing relative strenghts and weaknesses.

The technical proofs of Theorems~\ref{general_symplectic_bound_thm},~\ref{recap_thm}, and~\ref{recap_small_times_thm} are presented in Appendices~\ref{subsec_proof_general_symplectic_bound}--\ref{subsec_proof_recap_small_times}. We strive to achieve a high level of mathematical rigour, which is in itself one of our foremost conceptual contributions. For example, most of our approximations come with rigorous quantitative estimates. As a by-product of our approach, we show how one can re-derive and in fact extend previous results on the theory of quantum benchmarks for teleportation~\cite{Massar1995, Werner1998, Horodecki-teleportation, Bruss1999, Hammerer2005, BUCCO}. This is detailed in Appendix~\ref{subsec_telep_threshold}. Besides being more general, our proofs have the advantage of being more transparent and shorter.

Let us summarise the organisation of the paper as follows:
\begin{enumerate} %\setcounter{enumi}{1}

    \item[\S~\ref{sec_preliminaries}] We introduce our notation: quantum states and channels, continuous-variable systems and Gaussian operations, entanglement theory, and conditional min-entropy.

    \item[\S~\ref{sec_dynamical_experiments}] In \S~\ref{subsec_unitary_simulation_LOCCs} we define the task of unitary simulation via LOCCs. We describe how this could be applied to detecting the quantum nature %quantumness 
    of gravity in \S~\ref{subsec_experiments}. A general LOCC inequality is derived in \S~\ref{subsec_general_benchmarks}. 
    %In \S~\ref{subsec_telep_threshold} we take a little detour to show how our methods can be used to derive and generalise many prior results in the theory of quantum benchmarks in a more transparent way.

    \item[\S~\ref{sec_proposed_implementation}] In \S~\ref{subsec_proposal} we describe a particular family of experiments with gravitationally interacting harmonic oscillators. Applying the general results of Theorem~\ref{general_bound_thm} to this setting results in Theorems~\ref{general_symplectic_bound_thm},~\ref{recap_thm}, and~\ref{recap_small_times_thm} in \S~\ref{subsec_quantitative_analysis}. Proofs of these results can be found in Appendices~\ref{subsec_proof_general_symplectic_bound}--\ref{subsec_proof_recap_small_times}. Some examples are examined in \S~\ref{subsec_examples}, while experimental considerations and comparisons with prior proposals are deferred to \S~\ref{subsec_experimental_considerations}. Finally, \S~\ref{Assessment} is devoted to a more detailed assessment of a possible setup based on torsion pendula.

    \item[\S~\ref{sec_conclusions}] A thorough discussion of the implications of our findings is presented in \S~\ref{subsec_discussion}. There we also compare our proposal with previous ones, highlighting the main conceptual differences and the practical consequences these entail.

%   \item[\S~\ref{sec_proofs}] This section is occupied by the complete mathematical proofs of Theorems~\ref{general_symplectic_bound_thm},~\ref{recap_thm}, and~\ref{recap_small_times_thm}.

\end{enumerate}

\section{Notation} \label{sec_preliminaries}

This section contains a description of the (rather standard) notation to be used throughout the paper.

A generic quantum system is mathematically described by a complex Hilbert space $\HH$, not necessarily finite dimensional. A prominent role in the theory is played by two spaces~\footnote{If $\dim\HH = n <\infty$, all three of these can be identified with the space of $n\times n$ complex matrices.} of operators acting on $\HH$: that of all trace class operators, $\T(\HH)$, equipped with the \deff{trace norm}
\bb
\|X\|_1\coloneqq \Tr \sqrt{X^\dag X}\, ,
\label{trace_norm}
\ee
and that of all bounded operators, $\B(\HH)$, equipped with the \deff{operator norm}
\bb
\|X\|_\infty \coloneqq \sup_{\ket{\psi}\in \HH,\, \|\ket{\psi}\|=1} \left\|X\ket{\psi}\right\| .
\label{operator_norm}
\ee
For example, physical states are represented by density operators, i.e.\ positive semi-definite trace class operators with trace one. The simplest such states are pure states, of the form $\psi = \ketbra{\psi}$ for some $\ket{\psi}\in \HH$ with unit norm $\left\|\ket{\psi}\right\|=1$. The resemblance between two states $\rho$ and $\sigma$ of the same quantum system can be quantified via the \emph{Uhlmann fidelity}, or simply the \emph{fidelity}, defined by~\cite{Uhlmann-fidelity} %~\cite{Uhlmann-fidelity, WILDE}
\bb
F(\rho,\sigma) \coloneqq \left\|\sqrt{\rho}\sqrt{\sigma}\right\|_1^2 .
\label{fidelity}
\ee
Importantly, if one of the two states is pure then~\eqref{fidelity} reduces to the overlap, i.e.
\bb
F(\rho,\psi) = \braket{\psi|\rho|\psi}\, .
\label{fidelity_overlap}
\ee

A generic quantum measurement is represented by a normalised positive operator-valued measure (POVM)~\cite[Definition~11.29]{HOLEVO-CHANNELS-2}. In our case, we will mostly care about the case of binary measurements, modelled by pairs of bounded operators $(E_0, E_1)$ on $\HH$ satisfying $E_i\geq 0$ for all $i=0,1$ and $E_0+E_1=\id$. This measurement carried out on the state $\rho$ yields as outcome $i\in \{0,1\}$ with probability
\bb
P(i|\rho) = \Tr \rho E_i\, .
\label{outcome}
\ee

Consider now two generic quantum systems $A,A'$ with Hilbert spaces $\HH_A, \HH_{A'}$. A quantum channel $\Lambda$ from $A$ to $A'$, denoted $\Lambda:A\to A'$, is a physically implementable transformation between states of $A$ and states of $A'$~\cite{Stinespring}. It can be represented mathematically by a linear map $\Lambda:\T(\HH_A) \to \T(\HH_{A'})$ that is in addition completely positive~\cite{Choi} and trace preserving.

\subsection{Continuous-variable systems} \label{subsec_CV_systems}

A $n$-mode continuous-variable system has Hilbert space $\HH_n\coloneqq L^2(\R^n)$. Arranging the dimensionless canonical operators $\bar{x}_j,\bar{p}_k$ in a vector
\bb
\bar{r}\coloneqq (\bar{x}_1,\bar{p}_1,\ldots, \bar{x}_n,\bar{p}_n)^\intercal
\label{r}
\ee
which we will call the mode-wise decomposition, the \deff{canonical commutation relations} take the form
\bb
[\bar{r},\bar{r}^\intercal] = i \Omega\, \id_{\HH_n}\, ,\qquad \Omega \coloneqq \begin{pmatrix} 0 & 1 \\ -1 & 0 \end{pmatrix}^{\oplus n} ,
\label{CCR}
\ee
where $\Omega$ is referred to as the \deff{standard symplectic form}. Sometimes it is useful to construct also the \deff{annihilation} and \deff{creation operators}, defined by $a_j \coloneqq \left(\bar{x}_j + i \bar{p}_j\right)/\sqrt2$ and $a_j^\dag \coloneqq \left(\bar{x}_j - i \bar{p}_j\right)/\sqrt2$, respectively, whose commutation relations can be written as $[a_j,a_k]=0$ and $[a_j, a_k^\dag] = \delta_{j,k}$.

Acting repeatedly with creation operators on the \deff{vacuum state} $\ket{0}$, defined via $a_j\ket{0}=0$ for all $j$, yields the \deff{Fock states}. Considering for the moment only the single-mode case $n=1$, these are indexed by $k\in \N$ and are given by $\ket{k}\coloneqq (k!)^{-1/2} (a^\dag)^k\ket{0}$. Other notable states are the \deff{coherent states}, parametrised by a complex number $\alpha\in \C$ and given by
\bb
\ket{\alpha} \coloneqq e^{-|\alpha|^2/2} \sum_{k=0}^\infty \frac{\alpha^k}{\sqrt{k!}}\, \ket{k}\, .
\label{coherent}
\ee
These states form an over-complete set, in the sense that $\int \frac{\dd^2\!\alpha}{\pi}\, \ketbra{\alpha} = \id$~\cite[Exercise~12.6]{HOLEVO-CHANNELS-2}. Multi-mode coherent states are simply product states of single-mode coherent states, i.e.
\bb
\ket{\alpha}\coloneqq \ket{\alpha_1}\otimes \ldots \otimes \ket{\alpha_n}\, ,
\label{coherent_multimode}
\ee
where $\alpha = (\alpha_1,\ldots,\alpha_n)^\intercal \in \C^n$.

Again in the multi-mode case of arbitrary $n$, the \deff{displacement} (unitary) operator associated with a real vector $\delta\in \R^{2n}$ is defined by
\bb
D_\delta \coloneqq e^{-i\delta^\intercal \Omega r}\, .
\label{displacement}
\ee
When acting on the vacuum state, displacement operators yield coherent states. More precisely, we have that
\bb
D_\delta \ket{0} = \ket{\alpha_\delta}\, ,
\label{displacement_vacuum}
\ee
where $\alpha_\delta = \frac{1}{\sqrt2} \left( \delta_X + i \delta_P \right)$, $\delta_X \coloneqq \left( \delta_1,\delta_3,\ldots, \delta_{2n-1}\right)^\intercal$, and $\delta_P \coloneqq \left( \delta_2,\delta_4,\ldots, \delta_{2n}\right)^\intercal$.

A $2n\times 2n$ real matrix $S$ is called \deff{symplectic} if it preserves the quadratic form $\Omega$, i.e.\ if $S\Omega S^\intercal = \Omega$. 
%\label{symplectic}
Symplectic matrices form a group, called the symplectic group. To any symplectic matrix $S$ we can associate a unique unitary operator $U_S$ on the Hilbert space $\HH_n$, called the \deff{Gaussian unitary} corresponding to $S$, such that
\bb
U_S^\dag\, r\, U_S^{\phantom{\dag}} = Sr\, ,
\label{U_S}
\ee
where the vector of operators on the right-hand side is intended to have components $(Sr)_j = \sum_k S_{jk} r_k$. From~\eqref{U_S} it is not difficult to verify that displacement operators transform according to the rule $U_S^{\phantom{\dag}} D_\delta U_S^\dag = D_{S\delta}$ under the action of Gaussian unitaries.

Crucially, the group of Gaussian unitaries is generated by the evolutions induced by quadratic Hamiltonians. Here, a Hamiltonian is called quadratic if it is of the form $H_\mathrm{q} = \frac12 r^\intercal Q r$ for some $2n\times 2n$ real symmetric matrix $Q$, and the associated evolution is the unitary operator $e^{-i H_\mathrm{q}}$. (Without loss of generality, we incorporated the time in the Hamiltonian.) It turns out that
\bb
e^{-\frac{i}{2} r^\intercal Q r} = U_S\, ,\qquad S = e^{\Omega Q}\, .
\label{quadratic<-->symplectic}
\ee
Also, the correspondence $S\mapsto U_S$ is a group isomorphism between symplectic matrices and Gaussian unitaries.

A special sub-group within the symplectic group is formed by \emph{orthogonal} symplectic matrices $K$, which in addition to $K\Omega K^\intercal = \Omega$ also satisfy the orthogonality condition $KK^\intercal = \id$. A remarkable property of the corresponding Gaussian unitaries, also called in this case `passive unitaries', is that they send coherent states to coherent states. Specifically, if $K$ is orthogonal symplectic then
\bb
U_K \ket{\alpha} = \ket{K\alpha} \qquad \forall\ \alpha\in \C\, .
\label{passive_on_coherent}
\ee

Williamson's theorem states that any real positive definite $2n\times 2n$ matrix $W>0$ can be brought into the normal form
\bb
S^{-1} W S^{-\intercal} = \bigoplus_{j=1}^n \begin{pmatrix} \nu_j & 0 \\ 0 & \nu_j \end{pmatrix}
\ee
via symplectic congruence. Here, \tcr{$S^{-\intercal}$ denotes the inverse transpose of $S$, and} the numbers $\nu_j >0$ are called the \deff{symplectic eigenvalues} of $W$. Since these depend only on $W$, we will often denote them by $\nu_j(W)$. Using the fact that $S^{-\intercal} \Omega = \Omega S$ for every symplectic $S$, it is not difficult to show that they can be computed as
\bb
\spec\left(W\Omega\right) = \left\{\pm i\,\nu_1(W),\ldots,\pm i\,\nu_n(W)\right\}\, ,
\label{symplectic_spectrum_eigenvalues}
\ee
where $\spec(X)$ is the spectrum of $X$. For a given $W>0$ with the above Williamson's decomposition and an arbitrary $\delta\in \R^{2n}$, we can define the corresponding \deff{Gaussian operator} by
\bb
&\mathrm{G}_1[W,\delta] \coloneqq D_{\!\delta} U_S \left(\bigotimes_{j=1}^n \sum_{k=0}^\infty \frac{2}{\nu_j\!+\!1} \left( \frac{\nu_j\!-\!1}{\nu_j\!+\!1} \right)^k \!\ketbra{k}_j \right) U_S^\dag D_{\!\delta}^\dag ,
\label{Gaussian_operator}
\ee
where $\ket{k}_j$ is the $k^\text{th}$ Fock state on the $j^\text{th}$ mode. It turns out that the right-hand side of~\eqref{Gaussian_operator} does not depend on the particular choice of Williamson's decomposition for $W$, but only on $W$ and $\delta$ themselves. Indeed, $W$ and $\delta$ can be recovered from $\mathrm{G}_1[W,\delta]$ via the formulae
\bb
\delta = \Tr \mathrm{G}_1[W\!,\delta]\, r ,\quad W = \Tr \mathrm{G}_1[W\!,\delta]\, \{r,r^\intercal\} - \delta \delta^\intercal\! ,
\label{trace_formulae_V_delta}
\ee
with $\{X,Y\}\coloneqq XY +YX$ denoting the anti-commutator. Note that $\Tr \mathrm{G}_1[W,\delta] = 1$; moreover, $\mathrm{G}_1[W,\delta]\geq 0$ if and only if $\nu_j\geq 1$ for all $j=1,\ldots, n$, which in turn happens if and only if $W+i\Omega\geq 0$. When this is the case, $\mathrm{G}_1[W,\delta]$ is in fact a density operator, and is called a \deff{Gaussian state}. We will refer to $W$ as its \deff{covariance matrix} and to $\delta$ as its \deff{mean} (or \deff{displacement}) \deff{vector}. A Gaussian state is called \deff{centred} if its mean vanishes. Using~\eqref{U_S}, it is not difficult to verify that for all Gaussian unitaries $U_S$ it holds that
\bb
U_S^{\phantom{\dag}} \mathrm{G}_1[W,\delta] U_S^\dag = \mathrm{G}_1\big[SWS^\intercal, S\delta\big] .
\label{QCM_transformation_symplectic}
\ee
Another formula that we will find useful is that for the maximal eigenvalue of a Gaussian operator, which, incidentally, coincides with its operator norm~\eqref{operator_norm}. Indeed, one sees %immediately 
from~\eqref{Gaussian_operator} that
\bb
\lambda_{\max}\left(\mathrm{G}_1[W,\delta]\right) = \left\| \mathrm{G}_1[W,\delta] \right\|_\infty = \prod_{j=1}^n \frac{2}{\nu_j(W)+1}\, .
\label{operator_norm_QCM}
\ee
See also~\cite[Eq.~(346)]{LL-Renyi} for a related statement. \tcr{To deduce~\eqref{operator_norm_QCM} from~\eqref{Gaussian_operator}, it suffices to observe that the spectrum is invariant under unitary action, so that
\bb
\spec\left(\mathrm{G}_1[W,\delta]\right) &= \spec\left(\bigotimes_{j=1}^n \sum_{k=0}^\infty \frac{2}{\nu_j\!+\!1} \left( \frac{\nu_j\!-\!1}{\nu_j\!+\!1} \right)^k \!\ketbra{k}_j \right) \\
&= \left\{ \prod_{j=1}^n \frac{2}{\nu_j\!+\!1} \left( \frac{\nu_j\!-\!1}{\nu_j\!+\!1} \right)^{k_j}:\ k_1,\ldots,k_n \in \N \right\} ,
\label{spectrum_Gaussian_calculation}
\ee
where we used the shorthand $\nu_j = \nu_j(W)$. Now, due to the fact that $\left|\frac{\nu-1}{\nu+1}\right|\leq 1$ for all $\nu\geq 0$, it is clear that the maximum eigenvalue of $\mathrm{G}_1[W,\delta]$ coincides with the eigenvalue of maximum \emph{modulus}, and hence also with the operator norm of $\mathrm{G}_1[W,\delta]$ --- which, being self-adjoint, is such that $\left\|\mathrm{G}_1[W,\delta]\right\|_\infty = \max_{\lambda\in \spec(\mathrm{G}_1[W,\delta])} |\lambda|$. To calculate the right-hand side, it suffices to set $k_1 = \ldots = k_n=0$ in~\eqref{spectrum_Gaussian_calculation}, which yields precisely~\eqref{operator_norm_QCM}.}

\subsection{Separable states}

Assume that the quantum system of interest $A$ is multi-partite, which we will write $A=A_1\ldots A_n$; concretely, this means that $\HH_A$ admits the tensor product structure $\HH_A = \HH_{A_1}\otimes \HH_{A_n}$. Then we can identify a special class of states on $A$ called the (fully) \deff{separable states}; a state $\rho_A$ is separable --- or, more pedantically, separable with respect to the partition $A_1:A_2:\ldots: A_n$ --- if it admits the decomposition
\bb
\rho_A = \int \!\dd\mu(\psi_{A_1}, \ldots, \psi_{A_n})\ \psi_{A_1} \otimes \ldots \otimes \psi_{A_n}
\label{separable}
\ee
for some Borel probability measure $\mu$ on the product of the sets of local (normalised) pure states.

It is often of interest to determine whether a given state is separable or not. %However, this problem turns out to be computationally intractable in general --- more precisely, it is NP-hard for an inverse polynomial approximation error~\cite{GurvitsNPhard, GharibianNPhard}. 
To address this issue, quantum information scientists have developed a number of criteria that may help to reach a decision efficiently for certain states. The \emph{positive partial transposition} (PPT) criterion for either finite-dimensional~\cite{PeresPPT, HorodeckiPPT, Horodecki-review} or continuous-variable systems~\cite{Simon00, Duan2000, revisited, BUCCO} is the most important for the purposes of this work. %The main criterion of this sort is the \emph{positive partial transposition} (PPT) criterion.
%~\cite{Horodecki-review, PeresPPT, HorodeckiPPT, Horodecki-PPT-entangled, Horodecki1999, GurvitsBarnum, Resh1, Resh5, Resh6, PPT-high-SN} or belonging to continuous-variable systems~\cite{Simon00, Duan2000, Werner01, Giedke01, introeisert, revisited}. The main criterion of this sort is the \emph{positive partial transposition} (PPT) criterion.

To describe it we first need to first introduce the notion of partial transposition. Given a multi-partite quantum system $A = A_1\ldots A_n$ and some $J\subseteq [n]\coloneqq \{1,\ldots, n\}$, the partial transposition is a map $\Gamma_{\!J}: \HS(\HH_A) \to \HS(\HH_{A})$ on the space of Hilbert--Schmidt operators on $\HH_A$ whose action $\Gamma_{\!J}(X) = X^{\Gamma_{\!J}}$ is defined by
\bb
\left( \bigotimes_{j\in J} X_{A_j} \otimes \bigotimes_{k\notin J} X_{A_k} \right)^{\Gamma_{\!J}} \coloneqq \bigotimes_{j\in J} X_{A_j}^\intercal \otimes \bigotimes_{k\notin J} X_{A_k}
\label{Gamma_J}
\ee
on product operators, and extended by linearity and continuity to the whole $\HS(\HH_A)$. It is easy to see that $\Gamma_{\!J}$ is an involution, i.e.\ $\big(X^{\Gamma_J}\big)^{\Gamma_J} = X$, and that moreover it is an isometry on the Hilbert space $\HS(\HH_A)$ equipped with the Hilbert--Schmidt scalar product, in formula
\bb
\Tr \left[ XY \right] = \Tr \left[ X^{\Gamma_{\!J}} Y^{\Gamma_{\!J}} \right] \qquad \forall\ X,Y\, .
\label{PT_isometry}
\ee
In practice, for us this will mean that when one takes the partial transpose of a trace class operator (e.g.\ a state) one is guaranteed to end up with a Hilbert--Schmidt and thus bounded operator, but \emph{not} with a trace class operator.

In~\eqref{Gamma_J}, the transposition $^\intercal$ is taken with respect to fixed orthonormal bases of the local spaces $\HH_{A_i}$, for $i=1,\ldots, n$; although the definition of $\Gamma_{\!J}$ does depend on it, this choice of bases turns out to be immaterial in applications. We will therefore assume that an orthonormal basis has been chosen on each $\HH_{A_i}$, but we will not specify it unless it is necessary for computations. The PPT criterion is then easily stated as follows.

\begin{lemma}[(PPT criterion~\cite{PeresPPT})] \label{PPT_criterion_lemma}
Every separable state $\rho_A$ on $A=A_1\ldots A_n$ satisfies that $\rho_A^{\Gamma_{\!J}}\geq 0$ for all $J\subseteq [n]$.
\end{lemma}

The proof of this fact leverages~\eqref{separable} together with the observation that the transposition acts as complex conjugation on pure states, i.e.\ $\psi^\intercal = \psi^*$.
%\bb
%\psi^\intercal = \psi^*\, .
%\label{PT_pure}
%\ee

\begin{rem} \label{nontrivial_J_rem}
Clearly, the above criterion yields a non-trivial condition only when $1\leq |J|\leq n-1$, i.e.\ when $J$ is neither empty nor the full set $[n]$. Moreover, the conditions $\rho_A^{\Gamma_{\!J}}\geq 0$ and $\rho_A^{\Gamma_{\!J^c}}\geq 0$, where $J^c = [n]\setminus J$ is the complementary set to $J$ within $[n]$, are equivalent. Therefore, in Lemma~\ref{PPT_criterion_lemma} we can assume without loss of generality that $1\leq |J| \leq \floor{n/2}$.
\end{rem}

When $A=A_1\ldots A_n$ is an $n$-mode continuous-variable system, a natural choice for a local orthonormal basis is the Fock basis. Since with respect to this basis $x^\intercal = x$ and $p^\intercal = -p$, the corresponding partial transposition can be interpreted as a `partial time reversal'. It will be of interest for us to be able to compute the partial transpose of any $n$-mode Gaussian operator. Using~\eqref{trace_formulae_V_delta} and the above observation, it is not difficult to see that
\bb
\mathrm{G}_1[W,\delta]^{\Gamma_J} =&\ \mathrm{G}_1\big[ \Sigma_J W \Sigma_J, \Sigma_J \delta \big] , \\
\Sigma_J \coloneqq&\ \id_{A_{J^c}} \oplus \bigoplus_{j\in J} \begin{pmatrix} 1 & \\ & -1 \end{pmatrix} .
\label{partial_transpose_Gaussian}
\ee

\subsection{LOCC channels}

In this paper we aim to study dynamical processes involving the gravitational interaction. Therefore, we need to translate the above theory from the state to the channel setting. \tcr{A channel $A\to A'$ is a linear, completely positive, and trace preserving (CPTP) map that %links operators in the Hilbert space of $A$ to operators in the Hilbert space of $A'$.
takes as input trace class operators acting on the Hilbert space of $A$ and outputs operators acting on the Hilbert space of $A'$.} If both $A=A_1\ldots A_n$ and $A'=A'_1\ldots A'_n$ are multi-partite, then certain channels $A\to A'$ can be realised by means of local (quantum) operations on each of the parties $1,\ldots, n$ assisted by classical communication among those parties (LOCC)~\cite{LOCC}. The set of such channels will be denoted by $\locc$, or $\locc(A\to A')$ if there is need to specify the input and output systems.

In some sense, LOCC channels generalise separable states. Therefore, it is not surprising that Lemma~\ref{PPT_criterion_lemma} can be extended to provide a necessary criterion for a channel to be LOCC. To do this, we need a systematic way to connect states and channels; this can be done thanks to the Choi--Jamio\l kowski isomorphism~\cite{Choi, Jamiolkowski72}. We will now describe it briefly assuming for simplicity that the input system $A$ has finite dimension $d \coloneqq \dim \HH_A = \prod_{i=1}^n \dim \HH_{A_i}$.

Given a channel $\Lambda:A\to A'$, let us consider a copy $A''\simeq A$ of the system $A$, and let us construct the \deff{maximally entangled state} on $A A''$ as
\bb
\ket{\Phi_d^{AA''}}\coloneqq \frac{1}{\sqrt{d}} \sum_{\ell=1}^d \ket{\ell}_A\ket{\ell}_{A''}\, .
\label{max_ent}
\ee
Up to an irrelevant normalisation factor, the \deff{Choi--Jamio\l kowski operator} associated with $\Lambda$ is defined by
\bb
D_\Lambda^{AA'} \coloneqq d \left( I_{A} \otimes \Lambda_{A''\to A'} \right) \left(\Phi_{d}^{AA''} \right) .
\label{dynamical_matrix}
\ee
Now, the following generalises Lemma~\ref{PPT_criterion_lemma} to the channel setting.

\begin{lemma}[{\cite{Rains1997, LOCC}}] \label{LOCCs_are_PPT_lemma}
Let $A=A_1\ldots A_n$ and $A'=A'_1\ldots A'_n$ be multi-partite systems, and let $\Lambda:A\to A'$ be an LOCC. Then the Choi--Jamio\l kowski operator $D_\Lambda^{AA'}$ defined by~\eqref{dynamical_matrix} is separable with respect to the partition $A_1A'_1:\ldots :A_n A'_n$, and in particular
\bb
\left( D_\Lambda^{AA'} \right)^{\Gamma_{\!J}}\geq 0 \qquad \forall\ J\subseteq [n]\, ,
\ee
where the partial transposition $\Gamma_{\!J}$ is over all systems $A_j, A'_j$, for $j\in J$.
\end{lemma}

As explained in Remark~\ref{nontrivial_J_rem}, also here we can assume without loss of generality that $1\leq |J|\leq \floor{n/2}$.

\subsection{Conditional min-entropy} \label{subsec_conditional_min_entropy}

We now introduce and discuss two quantities known as max-relative entropy and conditional min-entropy.

\begin{Def}[{\cite{Datta08}}]
Let $X,Y$ be bounded self-adjoint operators, not necessarily positive semi-definite. Their \deff{max-relative entropy} is given by $D_{\max}(X\|Y) \coloneqq \log d_{\max}(X\|Y)$, where
\bb
d_{\max}(X\|Y) &\coloneqq \inf\left\{\lambda > 0:\, X\leq \lambda Y\right\} ,
\label{d_max}
\ee
with the convention that $\log 0 = -\infty$ and $\log \inf \emptyset = +\infty$.
\end{Def}

A simple yet useful observation is the following: if $X$ and $Y>0$ are such that $Y^{-1/2} X Y^{-1/2}$ is a bounded operator, then
\bb
d_{\max}(X\|Y) &= \lambda_{\max}\left(Y^{-1/2} X Y^{-1/2} \right) ,
\label{d_max_lambda_max}
\ee
where $\lambda_{\max}(Z)$ denotes the maximal eigenvalue of $Z$.

\begin{Def}[{\cite[Definition~6.2]{TOMAMICHEL}}] \label{conditional_min_entropy_def}
Let $L_{AB}$ be a (not necessarily positive semi-definite) bounded operator on a bipartite quantum system $AB$. Define its \deff{conditional min-entropy} as
\bb
H_{\min}(A|B)_L \coloneqq&\ - \inf_{\xi_B} \left\{ - D_{\max}\left( L_{AB}\, \big\|\, \id_A\otimes \xi_B \right) \right\} \\
=&\ - \log \inf_{\xi_B} d_{\max}\left(L_{AB}\,\big\|\, \id_A\otimes \xi_B\right) ,
\label{min_cond_entropy}
\ee
where the optimisation runs over all normalised quantum states $\xi_B$ on $B$, and once again we convene that $-\log 0 = +\infty$ and $-\log \inf \emptyset = -\infty$.
\end{Def}

The conditional min-entropy of an arbitrary operator cannot in general be represented in closed form. However, the optimisation in~\eqref{min_cond_entropy} is a semi-definite program; as such, for finite-dimensional systems it can be solved efficiently, i.e.\ in time polynomial in the dimension~\cite{vandenberghe_1996}.

\section{General dynamical experiments to test the quantum nature of an interaction} \label{sec_dynamical_experiments}

\subsection{Unitary simulation via LOCCs} \label{subsec_unitary_simulation_LOCCs}

The quantum information-theoretic task that underpins our proposal for detecting the quantum nature %quantumness 
of an interaction is that of simulation of a unitary by means of LOCCs. We can define it in rigorous terms by using the language of \deff{quantum hypothesis testing}. For a general introduction to this subject, we refer the reader to the excellent textbooks by Hayashi~\cite[Chapter~3]{HAYASHI} and Tomamichel~\cite[\S~7.1]{TOMAMICHEL}. Let $A = A_1\ldots A_n$ and $A'=A'_1\ldots A'_n$ be two $n$-partite quantum systems. To define the task, we need two ingredients. The first is a source that outputs random pure states of $A$ drawn from the ensemble
\bb
\mathcal{E} = \{p_\alpha, \psi_\alpha\}_\alpha\, ,
\label{ensemble}
\ee
where $\psi_\alpha=\ketbra{\psi_\alpha}$, and $p_{\alpha}$ is the probability for it to be drawn. This means that the source produces the state $\psi_\alpha$ with probability $p_\alpha$. The index $\alpha$ could be discrete or continuous, in which case $p$ will become a probability measure over some measure space. The second ingredient is an isometry
\bb
U: \HH_A\to \HH_{A'}
\label{isometry}
\ee
that connects $A$ with $A'$. Both $\mathcal{E}$ and $U$ are publicly known.

At the beginning of each round, a pure state is drawn from $\mathcal{E}$ and then handed over to an agent $G$, who knows the ensemble but not its particular realisation. $G$ takes one of the following two actions on the system:
\begin{itemize}
    \item[(NH)] \deff{Null hypothesis:} $G$ applies $U$.
    \item[(AH)] \deff{Alternative hypothesis}: $G$ attempts to simulate $U$ by implementing a suitably chosen LOCC channel $\Lambda\in \locc(A\to A')$.
\end{itemize}
These two options are depicted in Figure~\ref{simulation_fig} below. 

After $G$ has carried out one of these two actions, the output system $A'$ is sent to a verifier $V$, who knows $\alpha$ but not the strategy adopted by $G$. The goal of $V$ is to decide between the null hypothesis (NH), corresponding to the pure quantum state $\psi'_\alpha\coloneqq U\psi_\alpha U^\dag$, and the alternative hypothesis (AH), corresponding to a possibly mixed state $\Lambda(\psi_\alpha)$. Remember that $V$ knows $\psi_\alpha$ and $U$, and hence $\psi'_\alpha$, but \emph{not} $\Lambda$. In attempting to decide between (NH) and (AH), $V$ can make one of two distinct errors:
\begin{itemize}
\item[(1)] the \deff{type-1 error} consists in guessing (AH) while (NH) holds;
\item[(2)] conversely, the \deff{type-2 error} consists in guessing (NH) while (AH) is true.
\end{itemize}
These two errors are not necessarily equally consequential, an observation on which the study of asymmetric hypothesis testing is founded~\cite{Hiai1991, Ogawa2000, Brandao2010, berta_composite, gap}. 
%~\cite{Hiai1991, Ogawa2000, Brandao2010, Tomamichel2013, Li2014, brandao_adversarial, berta_composite, gap}. 
In our setting, the goal of $G$ is to maximise the probability that $V$ incurs a type-2 error, that is to say, the probability that $G$'s attempt to simulate $U$ via LOCCs goes undetected.

In our case, we are particularly interested in a simple test that $V$ can carry out. While this is not necessarily optimal in the most general sense, we will see that in many important cases it is experimentally realisable with conceivable technology. Knowing both $\psi_\alpha$ and $U$, and thus being able to calculate $\psi'_\alpha$, $V$ can carry out the quantum measurement represented by the POVM $(\psi'_\alpha, \id-\psi'_\alpha)$ on the unknown state. If the outcome corresponding to $\psi'_\alpha$ is obtained, $V$ guesses (NH), otherwise (AH). In this way, assuming that the measurement is ideal the probability of a type-1 error is
\bb
P_1 = \sum_\alpha p_\alpha \Tr \left[ \psi'_\alpha \left(\id-\psi'_\alpha\right) \right] = 0\, ,
\label{P-I}
\ee
while the corresponding type-2 error probability evaluates to
\bb
P_2 &= \sum_\alpha p_\alpha \Tr \left[ \Lambda(\psi_\alpha) \psi'_\alpha \right] .
\label{P-II}
\ee
Knowing that $V$ will carry out the above test, $G$ now attempts to maximise $P_2$. We call the maximal value of $P_2$ that is obtainable by $G$ the \deff{LOCC simulation fidelity} (or the \deff{classical simulation fidelity}) of the isometry $U$ on the ensemble $\mathcal{E}$. It is given by 
\bb
F_{\cls}(\mathcal{E}, U) \coloneqq \sup_{\Lambda\in \locc(A\to A')} \sum_\alpha p_\alpha \Tr \left[ \Lambda(\psi_\alpha) \psi'_\alpha \right] .
\label{Fs}
\ee
As the name suggests, this is nothing but the maximal average fidelity (cf.~\eqref{fidelity}--\eqref{fidelity_overlap}) between the target state $\psi'_\alpha$ and its simulation $\Lambda(\psi_\alpha)$. The external optimisation in~\eqref{Fs} correspond to $G$ choosing the best possible strategy, i.e.\ the LOCC that minimises the average probability of the simulation being detected, among those allowed. Since $G$ does not know $\alpha$, the selected LOCC should work well simultaneously for most of the $\alpha$'s.

An elementary yet helpful observation is the following. Since local unitaries, i.e.\ unitaries of the form $V_{A} = \bigotimes_j V_{A_j}$ and $V'_{A'} = \bigotimes_j V'_{A'_j}$, are by construction LOCCs, and the set of LOCCs is closed by composition, the simulation fidelity cannot change by pre- or post-processing via local unitaries. Thus:

\begin{lemma} \label{local_unitaries_lemma}
Let $U:\HH_A \to \HH_{A'}$ be an isometry, and let $V_{A} = \bigotimes_j V_{A_j}$, $V'_{A'} = \bigotimes_j V'_{A'_j}$ be a pair of local unitaries. For every ensemble $\mathcal{E}=\{p_\alpha, \psi_\alpha\}_\alpha$ of states on $A$, it holds that
\bb
F_{\cls}\left(\mathcal{E},\, {V'}_{\hspace{-.8ex}A'}\, U\, V_A \right) = F_{\cls}\left(\mathcal{E},U \right) .
\ee
\end{lemma}

\begin{figure}[h!t] \centering
\begin{subfigure}{.45\textwidth}
\includegraphics[scale=1]{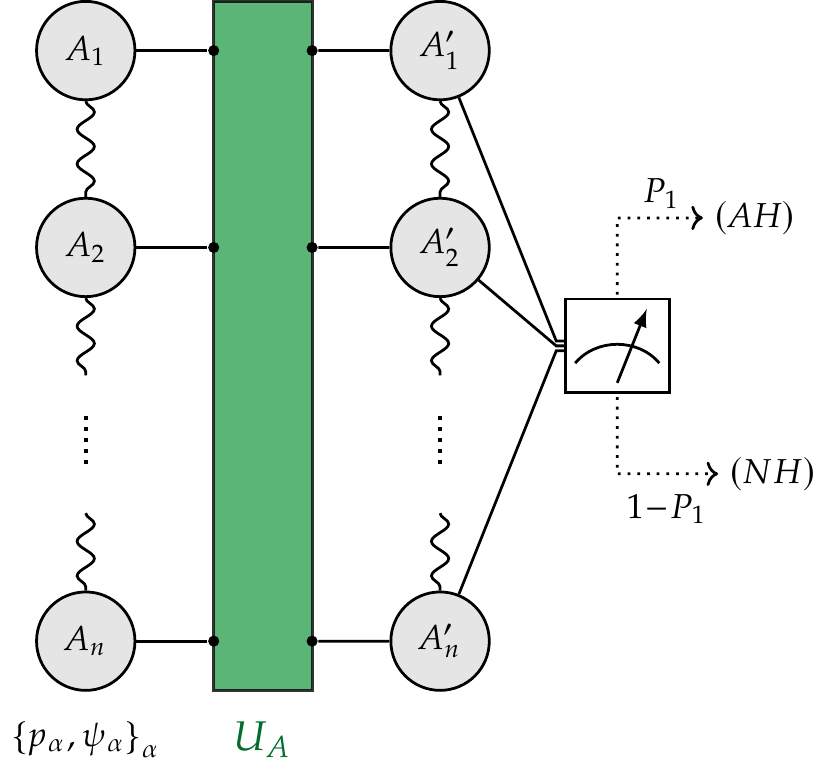}
\caption{}
\end{subfigure}
\par\bigskip
\begin{subfigure}{.45\textwidth}
\includegraphics[scale=1]{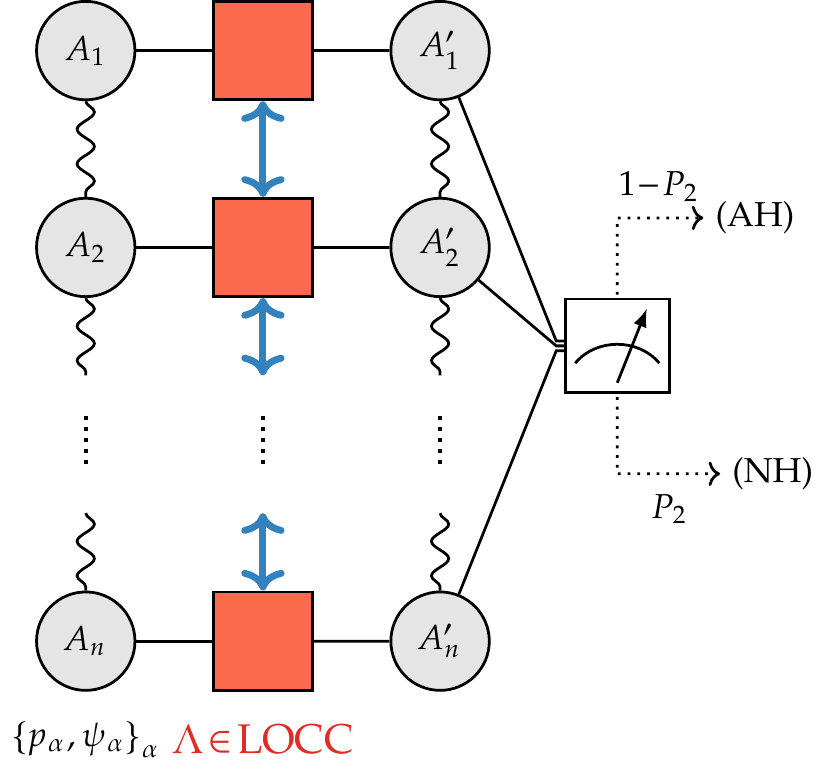}
\caption{}
\end{subfigure}
\caption{(a)~Quantum systems $A_1,\ldots, A_n$, initially prepared in a random pure state $\{p_\alpha, \psi_\alpha\}_\alpha$, evolve with a coherent global isometry $U$. (b)~The same systems undergo an evolution modelled by an LOCC. In both cases, at the end the value of $\alpha$ is revealed, and the system is subjected to the binary measurement $\left(\psi'_\alpha, \id-\psi'_\alpha\right)$, where $\ket{\psi'_\alpha}\coloneqq U \ket{\psi_\alpha}$.}
\label{simulation_fig}
\end{figure}

In \S~\ref{sec_dynamical_experiments} below, we will apply the above scheme to the case where $A_1,\ldots, A_n$ are massive quantum systems and $U$ is a unitary involving some local terms and an interaction of purely gravitational nature. The LOCC simulation of $U$ will then correspond to the dynamics induced by a classical --- and local --- gravitational field. However, it is worth remarking once again that the above scheme can actually test the quantum nature %quantumness 
of an arbitrary interaction between quantum systems. Finally, a note of caution: here and in the rest of the paper the term 'quantum nature', %`quantumness', 
referred to a dynamical process undergone by a multi-partite system, is to be understood as indicating the `non-LOCCness' of the channel modelling the process.

\subsection{A class of experiments to look for the quantum nature %quantumness 
of gravity} \label{subsec_experiments}

We consider massive quantum systems $A_1,\ldots, A_n$, subjected to local Hamiltonians $H_1,\ldots, H_n$ and interacting exclusively via gravity, as modelled by some Hamiltonian $H_G$. The total Hamiltonian can thus be written as $H_{\mathrm{tot}} = \sum_{j=1}^n H_j + H_G$; accordingly, the time evolution operator associated to some time interval $t$ will be represented by the unitary
\bb
U = \exp\left[ -\frac{it}{\hbar}\, H_{\mathrm{tot}} \right] = \exp\left[ -\frac{it}{\hbar}\left( \sum_{j=1}^n H_j + H_G \right) \right]
\label{UA}
\ee
acting on $A$. Since we want to initialise the system $A=A_1\ldots A_n$ in a (random) pure state $\psi_\alpha$, ideally all such states should be easy to prepare.

Suppose that for a certain system $A$, unitary $U_A$ as in~\eqref{UA}, and ensemble $\{p_\alpha, \psi_\alpha\}_\alpha$, we have computed the LOCC simulation fidelity $F_{\cls}(\mathcal{E}, U)$ in~\eqref{Fs}. We now imagine to run the experiment depicted in Figure~\ref{simulation_fig} many times independently, drawing the initial pure states at random in an i.i.d.\ fashion. If gravity behaves as a quantum Hamiltonian, then in the case of ideal measurements we will always get the outcome (NH) corresponding to a null hypothesis guess. On the contrary, if gravity behaves as a local classical field then we will get the outcome (AH) corresponding to an LOCC simulation guess with frequency at least $1-F_{\cls}(\mathcal{E}, U)$. That is to say, and this is the key point, \emph{if we obtain (AH) with frequency lower than $1-F_{\cls}(\mathcal{E}, U)$, we can conclude that gravity did not behave as a local classical field.}

\subsection{General LOCC inequality} \label{subsec_general_benchmarks}

The above reasoning should convince the reader that our role as theoreticians is to compute the number $F_{\cls}(\mathcal{E}, U)$ as accurately as possible. Unfortunately, this is in general a mathematically intractable task, because the expression in~\eqref{Fs} involves an optimisation over the set of LOCCs, which is notoriously hard to characterise~\cite{LOCC}. On second thought, however, one sees that an exact computation is not needed; even an upper bound such as $F_{\cls}(\mathcal{E}, U) \leq f$, where $f\in [0,1)$, would suffice for our purposes, because it would allow us to state that obtaining (AH) with frequency lower than $1-f$ (and hence of $1-F_{\cls}(\mathcal{E}, U)$) is direct evidence of the quantum nature %quantumness 
(i.e.\ non-LOCCness) of the gravitational interaction. 
%In light of this reasoning, we shall refer to any such upper bound as a \emph{quantum benchmark} for the task of unitary simulation by LOCCs. Our first contribution is a general method to compute quantum benchmarks for any ensemble of pure states $\{p_\alpha, \psi_\alpha\}_\alpha$ and any isometry $U$.
In this sense, the status of the inequality $F_{\cls}(\mathcal{E}, U) \leq f$ is in some ways analogous to that of a Bell inequality in the theory of non-locality~\cite{Brunner-review}. As an experimental violation of a Bell inequality provides a definitive proof that the underlying process producing the correlations is non-local, similarly the experimental violation of an inequality of the form $F_{\cls}(\mathcal{E}, U) \leq f$ would prove that the dynamics undergone by the system does not conform to \emph{any} LOCC-model. %is not an LOCC. 
In light of this reasoning, we shall refer to any such inequality as an \emph{LOCC inequality}. With this terminology, we can now say that our first contribution is a general LOCC inequality whose right-hand side is efficiently computable for any ensemble of pure states $\{p_\alpha, \psi_\alpha\}_\alpha$ and any isometry $U$.

%Our method builds on a technique developed in the context of quantum information science to approximate the set of LOCCs from the outside. 

\begin{thm}[(General LOCC inequality)] \label{general_bound_thm}
Let $\mathcal{E} = \left\{ p_\alpha, \psi_\alpha\right\}_\alpha$ be an ensemble of pure states $\psi_\alpha=\ketbra{\psi_\alpha}$ on $A=A_1\ldots A_n$. For an isometry $U:\HH_A \to \HH_{A'}$, where $A' = A'_1\ldots A'_n$, set $\ket{\psi'_\alpha}_{A'}\coloneqq U \ket{\psi_\alpha}_{A}$. Then the associated simulation fidelity satisfies
\bb
F_{\cls}(\mathcal{E}, U) &\leq \min_{J\subseteq [n]} \exp\left[ -H_{\min}(A'|A)_{R^{\Gamma_{\!J}}} \right] \\
&= \min_{J\subseteq [n]} \inf_{\xi_A} d_{\max} \left( R^{\Gamma_{\!J}}_{AA'} \,\Big\|\, \xi_A \otimes \id_{A'} \right) \\
&= \min_{J\subseteq [n]} \inf\left\{ \kappa:\, R^{\Gamma_{\!J}}_{AA'} \leq \kappa\, \xi_A \otimes \id_{A'} \right\}
\label{general_bound}
\ee
where $[n]= \{1,\ldots, n\}$, the state $R_{AA'}$ is defined by
\bb
R_{AA'} \coloneqq \sum_\alpha p_\alpha\, (\psi_\alpha^*)_{A} \otimes (\psi'_\alpha)_{A'}\, , \label{R_gamma}
\ee
the symbol $\Gamma_{\!J}$ denotes the partial transposition over all subsystems $A_j, A'_j$, for $j\in J$, the optimisations in~\eqref{general_bound} are over $\kappa\geq 0$ and states $\xi_A$, and the conditional min-entropy $H_{\min}(A'|A)_{R^{\Gamma_{\!J}}}$ is defined as in~\eqref{min_cond_entropy}.

If $\ket{\psi_\alpha} \coloneqq \bigotimes_{j=1}^n \ket{\psi_{\alpha,j}}_{A_j}$ is a product state for all $\alpha$, we can also state the simplified bound
\bb
F_{\cls}(\mathcal{E},U)
&\leq \min_{J\subseteq [n]} d_{\max}\left( R_{AA'}^{\Gamma_{\!J}}\, \Big\|\, R_A^{\Gamma_{\!J}} \otimes \id_{A'} \right) , 
\label{handier_general_bound}
\ee
where $R_A\coloneqq \Tr_{A'} R_{AA'}$, and $d_{\max}$ is defined by~\eqref{d_max}. With the above hypotheses, the right-hand side of~\eqref{handier_general_bound} is at most $1$. In both~\eqref{general_bound} and~\eqref{handier_general_bound}, we can assume without loss of generality that $1\leq |J|\leq \floor{n/2}$.
\end{thm}

Before we prove the above result, it is worth discussing %its importance
why it is useful. Inequality~\eqref{general_bound} constitutes a genuine LOCC inequality, in the form of 
%quantum benchmark in the sense specified in \S~\ref{subsec_experiments}, i.e.\ it is 
a general upper bound on the LOCC simulation fidelity~\eqref{Fs}. Although the right-hand side of~\eqref{general_bound}, which depends essentially on the conditional min-entropy, is not a closed-form expression but rather the result of an optimisation itself (cf.~\eqref{min_cond_entropy}), such an optimisation, being a semi-definite program, is efficiently solvable in time polynomial in the dimension (see the discussion in \S~\ref{subsec_conditional_min_entropy}). This is to be compared with the optimisation defining $F_{\cls}$, which involves the mathematically intractable set of LOCCs and thus cannot be tackled effectively from a computational standpoint. That being said, since to design experiments to test the quantum nature %quantumness 
of gravity we will need to apply the above result to infinite-dimensional systems composed of quantum harmonic oscillators, the more explicit upper bound~\eqref{handier_general_bound} will turn out to be quite helpful.

Another important point is that the right-hand side of~\eqref{general_bound} is a continuous function of the initial states and of the unitary evolution $U$, because the conditional min-entropy is itself continuous, as established by~\cite[Lemma~21]{Tomamichel2010}. This means that small experimental uncertainties in the states or the unitary will not invalidate our tests.

On a different note, the setting of the above Theorem~\ref{general_bound_thm} subsumes that of classical simulation of teleportation. In that context, our bound~\eqref{handier_general_bound} reproduces an important benchmark for quantum experiments known as the \deff{classical threshold}, whose computation in a variety of contexts has been the subject of intensive investigation~\cite{Massar1995, Werner1998, Horodecki-teleportation, Bruss1999, Hammerer2005, Adesso2008, Owari2008, jensen2011quantum, Chiribella2014, Calsamiglia2009, Yang2014}. In Appendix~\ref{subsec_telep_threshold} we explain how our techniques can be used to immediately recover and generalise many of these results, with the added benefit of more transparent proofs.

\begin{ex} \label{swap_ex}
Let us pause for a moment and look at a simple yet very instructive application of Theorem~\ref{general_bound_thm}. Consider a bipartite system $A=A_1A_2$, with both $A_1$ and $A_2$ having dimension $d$. Let the \deff{swap operator} $F_{A_1A_2}$ act on it as
\bb
F_{A_1A_2} \ket{\psi}_{A_1}\ket{\phi}_{A_2} \coloneqq \ket{\phi}_{A_1}\ket{\psi}_{A_2}\, .
\label{swap}
\ee
(In this case, $A' \simeq A$ is simply a copy of $A$.) We want to simulate the action of the swap on the ensemble $\EE_{\mathrm{H} \times \mathrm{H},\,d}\coloneqq \left\{ \psi\otimes \phi \right\}_{\psi,\phi}$, where $\psi = \ketbra{\psi}$ and $\phi = \ketbra{\phi}$ are independent, Haar-distributed local pure states. As we discussed in the Introduction, the swap operation comes up naturally when considering a simple model of two gravitationally interacting quantum harmonic oscillators. At the same time, our intuition suggests that its LOCC simulation fidelity should be bounded away from~$1$ as the swap operation can create entanglement in the presence of local ancillae~\cite{eisert2000optimal}. Now, due to Theorem~\ref{general_bound_thm} we can make this intuition quantitative. To do so, it suffices to compute the state $R_{AA'}$ in~\eqref{R_gamma} as
\bb
R_{AA'} &= \int \!\!\dd\psi\! \int \!\!\dd \phi\ \psi_{A_1}^* \otimes \phi_{A_2}^* \otimes \phi_{A'_1} \otimes \psi_{A'_2}\, ,
\ee
where $\dd\psi$ and $\dd\phi$ denote the local Haar measures. Setting $J=\{2\}$ and noting that $\psi^\intercal = \psi^*$ %according to~\eqref{PT_pure}, 
is simply the complex conjugate vector, this yields
\bb
R_{AA'}^{\Gamma_2} &= \left( \int \!\!\dd\psi\ \psi_{A_1} \otimes \psi_{A'_2} \right)^* \otimes \left(\int\!\! \dd \phi\ \phi_{A_2} \otimes \phi_{A'_1} \right) \\
&= \frac{2 S_{A_1A'_2}}{d(d+1)} \otimes \frac{2 S_{A_2A'_1}}{d(d+1)}\, ,
\label{R_AA'_Gamma_swap}
\ee
where $S \coloneqq (\id + F)/2$ is the projector onto the symmetric subspace. Eq.~\eqref{R_AA'_Gamma_swap} can be easily proved by using the techniques of~\cite{Werner}. One first notes that $\sigma \coloneqq \int \!\dd\psi\ \psi \otimes \psi$ is invariant under conjugation by any operator of the form $U\otimes U$, where $U$ is a generic $d\times d$ unitary. States with this property are precisely the linear combinations of the projectors onto the symmetric and anti-symmetric subspaces. Since $\sigma$ is furthermore supported onto the symmetric subspace, we conclude that it must be proportional to the corresponding projector.

Observing that $R_A^{\Gamma_2} = R_A = \id/d^2$ is the maximally mixed state and remembering~\eqref{d_max_lambda_max}, the bound~\eqref{handier_general_bound} yields immediately
\bb
F_{\cls} \left(\EE_{\mathrm{H} \times \mathrm{H},\,d},\, F \right) &\leq d^2 \lambda_{\max}\left( \frac{2 S_{A_1A'_2}}{d(d+1)} \otimes \frac{2 S_{A_2A'_1}}{d(d+1)} \right) \\
&= \left(\frac{2}{d+1}\right)^2 .
\label{swap_bound}
\ee
This bound is strictly larger than that in~\cite{Siddiqui2020, Siddiqui2022} (equal to $1/d^2$, see e.g.~\cite[Proposition~4]{Siddiqui2022}) because the setting considered there requires simulation over all possible input states, while we only look at an ensemble of product states. We will see later that~\eqref{swap_bound} is in fact tight, i.e.\ the LOCC simulation fidelity in this case turns out to be exactly equal to $4/(d+1)^2$: for the optimal strategy, $A_1$ and $A_2$ measure locally with the uniform POVM $(\psi)_\psi$, where $\psi$ ranges over all pure states, and send a classical description of the outcome to the other party, who reconstructs the corresponding state locally.

To interpret the above result~\eqref{swap_bound} consider the case of two qubits, $d=2$. The ensemble $\EE_{\mathrm{H}\times \mathrm{H},\, 2}$ can then be described geometrically in simple terms: it is constructed by drawing at random, and independently, two pure states distributed uniformly on the surface of the Bloch sphere. In this simple setting, the above bound implies that the best fidelity with which an LOCC can simulate a swap operation applied to such a random product state is $4/9\approx 0.44$. As our intuition suggested, the LOCC simulation fidelity is bounded away from $1$. Remember from the discussion around~\eqref{even_more_elementary_evolution} that the swap operation can indeed appear in a simple model of gravitational interaction between two quantum harmonic oscillators. In this context, and under our assumptions to be detailed later, any experiment of the type described in \S~\ref{subsec_experiments} that detects an LOCC simulation (AH) fewer than $5/9 \approx 56\%$ of the times could be considered as indicating the quantum nature %quantumness 
of gravity.
\end{ex}

We are now ready to provide a complete mathematical proof of Theorem~\ref{general_bound_thm}. The reader who is rather interested in the applications of this general result to gravitationally interacting systems can jump to \S~\ref{sec_proposed_implementation}.

\begin{proof}[Proof of Theorem~\ref{general_bound_thm}]
For the sake of simplicity, we consider here the case where $A$ is finite dimensional and the ensemble is discrete. The general case where both of these assumptions are lifted is deferred to Appendix~\ref{infinite_dim_app}. Fix some $J\subseteq [n]$, and pick $\kappa\geq 0$ and a state $\xi_A$ such that
\bb
R_{AA'}^{\Gamma_J} \leq \kappa\, \xi_A\otimes \id_{A'}\, .
\label{general_bound_proof_aux}
\ee
Consider that
\begin{align}
&F_{\cls}(\mathcal{E}, U) \nonumber \\
&\quad = \sup_{\Lambda \in \locc} \sum_\alpha p_\alpha \braket{\psi'_\alpha| \Lambda(\psi_\alpha) |\psi'_\alpha} \nonumber \\
&\quad\eqt{1} \sup_{\Lambda \in \locc} \sum_\alpha p_\alpha \braket{(\psi_\alpha^*)_{A} (\psi'_\alpha)_{A'} \big| D_\Lambda^{AA'} \big| (\psi_\alpha^*)_{A} (\psi'_\alpha)_{A'}} \nonumber \\
&\quad\eqt{2} \sup_{\Lambda \in \locc} \Tr \left[ R_{AA'} D_\Lambda^{AA'} \right] \nonumber \\
&\quad\eqt{3} \sup_{\Lambda \in \locc} \Tr \left[ R_{AA'}^{\Gamma_{\!J}} \left(D_\Lambda^{AA'}\right)^{\Gamma_{\!J}} \right] \label{general_bound_proof} \\
&\quad\leqt{4} \sup_{\Lambda \in \locc} \kappa \Tr \left[ \xi_A\otimes \id_{A'} \left(D_\Lambda^{AA'}\right)^{\Gamma_{\!J}} \right] \nonumber \\
&\quad= \sup_{\Lambda \in \locc} \kappa \Tr_A \left[ \xi_A\, \Tr_{A'}\left(D_\Lambda^{AA'}\right)^{\Gamma_{\!J}} \right] \nonumber \\
&\quad\eqt{5} \sup_{\Lambda \in \locc} \kappa \Tr_A \left[ \xi_A\, \id_A \right] \nonumber \\
&\quad=\kappa\, . \nonumber \\
\end{align}
The above steps can be justified as follows. In~1 we introduced the un-normalised Choi--Jamio\l kowski state of $\Lambda$, defined by~\eqref{dynamical_matrix}, and used the formula
\bb
M\otimes \id \ket{\Phi_d} = \id\otimes M^\intercal \ket{\Phi_d}\, ,
\label{max_ent_trick}
\ee
valid for any $d\times d$ matrix $M$, to compute
\begin{align}
&\braket{\psi_\alpha^* \psi'_\alpha | D_\Lambda |\psi_\alpha^* \psi'_\alpha} \nonumber \\
&\qquad = d \Tr \left[\psi_\alpha^*\otimes \psi'_\alpha\, (I\otimes \Lambda)(\Phi_d) \right] \nonumber \\
&\qquad = d \Tr \left[\id \otimes \psi'_\alpha\, (I\otimes \Lambda)\left((\psi_\alpha^* \otimes \id)\, \Phi_d\right) \right] \nonumber \\
&\qquad = d \Tr \left[\id \otimes \psi'_\alpha\, (I\otimes \Lambda)\left((\id\otimes \psi_\alpha)\, \Phi_d\right) \right] \\
&\qquad = d \Tr \left[\psi'_\alpha\,\Lambda\left(\Tr_1\, [ (\id\otimes \psi_\alpha)\, \Phi_d ] \right) \right] \nonumber \\
&\qquad = \Tr \left[\psi'_\alpha\, \Lambda( \psi_\alpha) \right] \nonumber \\
&\qquad = \braket{\psi'_\alpha| \Lambda(\psi_\alpha)|\psi'_\alpha} . \nonumber
\end{align}
Note that here we omitted the subscripts identifying the systems for simplicity, and we introduced the notation $\Tr_1$ for the trace over the first tensor factor in the third-to-last line. Continuing with the justification of~\eqref{general_bound_proof}, in~2 we introduced the state $R_{AA'}$ defined by~\eqref{R_gamma}, while in~3 we applied the partial transposition over $J$; this leaves the trace unchanged, as per~\eqref{PT_isometry}. In~4 we applied the operator inequality~\eqref{general_bound_proof_aux}, which is made possible by the fact that $\left(D_\Lambda^{AA'}\right)^{\Gamma_{\!J}} \!\geq 0$ because of Lemma~\ref{LOCCs_are_PPT_lemma}. Finally, in~5 we observed that
\begin{align}
\Tr_{A'}\! \left[ \left(D_\Lambda^{AA'}\right)^{\Gamma_{\! J}} \right] &= \left( \Tr_{A'} D_\Lambda^{AA'} \right)^{\Gamma_{\! J}} \nonumber \\
&= d \left( \Tr_{A'} \left(I_A\otimes \Lambda_{A''\to A'}\right) \left(\Phi_d^{AA''}\right) \right)^{\Gamma_{\! J}} \nonumber \\
&= d \left( \Tr_{A''} \Phi_d^{AA''} \right)^{\Gamma_{\! J}} \\
&= d \left( \frac{\id_{A}}{d} \right)^{\Gamma_{\! J}} \nonumber \\
&= \id_A \nonumber
\end{align}
thanks to the fact that $\Lambda$ is trace preserving, and moreover $\Tr \xi_A =1$ by hypothesis. Thanks to the definition of conditional min-entropy~\eqref{min_cond_entropy}, taking the infimum over $J\subseteq [n]$, $\kappa$, and $\xi_A$ such that~\eqref{general_bound_proof_aux} is obeyed yields~\eqref{general_bound}. This concludes the proof of~\eqref{general_bound}. Additional technical details are deferred to Appendix~\ref{infinite_dim_app}.

To derive~\eqref{handier_general_bound}, it suffices to take as ansatz $\xi_A = R_A^{\Gamma_J}$ in the second line of~\eqref{general_bound}. This is possible if $\psi_\alpha = \bigotimes_j \psi_{\alpha,j}$ because in that case
\bb
R_A^{\Gamma_{\!J}} &= \Tr_{A'} R_{AA'}^{\Gamma_{\!J}} \\
&= \sum_\alpha p_\alpha (\psi_\alpha^*)_A^{\Gamma_{\!J}} \\
&= \sum_\alpha p_\alpha \bigotimes_{j\in J} \psi_{\alpha,j}\otimes \bigotimes_{k\in J^c} \psi_{\alpha,k}^* \\
&\geq 0
\ee
and $\Tr_A R_A^{\Gamma_{\!J}} = \Tr_A R_A = \Tr_{AA'} R_{AA'} = 1$, so that $R_A^{\Gamma_{\!J}}$ is indeed a density operator. The fact that the right-hand side of~\eqref{handier_general_bound} is at most $1$ amounts to the operator inequality
\bb
R_{AA'}^{\Gamma_{\!J}} \leq R_A^{\Gamma_{\!J}} \otimes \id_{A'}\, .
\label{H_min_positive}
\ee
This can be verified as follows:
\begin{enumerate}[(i)] \itemsep0em
\item $\psi_{A'}^{\Gamma_{\!J}} \leq \id_{A'}$ holds for all pure states $\ket{\psi}_{A'}\in \HH_{A'}$ and for all  $J\subseteq [n]$, as one can see e.g.\ by writing out the Schmidt decomposition of $\ket{\psi}_{A'}$ with respect to the cut $A'_J: A'_{J^c}$;
\item in particular, $(\psi_\alpha')_{A'}^{\Gamma_{\!J}} \leq \id_{A'}$ for all $\alpha$;
\item thanks to the fact that $(\psi^*_\alpha)_{A}^{\Gamma_{\!J}}\geq 0$ because $\ket{\psi_\alpha^*}$ is a product state (Lemma~\ref{PPT_criterion_lemma}), we deduce that
\bb
(\psi_\alpha^*)_{A}^{\Gamma_{\!J}} \otimes (\psi'_\alpha)_{A'}^{\Gamma_{\!J}} \leq (\psi^*_\alpha)_{A}^{\Gamma_{\!J}} \otimes \id_{A'}\, ,
\ee
which immediately implies~\eqref{H_min_positive} upon multiplying by $p_\alpha$ and summing over $\alpha$.
\end{enumerate}
Finally, the restrictions on $J$ follow from Remark~\ref{nontrivial_J_rem}.
\end{proof}

\begin{rem}
The above proof yields an upper bound not only on the LOCC simulation fidelity, but indeed on the a priori possibly larger \emph{PPT simulation fidelity}. This is given by a formula similar to~\eqref{Fs}, but where the supremum over LOCC operations is replaced by a supremum over all PPT quantum channels~\cite{Rains1999, Rains2001}, i.e.\ over all quantum channels whose Choi state has a positive partial transpose.
\end{rem}

\section{A proposed implementation} \label{sec_proposed_implementation}

We proceed with a discussion of an experimental concept in which the quantum nature %quantumness 
of the gravitational interaction could be tested in principle. Following a discussion of the required experimental parameter ranges, we then proceed to make the setting more concrete with a discussion of torsion pendula. We come to the conclusion that albeit being extremely challenging experimentally, such a set-up can conceivably satisfy all the constraints that have been identified here. The conceptual development of these experiments, their analysis  by means of Theorem~\ref{general_bound_thm} and the transfer to a realistic physical system represent one of our main contributions. In particular, on the mathematical side the computation of the right-hand side of~\eqref{handier_general_bound} will pose substantial technical challenges.

In a nutshell, the systems we consider are one-dimensional quantum harmonic oscillators that interact exclusively gravitationally, with otherwise arbitrary spatial geometry, and that are initialised in coherent states drawn from i.i.d.\ Gaussian ensembles. More in detail, consider a system of $n$ one-dimensional quantum harmonic oscillators $A_1,\ldots, A_n$ as in Figure~\ref{system_oscillators_fig}. The $j^\text{th}$ particle has mass $m_j$ and is bound to move on a straight line oriented in direction $\hat{n}_j$, where $\hat{n}_j\in \R^3$ is a unit vector. It is confined in the vicinity of a centre located at $\vec{R}^0_j$ by a harmonic potential with frequency $\omega_j$. The position of the $j^\text{th}$ particle is thus $\vec{R}_j = \vec{R}_j^0 + x_j \hat{n}_j$, so that the total Hamiltonian of the system is of the form
\bb
H_{\mathrm{tot}} &= \sum_j \left( \frac12 m_j \omega_j^2 x_j^2 + \frac{p_j^2}{2m_j} \right) - \sum_{j<k} \frac{G m_j m_k}{\left\|\vec{d}_{jk} - x_j \hat{n}_j + x_k \hat{n}_k\right\|}\, .
\label{system_oscillators_H_tot}
\ee
For what follows, it will be useful to define the angles $\theta_{jk}$, $\theta_{kj}$, and $\varphi_{jk}$ (see Figure~\ref{system_oscillators_fig}) by the relations 
\bb
\cos\theta_{jk} &\coloneqq \hat{n}_j \cdot \hat{d}_{jk}\, ,\\
- \cos\theta_{kj} &\coloneqq \hat{n}_k \cdot \hat{d}_{jk}\, , \\
\cos\varphi_{jk} &\coloneqq \hat{n}_j \cdot \hat{n}_k\, ,
\label{angles}
\ee
where $\hat{d}_{jk}\coloneqq \vec{d}_{jk} \big/ \big\|\vec{d}_{jk}\big\|$, and $\vec{d}_{jk} \coloneqq \vec{R}_k^0 - \vec{R}_j^0$.

\begin{figure}[ht]
\includegraphics[scale=1]{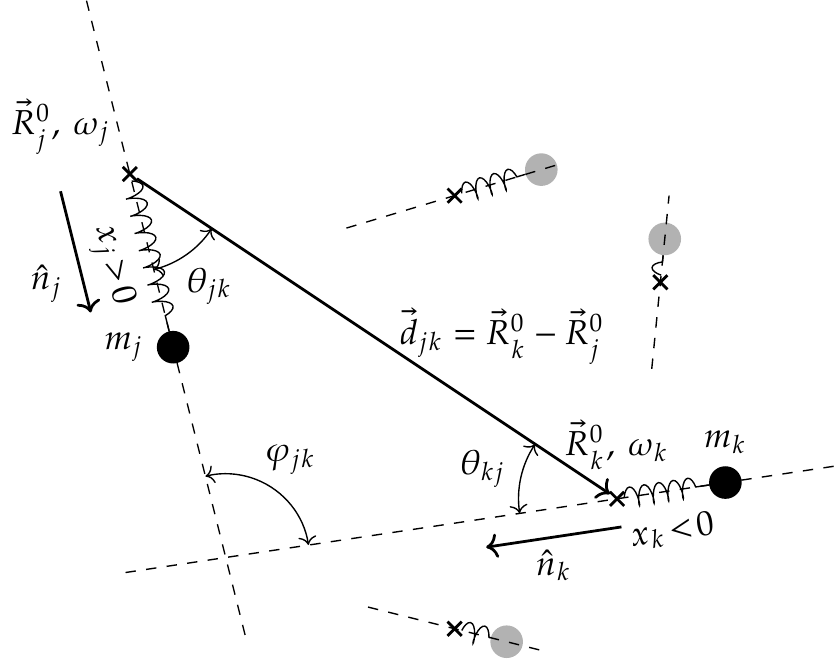}
\caption{A system of one-dimensional quantum harmonic oscillators. The various angles are defined in~\eqref{angles}. Note that for the most general three-dimensional arrangement the two dashed lines will not intersect.}
\label{system_oscillators_fig}
\end{figure}

\subsection{The concept} \label{subsec_proposal}

Our proposal for the setting considered in \S~\ref{subsec_experiments} is now completely specified as follows:
\begin{enumerate}[(1)]
\item The systems $A_j = A'_j$ are the one-dimensional quantum harmonic oscillators described above (Figure~\ref{system_oscillators_fig}), placed in space with arbitrary geometry.
\item The unitary $U_A$ to be simulated is simply the  time evolution operator $U_A = \exp \left[ - \frac{it}{\hbar}\, H_{\mathrm{tot}} \right]$ associated with the total Hamiltonian~\eqref{system_oscillators_H_tot}.
\item The oscillators are initialised in coherent states drawn independently from i.i.d.\ centred Gaussian ensembles with variance $1/\lambda$, where $\lambda>0$ is a fixed parameter. In other words, the ensemble $\mathcal{E}$ takes the form
\bb
\mathcal{E}_\lambda &\coloneqq \left\{ p_\lambda(\alpha),\, \ketbra{\alpha} \right\}_{\alpha\in \C^n}\, , \\
p_\lambda(\alpha) &\coloneqq \left(\frac{\lambda}{\pi}\right)^n e^{-\lambda \|\alpha\|^2}\, ,
\label{E_lambda}
\ee
where $\ket{\alpha}$ is the multi-mode coherent state defined by~\eqref{coherent}--\eqref{coherent_multimode}, and $\|\alpha\|^2\coloneqq \sum_{j=1}^n |\alpha_j|^2$.
\item Under assumptions (I)--(III) discussed below, the final states $\ket{\psi'_\alpha} \coloneqq U_A \ket{\alpha}_A$ turn out to be also coherent. Therefore, importantly, the final binary measurement $(\psi'_\alpha, \id-\psi'_\alpha)$ discussed in \S~\ref{subsec_unitary_simulation_LOCCs} \emph{can always be realised by local operations, namely, by (a)~applying local displacement operators on the oscillators; and then (b)~performing a phonon counting measurement $(\ketbra{0}, \id-\ketbra{0})$ on each oscillator locally.} The outcome ``$\psi'_\alpha$'' then corresponds to a zero-phonon detection on all modes.
\end{enumerate}

Let us assume for simplicity that all parameters $m_j, \omega_j, d_{jk}$ ($j\neq k$) are of the same order of magnitude $m,\omega, d$. We will analyse the above setting employing four assumptions, which we state now for clarity.

\begin{enumerate}[(I)]
\item The spatial amplitude of the local oscillations is much smaller than the distance between the oscillator centres. In formula,
\bb
\sqrt{\frac{n\hbar}{\lambda m \omega}} \ll \dmin\, .
\label{throw_away_3_order_approximation}
\ee
This assumption is needed to make the system mathematically tractable. We can use it to perform a Taylor expansion of the Hamiltonian~\eqref{system_oscillators_H_tot}.

\item Gravity yields the dominant contribution to the interaction Hamiltonian, which is well approximated by~\eqref{system_oscillators_H_tot}. This entails that the functional form of the gravitational potential should be the one predicted by Newton's law for distances of order $d$ and masses of order $m$~\cite{kapner2007tests}. Also, another relevant contribution to the interaction could come from Casimir forces. In the case of spheres with radii $R$ and under assumption~(I), these are negligible compared to gravity provided that~\cite{Pedernales2022}
\begin{align}
d &\gg R\, , \label{Casimir_small_1} \\
(d-2R)^6 &\gg \frac{207}{36\pi} \left(\frac{\vre-1}{\vre+2} \right)^2 \left(\frac{m_P}{m}\right)^2 R^6 ,
\label{Casimir_small_2}
\end{align}
where $\vre$ is the dielectric constant of the material of which the spheres are made, and $m_P$ is the Planck mass $m_P\coloneqq \sqrt{\frac{\hbar c}{G}} = \SI{2.18e-8}{\kg}$. The presence of other forces, e.g.\ due to static dipole moments of the test masses, may lead to even more stringent
constraints~\cite{Plenio2019} but as such dipole moments and hence the forces may, at least in principle, be removed, we choose to neglect 
them here.

\item There are at least two oscillators $A_j, A_k$ with \emph{almost identical} oscillation frequencies $\omega_j, \omega_k$, i.e.\ such that
\bb
|\omega_j - \omega_k| \ll \frac{G m}{\dmin^3 \omega}\, .
\label{almost_identical}
\ee
Moreover, if two oscillators $A_j, A_k$ do \emph{not} satisfy~\eqref{almost_identical}, then their frequencies satisfy
\bb
\frac{G m}{\dmin^3 \omega} \ll |\omega_j - \omega_k|\, .
\label{rotating_wave_assumption}
\ee
Finally, if \emph{all} frequencies are almost identical according to~\eqref{almost_identical}, we also require that
\bb
\frac{G m}{\dmin^3 \omega} \ll \min_j \omega_j\, .
\label{rotating_wave_assumption_n2}
\ee
\end{enumerate}

If~\eqref{almost_identical} is not obeyed for any pair $j,k$ with $j\neq k$, the rotating-wave approximation will wash out all effects due to the gravitational interaction, making its detection virtually impossible. Positing~\eqref{rotating_wave_assumption}, instead, is not strictly necessary to enable the mathematical analysis of this experiment, but we do it anyway because it simplifies the final expressions considerably. In fact, in what follows we will often formally assume that \emph{all} oscillators are characterised by almost identical frequencies. In this case we can rephrase~(III) as
\begin{enumerate}
\item[(III')] All frequencies $\omega_j$ are equal to $\omega$ up to a relative precision $\Delta\omega\coloneqq \max_{j\neq k} |\omega_j - \omega_k|$, and it holds that
\bb
\frac{\Delta\omega}{\omega} \ll \frac{G m}{\dmin^3 \omega^2} \ll 1\, .
\label{III'}
\ee
\end{enumerate}

\subsection{Quantitative analysis} \label{subsec_quantitative_analysis}

Due to assumption~(I), the interaction part of Hamiltonian~\eqref{system_oscillators_H_tot} can be Taylor expanded up to the second order. Since constant terms are physically irrelevant, they can be neglected right away. Terms that are linear in the canonical operators $x_j$, instead, do have a physical effect, but since they involve no interaction they do not affect the LOCC simulation fidelity anyhow. More formally, linear terms in the Hamiltonian result in local unitaries, and $F_{\cls}$ is invariant by local unitaries thanks to Lemma~\ref{local_unitaries_lemma}. As we argue in detail in Appendix~\ref{Taylor_app}, these considerations lead us to replace the original Hamiltonian~\eqref{system_oscillators_H_tot} with the effective Hamiltonian
\bb
H_{\mathrm{eff}} = \frac{\hbar}{2}\, \bar{r}^\intercal \left( \omega \id_{2n} + \widetilde{g} \right) \bar{r}\, ,
\label{effective_Hamiltonian}
\ee
where we employed the rearranged and dimensionless vector of canonical operators
\begin{align}
\bar{r} =&\ \left(\bar{x}_1,\ldots, \bar{x}_n,\bar{p}_1,\ldots, \bar{p}_n \right)^\intercal , \label{r_dimensionless} \\
\bar{x}_j \coloneqq&\ \sqrt{\frac{m_j \omega}{\hbar}}\,  x_j \, ,\quad \bar{p}_j \coloneqq \frac{1}{\sqrt{\hbar m_j \omega}} \, p_j
\label{dimensionless}
\end{align}
(note the different ordering with respect to~\eqref{r}), and we defined the matrix
\bb
\widetilde{g} \coloneqq&\ \begin{pmatrix} g & 0 \\ 0 & 0 \end{pmatrix} , \\[1ex]
g_{jk} =&\ \left\{ \begin{array}{ll} \displaystyle \sum_{\ell \neq j} \frac{G m_\ell}{d_{j\ell}^3 \omega} \left( 1-3\cos^2\theta_{j\ell}\right) & \ \text{$j=k$,} \\[4ex] \displaystyle -\frac{G \sqrt{m_j m_k}}{d_{jk}^3 \omega} \left( \cos\varphi_{jk} + 3 \cos\theta_{jk} \cos\theta_{kj} \right) & \ \text{$j\neq k$.} \end{array} \right.
\label{g}
\ee
It is worth remarking at this point that although the single-body terms do not contribute to the LOCC optimisation and thus do not change the LOCC simulation fidelity, they need to be known with precision, because they will influence the final state and hence the final measurement.

The Hamiltonian~\eqref{effective_Hamiltonian} is now purely quadratic. Hence, the time evolution operator resulting from letting~\eqref{effective_Hamiltonian} act for some time $t$ will be a Gaussian unitary. Thus, to make use of Theorem~\ref{general_bound_thm} we need to compute the right-hand side of~\eqref{handier_general_bound} for an a priori arbitrary Gaussian unitary. This calculation, which turns out to be highly non-trivial, is one of our main technical contributions.

\begin{thm}[(LOCC inequality for Gaussian unitaries)] \label{general_symplectic_bound_thm}
Let $S$ be an arbitrary $2n\times 2n$ symplectic matrix. Denote the corresponding Gaussian unitary over $n$ modes by $U_S$ (see~\eqref{U_S} for a definition). For $\lambda>0$, consider the Gaussian i.i.d.\ ensemble $\mathcal{E}_\lambda$ of $n$-mode coherent states with variance $1/\lambda$ defined by~\eqref{E_lambda}. Then the classical simulation fidelity defined by~\eqref{Fs} can be upper bounded by
\bb
F_{\cls}\left( \mathcal{E}_\lambda,\, U_S\right) \leq \min_{J\subseteq [n]} \frac{2^n(1+\lambda)^n}{\prod_{\ell=1}^{2n} \sqrt{2+\lambda+\left|z_\ell(\lambda,S,J)\right|}} ,
\label{general_symplectic_bound}
\ee
where $z_\ell(\lambda,S,J)$ is the $\ell^\text{th}$ eigenvalue of the Hermitian matrix
\bb
(1+\lambda) S^{-1} i\Omega_J S^{-\intercal} - i\Omega_J
\label{rotated_Omega_J}
\ee
and
\bb
\Omega_J &\coloneqq \bigoplus_{j\in J} \begin{pmatrix} 0 & 1 \\ -1 & 0 \end{pmatrix} \oplus \bigoplus_{k\in J^c} \begin{pmatrix} 0 & -1 \\ 1 & 0 \end{pmatrix}
\label{z_i_Omega_J}
\ee
with respect to a mode-wise decomposition. In~\eqref{general_symplectic_bound}, we can assume without loss of generality that $1\leq |J|\leq \floor{n/2}$.
\end{thm}

The proof of this theorem is quite involved, and will be presented in Appendix~\ref{subsec_proof_general_symplectic_bound}. However, what is of interest here is its application to the setting we are studying. This is captured by the following theorem.

\begin{thm}[(LOCC inequality for gravitationally interacting oscillators)] \label{recap_thm}
Consider a 3-dimensional array of $n$ one-dimensional quantum harmonic oscillators with identical frequencies $\omega$, characterised by masses $m_j$, angles $\theta_{jk}, \varphi_{jk}$, and distances $d_{jk}$, as defined in Figure~\ref{system_oscillators_fig}. Fix $\lambda>0$, and assume that the oscillators are initialised in coherent states drawn from the Gaussian i.i.d.\ ensemble $\mathcal{E}_\lambda$ with variance $1/\lambda$ defined by~\eqref{E_lambda}. Under assumptions~(I),~(II), and~(III') discussed above, the overall dynamics is approximately equivalent, up to local unitaries, to a Gaussian unitary $U_{S_{\mathrm{eff}}}$, where
\bb
S_{\mathrm{eff}} = \begin{pmatrix} \cos(gt/2) & \sin(gt/2) \\ -\sin(gt/2) & \cos(gt/2) \end{pmatrix}
\label{S_eff}
\ee
is an orthogonal symplectic matrix, and the symmetric matrix $g$ is defined in~\eqref{g}. The associated classical simulation fidelity~\eqref{Fs} can be upper bounded by
\bb
F_{\cls}\!\left( \mathcal{E}_{\!\lambda}, e^{-\frac{it}{\hbar}H_{\mathrm{tot}}} \right) &\approx F_{\cls}\!\left( \mathcal{E}_{\!\lambda}, U_{S_{\mathrm{eff}}} \right) \\
&\leq \min_{J\subseteq [n]} \frac{2^n(1+\lambda)^n}{\prod_{\ell=1}^{n} \left(2 + \lambda + \left|w_\ell ( \lambda, g t, J )\right|\right)}\, .
\label{f_simplified}
\ee
In~\eqref{f_simplified}, we denoted by $w_\ell\! \left( \lambda, g t, J \right)$ the $\ell^{\text{th}}$ eigenvalue of the Hermitian matrix
\bb
(1+\lambda)\, e^{i\frac{gt}{2}} \Xi_J e^{-i\frac{gt}{2}} - \Xi_J
\label{rotated_Xi_J}
\ee
and $\Xi_J$ is the $n\times n$ diagonal matrix with entries
\bb
(\Xi_J)_{jk} = \left\{ \begin{array}{rl} +1 & \quad\text{if $j=k\in J$,} \\[0.5ex] -1 & \quad\text{if $j=k\notin J$,} \\[0.5ex] 0 & \quad \text{if $j \neq k$.} \end{array} \right.
\label{Xi_J}
\ee
In~\eqref{f_simplified}, we can assume without loss of generality that $1\leq |J|\leq \floor{n/2}$.
\end{thm}

A complete mathematical proof of the above result can be found in Appendix~\ref{subsec_proof_recap_thm}.

One of the notable consequences of the above Theorem~\ref{recap_thm} is that the effective dynamics that the system undergoes can be approximated, up to local displacements, by a Gaussian unitary $U_{S_{\mathrm{eff}}}$, where $S_{\mathrm{eff}}$ is an orthogonal symplectic matrix. As we know from~\eqref{passive_on_coherent}, such a passive unitary maps coherent states to coherent states. (Note that the local displacement operators that result from the first-order terms we ignored in Appendix~\ref{Taylor_app} also share this property.) Since the system is initially prepared in a coherent state, it then follows that the global state of the system is at all times an (approximate) coherent state; in particular, due to~\eqref{coherent_multimode} \emph{it is approximately separable at all times.} Therefore, and rather remarkably, Theorem~\ref{recap_thm} allows us to detect the non-LOCCness of a dynamical process in which no entanglement is present at any point in time. The situation is somewhat reminiscent of what happens in quantum cryptography. 
The famous Bennett--Brassard `BB84' secret key distribution protocol~\cite{bennett1984quantum} has no entanglement at any time
but can be proven to be secure nevertheless. Indeed, at the root of the security proof of~\cite{shor2000simple} is the fact that
the security of BB84 is guaranteed for as long as the quantum channel that is being used has the capability in principle to establish
entanglement.

An important corollary of this discussion is that the final measurement, represented by the POVM $(\psi'_\alpha, \id-\psi'_\alpha) = (\ketbra{\alpha'},\, \id-\ketbra{\alpha'})$, where $\ket{\alpha'} = \ket{\alpha'_1}\otimes \ldots \otimes \ket{\alpha'_n}$ is the coherent state at the end of the protocol, can be implemented \emph{locally} by means of a relatively elementary procedure. It suffices to:
\begin{enumerate}[(i)]
    \item apply a local displacement operator $D(-\alpha') = D(-\alpha'_1)\otimes \ldots \otimes D(-\alpha'_n)$, which can be done by simply shifting the equilibrium position of each oscillator at the right time; and
    \item measure whether each oscillator contains $0$ vs at least $1$ phonons. This corresponds to carrying out the POVM $(\ketbra{0},\id-\ketbra{0})$. We then declare success (i.e.\ we guess the null hypothesis (NH)) if and only if all oscillators were found to have $0$ phonons.
\end{enumerate}
The fact that the final measurement is conceptually simple and that it can be implemented locally on each oscillator could simplify the experimental requirements considerably. Let us remark that this feature is not by construction; it is rather a somewhat unexpected consequence of the fact that the state of the system is approximately coherent at all times, as predicted by Theorem~\ref{recap_thm}.

The dependence of the right-hand side of~\eqref{f_simplified} on $\lambda$ and $t$ may seem a bit obscure. To find a simplified and somewhat more instructive expression, we can make use of two facts: (a)~the effect of the gravitational interaction integrated over time will likely be very weak in a practical experiment, thus we can expand the right-hand side of~\eqref{f_simplified} for small times; (b)~since $\lambda>0$ is ultimately a parameter of the experiment, we can try to optimise it in order to obtain maximum sensitivity. In fact, in all cases of practical interest it turns out that the optimal way to do so is to take $\lambda$ to be very small --- ideally, $\lambda\to 0$.
To see why, remember that we are trying to challenge a classical gravitational force to simulate (approximately) the correct dynamics on an arbitrary collection of coherent states drawn from i.i.d.\ Gaussian ensembles with variance $1/\lambda$. It is intuitively expected that the smaller $\lambda$, i.e.\ the greater the indeterminacy on the input coherent state, the more difficult it will be to carry out such a simulation. While we do not have full control over $F_{\cls}$, the same qualitative behaviour should be found in our upper bounds given by~\eqref{general_symplectic_bound} and~\eqref{f_simplified}. And indeed, this is seen to be the case in all practically interesting settings. Let us therefore formalise the following assumption.

\begin{enumerate}
\item[(IV)] The variance $1/\lambda$ of the Gaussian ensemble and the runtime of the experiment $t$ satisfy that
\bb
\max\left\{ \lambda,\, n \frac{Gmt}{\dmin^3 \omega} \right\} \ll 1\, .
\label{lambda_t_both_small}
\ee
\end{enumerate}
The origin of the second inequality can be understood by considering the special case $n=1$ in Eq.~\ref{elementary_evolution} where $\frac{G m t}{\dmin^3 \omega}=\pi$ yields the perfect swap and hence the most non-local channel possible. Waiting any longer does not add or is even detrimental to the discrimination of classical from quantum behaviour. 

The reader could wonder why we cannot take directly $\lambda\to 0$. Although this would not constitute a problem mathematically for the statements of Theorems~\ref{general_symplectic_bound_thm} or~\ref{recap_thm}, both of which make perfect sense and yield a non-trivial bound for $\lambda\to 0$, it would technically violate assumption~(I) (cf.~\eqref{throw_away_3_order_approximation}). Simply put, in that case the coherent state ensembles become so spread out that higher-order terms in the Taylor expansion of the gravitational Hamiltonian can no longer be neglected. Fortunately, the lower bound on $\lambda$ imposed by~\eqref{throw_away_3_order_approximation} is often many orders of magnitude below the other relevant parameter, i.e.\ the runtime measured in units of the time scale of the system, $\frac{Gmt}{\dmin^3 \omega}$. It is therefore possible to take $\lambda$ to be much smaller than $\frac{Gmt}{\dmin^3 \omega}$ while at the same time obeying~\eqref{throw_away_3_order_approximation}.

If this is true in theory, in practice having an ensemble with a large variance $1/\lambda$ may prove challenging. For this reason, we prefer to separate the more conservative assumption~(IV) above from the strengthened assumption~(IV') below, which requires more experimental control to be implemented but on the other hand yields a further simplification as well as a better scaling of the sensitivity with $t$.

\begin{enumerate}
\item[(IV')] The variance $1/\lambda$ of the Gaussian ensemble and the runtime of the experiment $t$ satisfy that
\bb
\frac{n \hbar}{m \dmin^2 \omega} \ll \lambda \ll \frac{Gmt}{\dmin^3 \omega} \ll \frac{1}{n}\, .
\label{lambda_effective_0_approximation}
\ee
(Note that the first inequality is a rephrasing of~\eqref{throw_away_3_order_approximation}, i.e.\ assumption~(I), and the latter is implied by~\eqref{lambda_t_both_small}, i.e.\ assumption (IV).) Note that~\eqref{lambda_effective_0_approximation} implies in particular that
\bb
t\gg \frac{n \hbar \dmin}{Gm^2}\, .
\label{lower_bound_t_hbar}
\ee
\end{enumerate}

If in addition to assumptions~(I),~(II), and~(III'), we take also either~(IV) or even~(IV'), we obtain a simplified version of Theorem~\ref{recap_thm} as follows.

\begin{thm}[(LOCC inequality, short-time expansion)] \label{recap_small_times_thm}
Under the same hypotheses of Theorem~\ref{recap_thm}, if in addition assumption~(IV) above is met, then the upper bound in~\eqref{f_simplified} can be expanded for small times as
\bb
F_{\cls}\left(\mathcal{E}_\lambda,\, U_{S_{\mathrm{eff}}} \right) &\leq 1 - \max_{\substack{J\subseteq [n],\\ |J|\leq n/2}} \sum_{\ell=1}^{|J|} \left( \sqrt{\lambda^2 + t^2 s_\ell^2(J)} - \lambda \right) + \Delta\, ,
\label{F_cls_expansion}
\ee
where
\bb
\Delta = O\left(t \|g\|_\infty \max\left\{ \lambda,\, t \|g\|_\infty \right\} \right) .
\label{Delta_big_O_estimate}
\ee
Here, $\|g\|_\infty$ is the operator norm~\eqref{operator_norm} of the matrix $g$ in~\eqref{g}, which only depends on the geometry of the system --- typically, $\|g\|_\infty \sim n \frac{Gmt}{d^3\omega}$. Also, $\{s_\ell(J)\}_{\ell=1}^{|J|}$ are the singular values of the $|J|\times (n-|J|)$ sub-block $g^{J,J^c}$ of $g$; in other words, labeling rows and columns of $g^{J,J^c}$ with $j\in J$ and $k\in J^c$, respectively, we have that
\bb
g^{J,J^c}_{jk} \coloneqq g_{jk} = -\frac{G \sqrt{m_j m_k}}{d_{jk}^3 \omega} \left( \cos\varphi_{jk} + 3 \cos\theta_{jk} \cos\theta_{kj} \right) .
\ee
If also assumption~(IV') is satisfied, then the expansion in~\eqref{F_cls_expansion} can be further simplified to
\bb
F_{\cls}\left(\mathcal{E}_\lambda,\, U_{S_{\mathrm{eff}}} \right) \leq&\ 1 - \eta t + O\!\left(\max\!\left\{n\lambda, \big(n \|g\|_\infty t \big)^2\right\}\! \right) ,
\label{F_cls_simple_expansion}
\ee
where the \deff{sensitivity} $\eta$ is given by
\bb
\eta \coloneqq&\ \max_{\substack{J\subseteq [n],\\ |J|\leq n/2}} \frac14 \left\| \left[g, \Xi_J\right] \right\|_1\, .
\label{sensitivity}
\ee
\end{thm}

The proof of Theorem~\ref{recap_small_times_thm} can be found in Appendix~\ref{subsec_proof_recap_small_times}. In the above statement, the singular values of a matrix $X$ are simply the eigenvalues of $\sqrt{X^\dag X}$. If $X=X^\dag$ is Hermitian --- more generally, normal --- then its singular values are simply the absolute values of its eigenvalues. The trace norm $\|X\|_1\coloneqq \Tr \sqrt{X^\dag X}$ (cf.~\eqref{trace_norm}) is the sum of the singular values; for Hermitian $X$, it coincides with the absolute sum of the eigenvalues $x_\ell$ of $X$, in formula $\|X\|_1 = \sum_\ell |x_\ell|$.

Importantly, the estimate~\eqref{F_cls_simple_expansion} reveals that under assumption~(IV') \emph{the classical simulation fidelity decreases linearly with time.} The figure of merit for the scheme is reduced to the simple quantity $\eta$, which depends only on the geometry of the problem. This very desirable feature is a consequence of the fact that we could effectively take $\lambda\to 0$ by virtue of the stronger assumption~(IV'). In light of these considerations, it makes sense to call the quantity $\eta$ defined by~\eqref{sensitivity} the `sensitivity' of a given configuration.
On a different note, we remark that if only~(IV) rather than~(IV') is met, then~\eqref{F_cls_expansion} reveals that our upper bound on $F_{\cls}$ decreases only quadratically with the runtime $t$, for very small $t$.

Finally, it is worth noticing that the remainder term $\Delta$ in~\eqref{F_cls_expansion} can be controlled mathematically. Therefore, the order-of-magnitude estimates in~\eqref{Delta_big_O_estimate} can be made fully rigorous. We do so in Appendix~\ref{subsec_proof_recap_small_times}, where we prove a rough estimate
\bb
\Delta \leq \frac{n}{2} \left(\lambda\! \left( e^{t\|g\|_\infty}\! - 1 \right) + e^{t\|g\|_\infty}\! - 1 - t\|g\|_\infty\!\right) + \frac{n^2}{8} t^2 \|g\|_\infty^2\, .
\label{Delta_UB}
\ee
Here, $\|g\|_\infty$ is the operator norm~\eqref{operator_norm} of the matrix $g$ in~\eqref{g}, which can be readily computed given the geometric configuration of the oscillators. Upper bounding $\|g\|_\infty$ can however also be done in full generality \emph{for all} such configurations. Doing so yields the explicit bound
\bb
\Delta \leq&\ \frac{n}{2} \left(\lambda\! \left( e^{s}\! - 1 \right) + e^{s}\! - 1 - s \right) + \frac{n^2 s^2}{8}\, , \\
s \coloneqq&\ \gamma t \min\left\{6(n-1),\ 288\ln(n-1) + 966 \right\} , \\
\gamma \coloneqq& \frac{G \left(\max_j m_j\right)}{\left(\min_{j\neq k} d_{jk}^3\right) \omega}\, ,
\label{Delta_UB_universal}
\ee
also proved in Appendix~\ref{subsec_proof_recap_small_times} thanks to the geometrical result of Appendix~\ref{operator_norm_g_app}. Observe that if $s$ is sufficiently small, as our assumptions guarantee, we have that $e^{s} - 1 \approx s$ and $e^{s}\! - 1 - s \approx s^2/2$, implying that
\bb
\Delta \approx \frac{n}{2} \left(\lambda s + \frac{s^2}{2}\right) + \frac{n^2 s^2}{8} = O\left( n s \max\{\lambda, n s\} \right) ,
\ee
which in turn entails~\eqref{Delta_big_O_estimate}.

\subsection{Examples} \label{subsec_examples}

\subsubsection{Two oscillators on a line}

As discussed in the Introduction, the simplest system to which we can apply the above results is that formed by two identical quantum harmonic oscillators with masses $m$ and frequencies $\omega$, oscillating along the same line and whose centres lie at a distance $d$ (Figure~\ref{2_line_fig} below).

\begin{figure}[ht]
\includegraphics[scale=1]{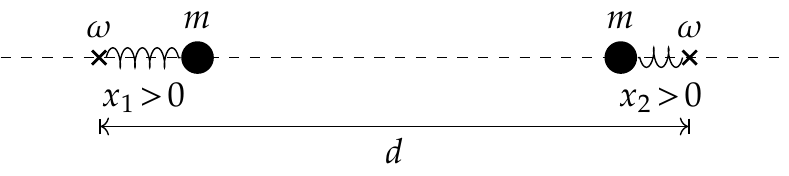}
\caption{Two harmonic oscillators aligned.}
\label{2_line_fig}
\end{figure}

In that case we have $n=2$, and the angles defined by~\eqref{angles} are $\theta_{12} = 0 = \varphi_{12}$ and $\theta_{21}=\pi$. Thus,~\eqref{g} becomes
\bb
g_{2,L} = 2\gamma \begin{pmatrix} -1 & 1 \\ 1 & -1 \end{pmatrix} ,
\ee
where
\bb
\gamma = \frac{Gm}{d^3\omega}
\label{gamma}
\ee
is the characteristic frequency here. The bound in~\eqref{f_simplified} can then be evaluated straightforwardly. First of all, as explained in Theorem~\ref{recap_thm} the only non-trivial choice for $J$ is $J = \{1\}$. In this case, the matrix in~\eqref{rotated_Xi_J} is just
\bb
&(1\!+\!\lambda)\, e^{\frac{it}{2} g_{2,L}}\, \Xi_J\, e^{-\frac{it}{2} g_{2,L}} - \Xi_J \\
&\quad = \left( (1\!+\!\lambda) \cos\!\left(2\gamma t\right)\! -\! 1\right) \sigma_{z} + (1\!+\!\lambda) \sin\!\left(2\gamma t\right) \sigma_{y} ,
\ee
where $\sigma_y$ and $\sigma_z$ denote the second and third Pauli matrices, respectively. Therefore,
\bb
w_{1} (\lambda\!, gt\!, J) &= - w_{2} (\lambda\!, gt\!, J) \\
&= \sqrt{\left( (1\!+\!\lambda) \cos(2\gamma t)\! -\! 1\right)^2 \!+ (1\!+\!\lambda)^2 \sin^2(2\gamma t)} \\
&= \sqrt{\lambda^2 + 4(1+\lambda)\sin^2(\gamma t)}\, ,
\ee
implying that the bound in~\eqref{f_simplified} in this case becomes
\bb
F_{\cls} \leq&\ f_{2,L}(\lambda, t) \coloneqq \frac{4(1\!+\!\lambda)^2}{\left(2\!+\!\lambda + \sqrt{\lambda^2 + 4(1\!+\!\lambda)\sin^2(\gamma t)}\right)^2}\, , 
\label{f_2L}
\ee
where $\gamma$ is defined by~\eqref{gamma}. Note that
\bb
f_{2,L}(0,t) = \frac{1}{\left(1+ \left| \sin(\gamma t)\right|\right)^2}\, ,
\label{f_2L_lambda=0}
\ee
which yields a dependence $f_{2,L}(0,t) = 1 - \eta_{2,L}t + O(t^2)$ with
\bb
\eta_{2,L} = \frac{\dd f_{2,L}(0,t)}{\dd t}\bigg|_{t=0} = 2\gamma
\label{eta_2L}
\ee
in~\eqref{F_cls_simple_expansion}. The function $f_{2,L}(\lambda, t)$ is plotted in Figure~\ref{f_2L_fig} for several values of $\lambda$, including $\lambda=0$.

\begin{figure}[ht]
\includegraphics[scale=1]{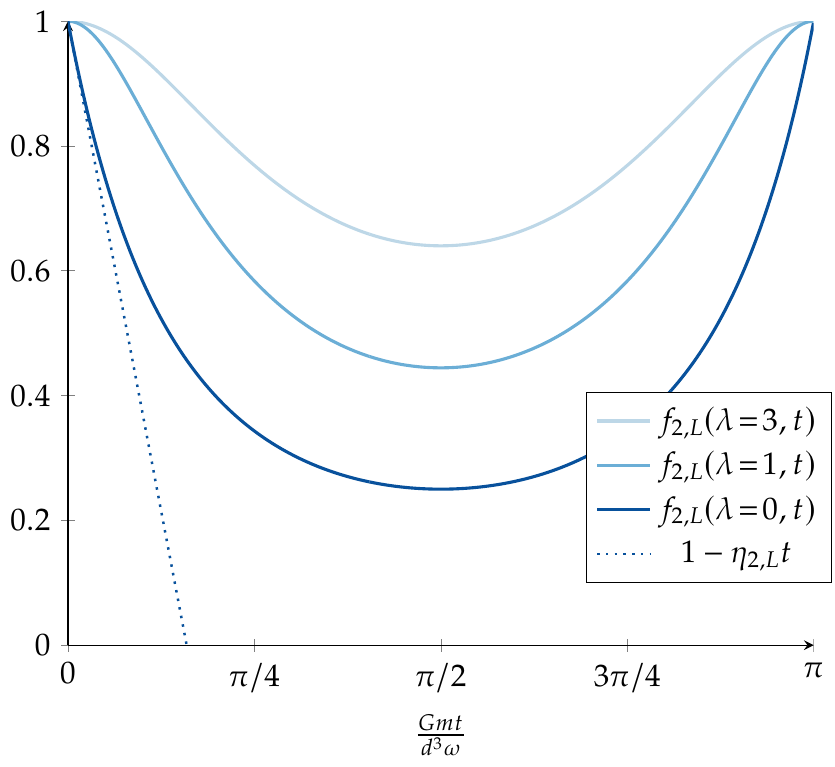}
\caption{The right-hand side of the LOCC inequality~\eqref{f_2L} plotted as a function of $\gamma t$ for several values of $\lambda$. Observe the approximately linear decrease with time of $f_{2,L}(0,t)$ for small times $t$, in line with the prediction of Theorem~\ref{recap_small_times_thm}.}
\label{f_2L_fig}
\end{figure}

\subsubsection{Many oscillators on a line} \label{subsubsec_many_oscillators}

A natural generalisation of the above scheme features $n$ identical oscillators, all with masses $m$ and frequencies $\omega$, aligned and equally spaced (Figure~\ref{4_line_fig}) at distance $d$. To simplify the analysis, it is useful to imagine that $n$ is even and to look at the short time limit only. 

\begin{figure}[ht]
\includegraphics[scale=1]{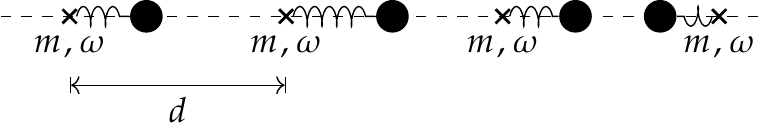}
\caption{Four identical oscillators aligned.}
\label{4_line_fig}
\end{figure}

A tight upper bound on $F_{\cls}$, i.e.\ a lower bound on $\eta$ (see~\eqref{sensitivity}), can be found by taking $J=\{1,3,\ldots, n-1\}$ as the set composed of all odd numbers between $1$ and $n-1$ (included). This yields
\bb
\eta_{n,L} \geq \frac14 \left\|[g,\Xi_J]\right\|_1 = \big\|K^{(n)}\big\|_1\, ,
\ee
where $K^{(n)}$ is an $(n/2)\times (n/2)$ real (but not symmetric) matrix with entries
\bb
K^{(n)}_{ij} = \frac{2\gamma}{\left|2(i-j)-1\right|^3}\, ,\qquad \forall\ i,j=1,\ldots, n/2\, ,
\ee
and $\gamma$ is given by~\eqref{gamma}. The trace norm of $K^{(n)}$ turns out to grow linearly in $n$, when $n$ is very large. Rather interestingly, the coefficient associated with this linear growth can be computed analytically. We have that
\bb
\big\|K^{(n)}\big\|_1 \sim n \gamma \zeta_L \qquad (n\to\infty)\, ,
\ee
where $\zeta_L$ is given by
\bb
\zeta_L \coloneqq \int_{-\pi}^{\pi}\frac{\dd \phi}{2\pi}\ \left|1 + e^{-i\phi/2} \chi_3(e^{-i\phi/2}) \right| \approx 1.267\, ,
\ee
where, for $|z|\leq 1$ and $s>1$,
\bb
\chi_\nu(z)\coloneqq \sum_{\ell=0}^\infty \frac{z^{2\ell+1}}{(2\ell+1)^\nu}
\ee
denotes the \deff{Legendre chi function}~\cite{MathWorld}. %~\cite[Chapter~6 and p.~283]{LEWIN}. 
These functions are related to more familiar objects such as the polylogarithm $\mathrm{Li}_\nu(z)\coloneqq \sum_{\ell=1}^\infty \frac{z^\ell}{\ell^\nu}$ by the functional equations ~\cite{MathWorld}% ~\cite{Cvijovic1995}
\bb
\chi_\nu(z) = \frac12 \left(\mathrm{Li}_\nu(z) - \mathrm{Li}_\nu(-z)\right) = \mathrm{Li}_\nu(z) - 2^{-\nu} \mathrm{Li}_\nu(z^2)\, .
\ee
The key observation that is needed to compute the trace norm of $K^{(n)}$ is that it is a Toeplitz matrix, i.e.\ a matrix whose entries are constant along each diagonal. 
A famous result stemming from the work of Szeg\H{o}~\cite{Szego1920} and formalised by Avram and Parter~\cite{Parter1986,Avram1988} (see also the textbook~\cite{GRUDSKY})
%, the lectures~\cite{Grudsky-lectures}, and the excellent pedagogical introduction~\cite{Gray2006}) 
allows to compute asymptotic averages of (functions of) the singular values of a large Toeplitz matrix by looking at the periodic function whose discrete Fourier coefficients generate the entries of the matrix. The details are left to the reader.

The above analysis shows that for large $n$ and very short times $t$ (see Theorem~\ref{recap_small_times_thm}) it holds that
\bb
F_{\cls} \lesssim 1 - n\zeta_L \gamma t \approx 1 - 1.267\cdot n\gamma t\, .
\label{F_cls_n_oscillators_line}
\ee
By virtue of~\eqref{eta_2L} we know that placing $n/2$ pairs of oscillators far apart from each other yields a dependence $F_{\cls} \lesssim 1 - n \gamma t$. Therefore, the single-line configuration explored here constitutes a genuine sensitivity improvement over a mere parallel repetition of the oscillator pair investigated above.

At this point, we could continue the list of examples by considering more general (e.g.\ genuinely 3-dimensional) configurations of oscillators. While this is certainly instructive, it is worth remarking that we have not been able to find any configuration that yields a faster increase of the sensitivity with $n$, (equivalently, a faster decrease of our upper bound on $F_{\cls}$), than that in~\eqref{F_cls_n_oscillators_line}. It therefore seems that the best configuration to run our experiment may be like that depicted in Figure~\ref{4_line_fig}, with a very large number $n\gg 1$ of oscillators.

\begin{rem}
    The sensitivity in~\eqref{F_cls_n_oscillators_line} grows linearly with $n$. Could another configuration of oscillators yield a better scaling in $n$, maybe a quadratic one? It is not too difficult to see that with our techniques it is not possible to obtain anything \emph{substantially} better than a linear scaling. Indeed, going back to~\eqref{sensitivity} and using Proposition~\ref{geometrical_prop} in Appendix~\ref{operator_norm_g_app}, one sees that
    \bb
        \eta = O\big(n \ln(n)\big)\, .
    \ee
    Therefore, already with this rather crude estimate we see that only a logarithmic improvement over a linear scaling is possible. As we mentioned, in all the families of examples we have examined the actual dependence of the sensitivity on $n$ seems to be at most linear.
\end{rem}

\subsection{Some experimental considerations} \label{subsec_experimental_considerations}

%\subsubsection{How realistic are our assumptions?} 

Here we would like to further discuss the above assumptions (I),~(II),~(III'), and~(IV). Starting from~(II), the first observation is that requiring that the Casimir effect is negligible compared to the gravitational interaction effectively determines a minimal scale for the system. Indeed, substituting in~\eqref{Casimir_small_2} the expression $m=\frac43 \pi R^3 \varrho$ for the mass of a sphere of density $\varrho$ yields immediately
\bb
(d-2R)^6 \gg \frac{207}{64\pi^3}\frac{m_P^2}{\varrho^2} \left( \frac{\vre-1}{\vre+2}\right)^2 .
\ee
Using the values for gold, %~\footnote{In spite of not being the densest element --- which is instead Osmium --- gold is much easier to manufacture.}, 
i.e.\ $\varrho = \SI{1.93e4}{\kg \per \metre^3}$ and $\vre = 6.9$, gives
\bb
(d-2R)^6 \gg \SI{0.58e-25}{\metre^6} .
\ee
The ratio between right-hand and left-hand side is $<0.1$ provided that
\bb
d-2R \gtrsim \SI{e-4}{\metre} = \SI{100}{\um}\, .
\ee
To fix ideas, let us take
\bb
d = \SI{125}{\um}\, ,\quad R = \SI{12.5}{\um}\, ,
\ee
which satisfies also~\eqref{Casimir_small_1}. We are thus experimenting with gold spheres of mass $m \approx \SI{1.58e-10}{\kg}$.

We now move on to~(III'). Note that
\bb
\frac{m}{d^3} \approx %10^{-3} \frac{m}{R^3} = \frac43 \pi 10^{-3} \varrho = \SI{94.7}{\kg \per \metre^3} .
\SI{80.8}{\kg \per \m^3} .
\ee
The second inequality in~\eqref{III'} then yields
\bb
\omega \gg \sqrt{\frac{Gm}{d^3}} \approx \SI{7.34e-5}{\Hz}\, .
\ee
If the above condition is satisfied we can look at the first inequality in~\eqref{III'}, which becomes
\bb
\label{eq:resonanceCond}
(\Delta \omega) \cdot \omega \ll \frac{Gm}{d^3} \approx \SI{5.39e-9}{\Hz^2} .
\ee
Therefore, the oscillators have to be manufactured with a precise control on the associated frequency. For the sake of continuing the discussion, let us assume that $\omega \sim \SI{e-3}{\Hz}$ and $\Delta\omega \sim \SI{e-7}{\Hz}$.

Taking $n=100$, $\eta\sim \frac{nGm}{d^3\omega}$, and $F_{\cls}\leq 0.9$ in~\eqref{F_cls_simple_expansion}, we infer that the runtime of a single repetition of the experiment is
\bb
t\sim (1-0.9) \frac{d^3 \omega}{nGm} \approx \SI{185}{\s}\, ,
\ee
or approximately $3$ minutes. This satisfies the inequality in~\eqref{lower_bound_t_hbar}, whose right-hand side is of the order of $\SI{8e-7}{\s}$. 

The only remaining inequality is~\eqref{lambda_effective_0_approximation}, which constrains the possible values of $\lambda$. With the above choices, observing that $\frac{n\hbar}{md^2\omega}\approx \SI{4.26e-12}{}$ we see that~\eqref{lambda_effective_0_approximation} would read
\bb
\SI{4.26e-12}{} \ll \lambda \ll 10^{-3} \ll 10^{-2}\, ,
\label{lambda_range}
\ee
which leaves an ample margin of choice for $\lambda$ so that the hypotheses of Theorem~\ref{recap_small_times_thm} are met. Let us remark in passing that $1/\lambda$ is the \emph{total} variance for all the $n$ oscillators combined. The variance per oscillator is instead $1/(n\lambda)$. Therefore, $\lambda \lesssim 10^{-3}$ is not as unrealistic as it may seem at first sight, as it requires to prepare single-mode coherent states of amplitude $|\alpha| \gtrsim \sqrt{10}$.

%\subsubsection{Possible noise sources}

Besides these considerations, any experiment like the one we are proposing will be extremely sensitive to noise. In Appendix~\ref{Noise} we provide an analysis of some of the main sources of noise and the isolation conditions that they impose. These include, among others: random fluctuations in mass taking place outside of the apparatus, due for instance to atmospheric turbulence or change in temperature; the effect of collisions between our test masses and some of the gas molecules filling the chamber where the experiment is taking place; decoherence induced due to black-body radiation; and the deleterious effect of electric and magnetic stray fields.

\subsection{Assessment of a concrete experimental setup}
\label{Assessment}

Taking into account all the considerations above, we now look into a potential experimental implementation of this concept, by considering a minimal arrangement with two oscillators. In order to guarantee that their interaction is resonant, Eq.~\eqref{eq:resonanceCond} tells us that their frequencies should differ by an amount much smaller than ${Gm/(d^3 \omega)}$, which in turn implies that the frequency of each individual resonator needs to be fixed to the same degree of precision. We are thus looking for high-Q, low-frequency oscillators, or more precisely, mechanical resonators that satisfy ${\omega^2 / Q \ll G m / d^3}$, where ${Q = \omega / \Delta \omega}$ is the quality factor. We identify torsion pendula as promising candidates for this %job
purpose. A torsion pendulum consists of a rigid body with a moment of inertia $I$ suspended from a thin wire that is attached to a fixed point, see Fig.~\ref{fig:torsion}a. Differential forces acting on the suspended object in a direction orthogonal to the wire induce rotations of the object in the horizontal plane, producing a torsion of the wire. The wire resists this deformation by applying a restoring torque $T$ that, for sufficiently small rotations, is directly proportional to the twisting angle $\theta$, that is,
\begin{equation}
T(\theta) = - \tau \theta,
\end{equation}
where $\tau$ is referred to as the torsion constant. Thus, the torsional degree of freedom behaves as a harmonic oscillator
\begin{equation}
\ddot \theta + 2 \gamma \dot \theta + \omega_I^2 \theta = 0,
\end{equation}
where the angular frequency is given by $\omega_I = \sqrt{\tau / I}$, and $\gamma$ is a damping constant that accounts for dissipation due to friction, which may arise both from the surrounding medium and from the internal wire structure. The torsion constant $\tau$ depends only on the geometry of the wire and is proportional to the %4th
fourth power of its diameter. Wires, made of different materials, can be manufactured with diameters down to the micrometer scale, resulting in torsion constants as low as $\tau \approx \SI{e-10}{\newton \m/\radian}$, corresponding to angular frequencies in the order of $\SI{}{\mHz}$ and Q factors as high as $10^5$~\cite{Schlamminger2017}. Thus, the operating parameter regime of the torsion pendula falls well within the requirements of Eq.~\eqref{eq:resonanceCond}. 

\begin{figure}[ht]
\begin{center}
   \includegraphics[width=\columnwidth]{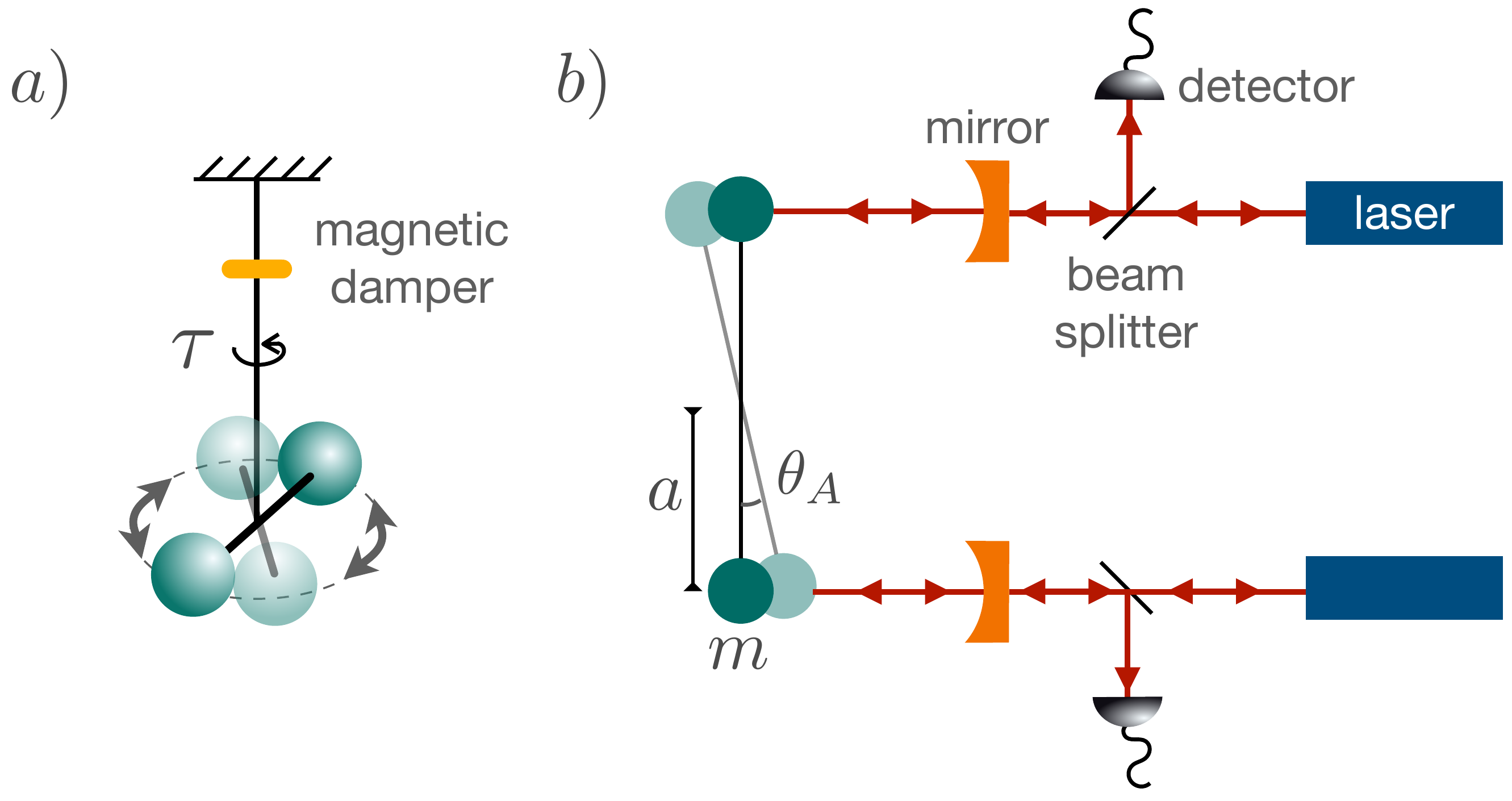} % requires the graphicx package
   \caption{{\bf Optomechanical torsion balance.} (a) shows a sketch of a torsion pendulum with a dumbbell-shaped body suspended from a thin wire with torsion constant $\tau$. A magnetic damper attached to the wire damps mechanical modes other than the torsional one~\cite{Luo2010}. (b) presents a scheme of the optomechanical system composed of the torsion pendulum and two optical cavities. The light output of each cavity is a measure of the position of the corresponding edge of the pendulum. Independent measurements of the position of each end of the pendulum can be added or subtracted to discriminate between purely rotational motion and pendular swinging. 
    }
   \label{fig:torsion}
   \end{center}
\end{figure}

In order to use torsion pendula in the envisioned protocol, we require to bring the torsional degree of freedom into the quantum regime and gain the ability to prepare and detect coherent states of motion. To the best of our knowledge, ground-state cooling of torsion pendula has not yet been demonstrated; however, we see no fundamental limitation that would prevent this possibility. One way to achieve this would be to couple the pendulum to an optical cavity and use the toolbox of optomechanical techniques~\cite{Marquardt2014, Liu2021}. %According to this strategy,
In this scheme, one of the ends of the pendulum would act as a movable mirror and together with a second, fixed mirror form an optical cavity. In this way, not only one but two cavities could be coupled to the mechanical resonator, one at each end of the torsional pendulum, see Fig.~\ref{fig:torsion}b. This would have the advantage that it would allow for differential measurement and control of the pendulum, and in this way increase the control precision, by allowing the rejection of common pendular oscillatory modes. As a matter of fact, the coupling of torsion pendula to cavities has been studied, both theoretically~\cite{Luo2009,Shao2015} and experimentally~\cite{Wang2008, Ando2020}. However, reaching the resolved sideband regime, where the mechanical frequency of the pendulum is larger than the linewidth of the cavity, is daunting with mechanical frequencies in the order of mHz, as desired for our protocol. Therefore, sideband cooling will most likely not be available in our desired arrangement, and one has to resort to some sort of feedback cooling~\cite{Pinard1999}. In Appendix~\ref{GSCooling} we analyze this possibility and conclude that ground-state cooling of the torsional degrees of freedom of a pendulum is, in principle, possible, provided that technical noise sources can be brought below the thermal noise threshold.

\subsubsection{State preparation}

Once our system has been cooled down to the ground state, we would like to generate a coherent state of motion. In general, preparing a coherent state starting from the ground state of an oscillator amounts to displacing either the oscillator or its equilibrium position by a distance proportional to the size of the target coherent state. This can be done in several ways. For example, a displacement of the equilibrium position can be achieved by applying a constant torque on the wire, while a force acting on the oscillator for a finite time can be used to displace it. Provided that our degree of control over the force or the torque is sufficiently precise, this could be used to prepare a coherent state. The force can be applied by an additional actuator, for example by the use of Coulomb interaction if the oscillator is charged, or, as for the feedback force described in Appendix~\ref{GSCooling} for cooling, by use of the optomechanical coupling. In the following, we heuristically describe one possible way of preparing a coherent state by use of the optomechanical coupling. 

The dynamics of the optomechanical system for each of the cavities is described by the Hamiltonian ($\hbar =1$)
\begin{equation}
H = \omega_{\rm cav} a^\dag a + \omega_I a^\dag_\theta a_\theta + g \left(a_\theta + a_\theta^\dag\right) a^\dag a,
\end{equation}
where $a$ and $a_\theta$ are the annihilation operators for the cavity field and the rotational degree of freedom of the pendulum, respectively. Here, $\omega_{\rm cav}$ is the frequency of the cavity field, and ${g = \omega_{\rm cav} a \sqrt{\hbar /(I \omega_I L^2)}}$ gives the strength of the optomechanical coupling. In the interaction picture, the Hamiltonian transforms as
\begin{equation}
H_I^{\rm int} (t) = g \left(a_\theta e^{- i \omega_I t} + a^\dag_\theta e^{i \omega_I t}\right) a^\dag a.
\end{equation}
For the parameters discussed earlier, we find ourselves in the regime ${g / \omega_I \approx 10^3 \gg 1}$. Thus, for times ${t \ll 1/ \omega_I \approx \SI{100}{\s}}$, we can neglect the time dependence of the Hamiltonian and find that the corresponding time-evolution operator amounts to a displacement acting on the torsional degree of freedom with a  parameter ${\alpha = gt a^\dag a}$ proportional to the number of photons in the cavity, 
\begin{equation}
D = e^{-i g a^\dag\! a t \left(a_\theta + a_\theta^\dag\right)}.
\end{equation}
Therefore, by injecting in one of the two cavities a Fock state $\ket{n}$, a coherent state $\ket{\alpha = g n t}$ can be prepared. Similarly, a displacement in the opposite direction  ${\alpha = - gnt}$ can be achieved by exciting the other cavity instead.

However, thus far we have ignored the decay of the cavity field, which can significantly constrain the accessible coherent states. For example, consider that a single photon state is injected into the cavity. This will remain inside the cavity only for a time $1/\kappa$, amounting to a total displacement of ${\alpha = g/\kappa \approx 10^{-6}}$, for the parameters discussed above. To circumvent this limitation, we envision a single photon source that injects photons into the cavity at the same rate $\kappa$ that the photons leave the cavity. In this way, we also gain the ability to control the start and end time of the displacement, by turning the single-photon source on and off, achieving a total displacement ${\alpha = g t_{\rm sing}}$, where $t_{\rm sing}$ is the time that the single photon source is on. With this method achieving coherent states with displacement parameter ${\alpha \approx 10}$, as required by Eq.~\eqref{lambda_range}, would take a time ${t = \alpha / g \approx 1}$ s. 

\subsubsection{Gravitational interaction}

With our torsion pendulum in the quantum regime and the ability to prepare an arbitrary coherent state, we now look into the prospects of coupling it to a second torsion pendulum through gravitational interaction. To that end, we place two torsion pendula with their equilibrium orientations standing parallel to each other and separated by a distance $d$, as depicted in Fig.~\ref{fig:interaction}. A Faraday shield is placed between the two pendula~\cite{Schmoele2016} in order to avoid photons incident on either pendulum scattering into the other one, in this way suppressing any potential optical interaction between the two oscillators~\cite{Zemanek2010}. The shield will also suppress undesired electric interactions, e.g.\ due to stray charges, Casimir force, or induced dipole moments~\cite{Schmoele2016}. %~\cite{Plenio2019}. 
In order to act also against undesired magnetic interactions arising, for example, from stray magnetic moments in the body of the pendula, the shield can be further improved by adding to it a layer made of a mu-metal.
\begin{figure}[h]
\begin{center}
   \includegraphics[width=\columnwidth]{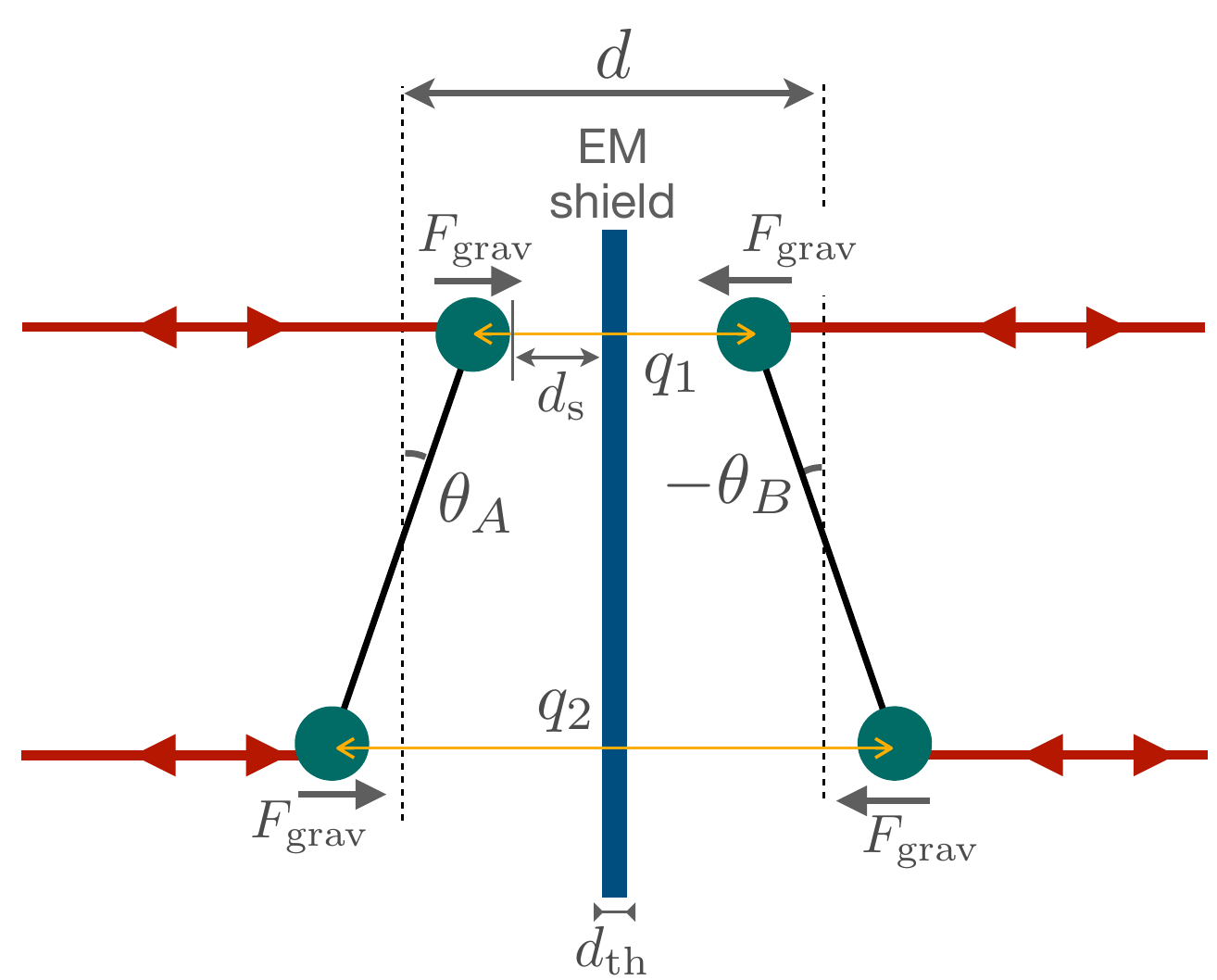} % requires the graphicx package
   \caption{{\bf Gravitational interaction between two torsion balances.} Two torsion pendula are placed with their equilibrium orientations (dashed lines) in parallel and let interact through gravity. An electromagnetic shield of thickness $d_{\rm th}$ is placed in between the two pendula to suppress electromagnetic interactions, EM shield. The minimum distance between the surface of any sphere and the EM shield is given by $d_s$. The rotational degrees of freedom of each pendulum are monitored through their coupling to two cavity fields (red lines).}
   \label{fig:interaction}
   \end{center}
\end{figure} 

The distance $d$ is chosen to guarantee that Casimir forces between the pendula and the shield are always weaker than the gravitational force between the two pendula. Strictly speaking, this is an overestimation of the required minimum distance. In fact, if the shield behaves as a perfect wall with no mechanical degrees of freedom, the Casimir interaction of the wall with each pendulum will never result in an interaction between the pendula, and thus, it would be enough to choose the distance as to prevent the spheres from sticking to the shield. In this case, the only requirement would be to consider the modification of the local dynamics of each oscillator in the light of this Casimir interaction, which, for small oscillations, will simply amount to a constant force and thus to a displacement of the equilibrium position of each oscillator. On the other hand, if the shield exhibits mechanical oscillatory modes, these could, at least in principle, mediate an undesired, non-gravitational interaction between the torsion pendula. Still, the condition that the Casimir force between the pendula and the shield is smaller than the gravitational force between the pendula would, in general, lead to an overestimation of $d$. This is because the effective interaction strength between the pendula mediated by the shield would decrease with the detuning between the oscillation frequencies of the shield and the pendula, which for a stiff shield is expected to be large. In any case, and to be on the safe side, we will stick to this stringent requirement, which, as we will see now, results in a surface-to-surface separation distance $d_{\rm s}$ that is anyway negligible with respect to the size of the spheres, and thus it will not affect the total distance $d$ significantly.  

In order to compute $d_{\rm s}$, we first estimate the maximum displacement that each pendulum is expected to have. According to Eq.~\eqref{lambda_range}, for a system of two oscillators, these have to be prepared in coherent states drawn at random from a normal distribution of width ${\Delta \alpha \approx \sqrt{10}}$.  Thus, the maximum angular displacement of each oscillator will be given by $
{\theta_{\rm max} = \sqrt{10} \Delta \theta_{\rm zpm}}$, with ${\Delta \theta_{\rm zpm} = \sqrt{\hbar/(I \omega_I)}}$ the width of the rotational ground state. For the parameter regime discuss here, this corresponds to ${\Delta \theta_{\rm zpm} \approx \SI{e-13}{\radian}}$, and thus, the angular displacements to be expected during the experiment are upper-bounded by ${\theta_{\rm max} < \SI{e-12}{\radian}}$. The variations in distance between the centers of mass of the gravitationally interacting spheres and the shield are of the order of ${a \theta_{\rm max} \approx \SI{e-13}{\m}}$, and thus, completely negligible compared to the radius of the sphere ${R= (3 m / 4 \pi \rho)^{1/3} \approx \SI{2e-3}{\m}}$. Here we have assumed a sphere of gold with mass density $\rho \approx \SI{1.9e4}{\kg/\m^3}$. The Casimir force between a metallic plate and a sphere can be computed to a good approximation using the proximity force theorem when the radius of the sphere $R$ is much larger than the surface-to-surface distance $d_{\rm s}$ between the sphere and the plate. The formula is~\cite{Capasso2001}
\begin{equation}
F_{\rm C} = - \frac{\pi^3 \hbar c}{360} R \frac{1}{d_{\rm s}}.
\end{equation}
We want to determine the minimum value of $d_{\rm s}$ above which the ratio ${F_{\rm g}/ F_{\rm C} > 1}$, where $F_{\rm g}$ is the gravitational force between the spheres. The ratio is given by
\begin{equation}
\frac{F_{g}}{F_{C}} = \frac{640}{\pi} \frac{G}{\hbar c} \rho^2 R^5 \frac{d_{\rm s}^3}{(2R + 2d_{\rm s} + d_{\rm th})^2},
\end{equation}
where $d_{\rm th}$ is the thickness of the shield. Assuming that ${R \gg \{d_{\rm s}, d_{\rm th}\}}$, we find that the required bound for $d_{\rm s}$ is given by
\begin{equation}
d_{\rm s} > \left( \frac{\pi}{160} \frac{\hbar c}{G} \frac{1}{\rho^2 R^3} \right)^{1/3}.
\end{equation}
For the parameters of our discussion, this amounts to ${d_{\rm s} > 10^{-6}}$ m, which is again a distance much smaller than the radius of the spheres. Thus, for subsequent calculations, we can assume, to a good approximation, that ${d \approx 2R}$. 

For small rotation angles $\theta_A$ and $\theta_B$, the relevant, instantaneous distances between the centers of mass of the four gravitationally interacting masses are given by ${q_1 = d + a (\theta_B - \theta_A)}$ and ${q_2 = d - a (\theta_B - \theta_A)}$, see Fig~\ref{fig:interaction}. For the parameters of our torsion pendula, the rest of the separation distances, e.g.\ between the front end of one pendulum and the rear end of the other one, are two orders of magnitude larger, resulting in gravitational forces four orders of magnitude weaker, and thus we will ignore them in the rest of our analysis.  As discussed earlier, at low energies the gravitational interaction can be described by a Hamiltonian term of Newtonian form, which in our case, with two pairs of interacting masses, takes the form
\begin{equation}
H_{\rm grav} = G m^2 \left( \frac{1}{q_1} + \frac{1}{q_2}\right) .
\end{equation}
For each pair, we can expand the interaction around ${q-d = 0}$ provided that ${a \Delta(\theta_A - \theta_B) \ll d}$,
\begin{equation}
 \frac{Gm^2}{q} \approx \frac{Gm^2}{d} - \frac{Gm^2}{d^2} (q - d) + \frac{2Gm^2}{d^3} (q-d)^2 + \dots \, .
\end{equation}
The term of our interest is the third term in the expansion, which is the first to introduce a quantum interaction between the Hilbert spaces of the two pendula. In particular, the form of the gravitational interaction between the rotational degrees of freedom of each pendulum, after summing over both pairs of particles, is
\begin{equation}
H_{\rm grav} = \dots + \hbar \gamma \left(a_\theta + a^\dag_\theta\right) \left(b_\theta + b^\dag_\theta\right) ,
\label{eq:interaction}
\end{equation}
where the interaction strength is given by 
\begin{equation}
\gamma = \frac{8 G m^2}{d^3} a^2 \Delta \theta_{zpm}^2 = \frac{4  G m^2 a^2}{d^3 \omega_I I},
\end{equation}
and $a_\theta$ and $b_\theta$ are the annihilation operators of rotational excitations in pendula $A$ and $B$, respectively. Assuming again the dumbbell shape of our pendulum $I = 2 m a^2$ and taking ${d \approx 2R}$, as discussed earlier, we get
\begin{equation}
\gamma = \frac{8 \pi}{3} \frac{G \rho}{\omega_I} \approx \SI{1.6e-4}{\Hz}\, .
%1.6 \cdot 10^{-4} \ {\rm Hz}.
\end{equation}
We check that this satisfies the regimes required to convert the interaction in Eq.~\eqref{eq:interaction} to the form of a beam-splitter interaction. That is, ${\Delta \omega = \omega/ Q \ll g}$ and ${\omega \gg g}$. For the first relation we find that ${10^{-8} \ll 10^{-5}}$, while for the second ${10^{-3} \gg 10^{-5}}$; thus, the check is satisfactory. 

Therefore, it seems that the proposed experimental arrangement is able to reach, at least in principle, all the required parameter regimes. The experiment would proceed according to the following protocol:
\begin{enumerate}

\item Both pendula are feedback cooled down to their rotational ground states.

\item The values of the initial coherent states are drawn from a normal probability distribution of width given by Eq.~\eqref{lambda_range}.

\item Each pendulum is optomechanically prepared in the corresponding coherent state.

\item The system is let to evolve coherently for a time $t$, under gravitational interaction.

\item The expected state of the pendula at time $t$ is computed, assuming the interaction follows a beam-splitter evolution.

\item A displacement operation is applied to bring the states computed in the previous step into the ground state.

\item The position of the pendula is monitored through their interaction with the optical cavities to detect if they are in the ground state.

\item The protocol is repeated a statistically significant number of times and the probability of finding the oscillators in their ground states is computed.

\item This is compared to the upper bound on the classical simulation fidelity for the corresponding time as given by Eq.~\eqref{f_2L}.

\item If the value is found to be higher, we have falsified that gravity acts as an LOCC channel.

\end{enumerate}

From the steps above, number~4 poses the greatest challenge. Our analysis in Sec.~\ref{subsec_examples} indicates that, for a two-oscillator system, the bound on the classical simulation fidelity decreases approximately as ${f_{2,L}(t) \approx 1 - 2 \gamma t}$ for times ${t \ll 1/ (2\gamma)}$. Thus, in order to observe a drop in the simulation fidelity of, say, $0.01$, with the parameters of our proposal, we need to let the oscillators interact coherently for a time ${t  \approx \SI{e2}{\s}}$. For an oscillator of quality factor $Q$ in contact with a thermal bath at temperature ${T_{\rm th}}$, one should expect phonon gains at a speed given by the heating rate $\Gamma_{\rm h} = k_B T_{\rm th}/ (\hbar Q) $. For small coherent states, as the ones considered here, a single phonon gain would be enough to displace the system by a distance on the order of the coherent state size. Thus, to guarantee that less than one phonon enters the system during each experimental run, we require that $\Gamma_{\rm h} t < 1$, which, for the times considered here, translates into the condition
\begin{equation}
\frac{T_{\rm th}}{Q} < \SI{e-13}{\K}\ .
\end{equation}

\begin{table}[h!] 
\caption{\textbf{Proposed set of parameters.} Here, symbols represent the following magnitudes: $m$ and $R$ are the mass and the radius of each sphere composing the pendulum; $a$, $\tau$ and $\omega_I$ are the arm length, torsion constant and oscillation frequency of the pendulum; $d_s$ is the surface-to-plate distance between the sphere and the Faraday shield; $d$ is the distance between the equilibrium positions of the two pendula; $\gamma$ the gravitational coupling between them; $t$ the experimental time of each experimental run; and $\frac{T_{\rm th}}{Q}$ the ratio between the environmental temperature and the quality factor. }
\label{tab:parameters}
\begin{center}
\begin{tabular}{c | l | c | l}
\textbf{Parameter} & \textbf{Value} & \textbf{Parameter} & \textbf{Value}\\
\hline
$m$(Kg) & $10^{-4}$ &  $d_s$(m) & $>10^{-6}$ \\
$R$(m) & $2 \cdot 10^{-3}$ & $d$(m) & $\approx 2 R$\\
$a$ (m) & $10^{-1}$ & $\gamma$(Hz) & $1.6 \cdot 10^{-4}$\\
$\omega_I$(Hz) & $7 \cdot 10^{-3}$ & $t$(s) & $10^2$\\
$\tau$(Nm/rad) & $10^{-10}$ & $\frac{T_{\rm th}}{Q}$(K) & $10^{-13}$
\end{tabular}
\end{center}
\end{table}

Admittedly, this is a very challenging condition to achieve experimentally. Considering temperatures in the order of $\SI{}{\milli\kelvin}$, e.g.\ achievable by attaching the pendulum to a dilution refrigerator, this would still require quality factors on the order of $Q = 10^{10}$. This is 5 orders of magnitude above the largest reported quality factors for torsion pendula~\cite{Newman2000, Schlamminger2017}. For some types of wire materials, such as sapphire~\cite{Ooi2018}, improvements in quality factors are to be expected from operation at lower temperatures. Further improvements could be achieved from operation at smaller amplitudes~\cite{Newman2000}, as required by our proposal, and from an enhancement of the vacuum conditions. Moreover, efforts to improve torsional quality factors over $Q> 10^8$ by suppressing dissipation mechanisms, such as surface losses or clamping losses, are underway for the construction of a gravitational-wave antenna based on torsion pendula~\cite{Ooi2018, Takamori2010}. Whether all of this will be enough to reach the quality factors required by our proposal remains to be seen. To enhance clarity, we have condensed the suggested set of experimental parameters in Table~\ref{tab:parameters}.

Before we close this section and move on to the conclusions, we refer the interested reader to Appendix~\ref{Noise}, where we discuss in some depth the main sources of noise that are likely to affect our experimental proposal.

\section{Discussion and conclusions} \label{sec_conclusions}

\subsection{Discussion} \label{subsec_discussion}

Let us summarise our contribution. We have described a general class of experiments that are capable, in principle, %to test 
of testing the quantum nature, %quantumness, 
or non-LOCCness, of an interaction (\S~\ref{subsec_experiments} and Theorem~\ref{general_bound_thm}). Within this class, we have isolated a family of experiments based on gravitationally interacting quantum harmonic oscillators, pinpointing the assumptions (I)--(III) (plus the optional (IV)--(IV')) that need to be met in order for our analysis to apply. These $n$ oscillators are initialised in a random, normally distributed coherent state $\ket{\alpha} = \ket{\alpha_1}\otimes \ldots \otimes \ket{\alpha_n}$, and are subsequently left to interact exclusively via gravity. Theorem~\ref{recap_thm} shows that the state of the system after some time is well approximated by a coherent state $\ket{\alpha'} = \ket{\alpha'_1}\otimes \ldots \otimes \ket{\alpha'_n}$, which we can compute knowing $\alpha$ and the interaction Hamiltonian. The final step of the protocol is to implement the POVM $(\ketbra{\alpha'}, \id-\ketbra{\alpha'})$, whose outcomes we labelled with (NH) and (AH), respectively.
This can be done locally, without any entanglement assistance, by first applying the tensor product displacement unitary $D(\alpha') = D(-\alpha'_1)\otimes \ldots \otimes D(-\alpha'_n)$ and by then measuring each oscillator separately with the POVM $(\ketbra{0},\id-\ketbra{0})$. The global outcome (NH) corresponds to the case where all oscillators are found in the vacuum state, while (AH) corresponds to any other case.

Let us assume that the experiment we described in \S~\ref{subsec_experiments} could indeed be realised, and that a careful investigation of the apparatus were to reveal that assumptions (I)--(III), and maybe even~(IV) or~(IV'), are really obeyed. We would then use either Theorem~\ref{recap_thm} or even Theorem~\ref{recap_small_times_thm} to %compute an upper bound of the form
compute the relevant LOCC inequality, namely, an upper bound of the form
\bb
F_{\cls}\!\left( \mathcal{E}_{\!\lambda}, e^{-\frac{it}{\hbar}H_{\mathrm{tot}}} \right) \leq \overline{F}
\ee
on the LOCC simulation fidelity. By running many rounds of the experiment, we could check whether the empirical frequency $p_{\mathrm{NH}}$ associated with the outcome (NH) violates the above LOCC inequality, i.e.\ it satisfies
\bb
p_{\mathrm{NH}} \overset{?}{>} \overline{F} \geq F_{\cls}\!\left( \mathcal{E}_{\!\lambda}, e^{-\frac{it}{\hbar}H_{\mathrm{tot}}} \right) .
\label{violation?}
\ee
If that is the case, what can we conclude? We propose that the answer to this question is that one of the following must be true:
\begin{enumerate}[(A)]
\item either gravity somehow `knows' the initial state of the system, $\psi_\alpha$, and uses this knowledge to by-pass the above theoretical bounds, reproducing at the output the correct state with high fidelity; or 
\item \tcr{the dynamical map induced by the gravitational interaction between two or more bodies cannot be described by LOCC operations alone.}
\end{enumerate}
This is, of course, provided that the assumptions stated above are taken for granted. The crux of the argument is that if~(A) is false, and hence gravity does not know the initial state of the system, then all it can do via LOCCs is to reproduce the outcome (NH) with frequency at most $F_{\cls}$. This contradicts~\eqref{violation?}.

Concerning~(A), note that even an agent subjected to locality constraints can reproduce the output state of the gravitational dynamics with high fidelity \emph{provided that they know the input} $\psi_\alpha$. Indeed, as remarked below Theorem~\ref{recap_thm}, the output state will be well approximated by a coherent state, which is separable by virtue of~\eqref{coherent_multimode} and can therefore be created from product states by means of LOCCs. This reasoning shows that the experiment we have described, as it stands, cannot exclude a sort of super-deterministic explanation in which gravity is indeed described by a local and classical gravitational field, the behaviour of which is however correlated with our choice of $\psi_\alpha$ --- that, remember, we assumed to be unknown to the agent simulating the unitary dynamics. In other words, our LOCC inequalities have \emph{loopholes}, in the same way as Bell inequalities do. And in precisely the same way, %That being said, 
it is not difficult to envisage how such loopholes could be closed, at least in principle.

To this end, one could use quantum mechanics and causality to generate the randomness needed to draw the state $\psi_\alpha$ from the ensemble $\EE_\lambda$. Due to the fact that $\EE_\lambda$ is a Gaussian coherent state ensemble, there is a conceptually simple way to do so. Before the experiment takes place, each oscillator should be paired with another, identical one, and the two should be prepared in a two-mode squeezed vacuum state $\ket{\Psi(r)}\coloneqq \frac{1}{\cosh(r)} \sum_{k=0}^\infty \tanh^k(r) \ket{kk}$, where $r>0$ is the squeezing parameter. The $n$ ancillary oscillators are then sent to a remote location. As the experiment begins, a heterodyne measurement represented by the POVM $\big(\ketbra{z}\big)_{z\in \C}$ is carried out on every ancillary oscillator, and the outcome $z=(z_1,\ldots, z_n)$ is transmitted at the speed of light towards the main laboratory. Following the measurement, the state in the main laboratory has immediately collapsed to the coherent state $\ket{\alpha}$, where $\alpha = z\tanh(r)$ is distributed normally (as in~\eqref{E_lambda}) with variance $1/\lambda = \cosh^2(r)$. The gravitational dynamics then starts to act on the system, although no information about the value of $z$ (equivalently, of $\alpha$) has reached the main laboratory yet.

When the signal containing $\alpha$ arrives, the appropriate measurement $(\psi'_\alpha, \id-\psi'_\alpha)$ is carried out on the main oscillators. The idea here is that the information on $\alpha$ reaches the main laboratory only at the end of the experiment; therefore, a malicious gravitational field would not have been able to use it to reproduce the final state with LOCCs only. Admittedly, this route may not be the most practical. Indeed, as we have seen in \S~\ref{subsec_experimental_considerations}, the experiment could well take minutes to complete; this would mean that the remote location should be a few light minutes away from the main laboratory. Nevertheless, it shows that a solution is conceptually possible.

As for~(B), note that it does not imply that gravity be described by a quantum field; in fact, it does not tell us much about how a theory of quantum gravity might look. It merely implies that whatever gravity is, it cannot be universally described by \tcr{LOCC operations. This excludes, for example, the notion of a classical field mediating the interaction, that is, a field that takes a specific value at each point in space and interacts locally with matter.}

If the empirical frequency $p_{\mathrm{NH}}$ in~\eqref{violation?} is close to $1$, then our experiment is also telling us something more, namely, that the dynamics of the system is represented to a good approximation by the unitary operator corresponding to the \emph{quantum} Newtonian Hamiltonian~\eqref{system_oscillators_H_tot}. What happens instead if the dynamics is really genuinely quantum, but it is not modelled by the Newtonian Hamiltonian~\eqref{system_oscillators_H_tot}, intended as a quantum operator? In that case the results of the experiment may well be inconclusive, simply because the measurement we are carrying out, $(\psi'_\alpha, \id-\psi'_\alpha)$, is not the right one to detect the final state of the system. Testing an alternative form of a quantum dynamics should be in principle possible by modifying that final measurement.

At this point it is prudent to note that the test described here requires some prior knowledge, specifically about the type of unitary evolution that the particles undergo when only subjected to gravitational interaction. This, in turn, requires knowledge of the mass %and distance of the test particles. 
of the test particles and of their distances. In real-world experiments, we will never know the quantities with absolute precision (see also Appendix \ref{Noise} for a more in-depth analysis of the main sources of noise and imperfections that we anticipate). We emphasize, however, that some degree of uncertainty in this regard, such as imprecise knowledge of mass, does not immediately invalidate our approach. Our LOCC inequality is a continuous function of the density matrix and the unitary operation involved. This follows from our Theorem~\ref{general_bound_thm}, given that the key quantity in question, the conditional min-entropy, is known to be continuous, as established in \cite{Tomamichel2010}. Consequently, while small uncertainties in the state or the unitary operation may reduce the gap between the highest fidelity achievable by LOCC gravity and by quantum gravity, due to the continuity of these properties only substantial uncertainties would render the test inconclusive.

We may also dispense with the requirement of prior knowledge assumed so far by employing process tomography on the channel created by the gravitational %, and possible other, 
interaction, and possibly by other interactions, between Alice and Bob's particles. This can, once again, be accomplished by introducing coherent states on both Alice and Bob's particles and subsequently applying quantum state tomography through local measurements on the states resulting from this interaction. In principle, this method is sufficient to determine the characteristics of the channel without the need for prior channel knowledge. However, it is worth noting that this approach is more resource-intensive in terms of experimental resources compared to the previously presented protocol.

Finally, let us discuss how our proposal compares with previous ones. Most of the recent literature has focused on the idea of certifying the non-LOCCness of gravity by detecting the entanglement it generates, either in a matter-wave interferometer~\cite{lindner2005testing, Bose2017, Pedernales2020, Pedernales2022}, or between modes of light~\cite{Miao2020}, or 
%--- in the same spirit as our paper --- 
else between quantum harmonic oscillators~\cite{Krisnanda2020, cosco2021enhanced, weiss2021large}. As previously mentioned, the main advantage of our approach is that it does away with the need to detect entanglement in the system altogether. This is because our theoretical bounds can certify the non-LOCCness of the dynamics even when entanglement is never present in the system.

To see how this conceptual distinction actually turns into a practical advantage, let us have a closer look at previous proposals such as~\cite{Krisnanda2020, cosco2021enhanced}. Note that the interaction Hamiltonian considered e.g.\ in~\cite[Eq.~(2)]{Krisnanda2020} is essentially the same as our~\eqref{H_eff_two}. Also, the base case considered in~\cite{Krisnanda2020}, that of two massive harmonic oscillators initialised in their ground states and interacting exclusively via gravity, is very similar to our setting, which features coherent states, i.e.\ displaced ground states. In light of these similarities, it is remarkable that the signals predicted by the two schemes are so different: while following~\cite{Krisnanda2020} one finds that the two oscillators achieve a maximum entanglement of the order of $\frac{2Gm}{d^3\omega^2} = \frac{2\gamma}{\omega}$ after a time $t\approx \pi/(2\omega)$, where $\gamma\coloneqq \frac{Gm}{d^3\omega}$ is as in~\eqref{gamma}, in \S~\ref{subsec_examples} we have shown that the right-hand side of our LOCC inequality, i.e.\ our upper bound on the classical simulation fidelity, decreases as $1-2\gamma t$ down to its smallest value $1/4$, achieved for $t=\pi/(2\gamma)$. Therefore, while waiting for longer boosts the signal in our scheme, at least in principle, the maximum signal strength in~\cite{Krisnanda2020, cosco2021enhanced} is bounded at all times by the typically small number $2\gamma/\omega$, at least \emph{in the absence of squeezing} --- more on that below. We believe that this advantage may prove important in practical implementations.

The technical reason that explains this fundamental difference is that the interaction Hamiltonian considered in both~\cite{Krisnanda2020, cosco2021enhanced} as well as in our work, of the form $x_1 x_2$, contains in fact two interaction terms of different types: (i)~a beam splitter term of the form $a^\dag b + a b^\dag$, which does not generate entanglement when the initial state is coherent; and (ii)~a two-mode squeezing term of the form $a b + a^\dag b^\dag$. Since only entanglement generation is considered in~\cite{Krisnanda2020, cosco2021enhanced}, it is only term~(ii) that plays a role there. The key drawback of that approach, however, is that such a term is strongly suppressed by the rotating-wave approximation. On the other hand, since we do away with the need to generate entanglement, the signal we predict, which comes mainly from term~(i) and thus survives the rotating-wave approximation, is substantially stronger.

It is important to note that although~\cite{Krisnanda2020, cosco2021enhanced} do discuss the case of entanglement generation starting from coherent states, this is done mainly for explanatory purposes, while the actual proposal there revolves around employing squeezed states of motion. This is a possible solution to obtain a decisive signal boost also in the entanglement-based scheme. However, preparing squeezed states is bound to entail a substantial increase in overall experimental complexity, while in our proposal we only employ much more easily attainable coherent states. For this reason, we believe that our scheme may have distinctive advantages over the squeezing-based ones.

The same considerations also apply when we compare our proposal to others based on matter-wave interferometry. These involve relatively large masses with their centers of mass prepared in a superposition of two distinct coherent states, that is, in a Schr\"odinger-cat-like state. Preparing these states requires the introduction of non-linearities in the Hamiltonian, either by adding an ancillary system, e.g. a two-level system, or by modifying the trapping potential to contain terms of order higher than quadratic. Arguably, both of these options entail experimental challenges significantly greater than that of preparing coherent states.
%The initial quantum states we need to prepare are coherent states. Hence, they are relatively simple compared to the states needed for matter-wave interferometry, which involve relatively large masses placed in superpositions with high degree of spatial delocalisation; this makes such states are susceptible to localising noise, such as black body radiation and background gas collisions to which our approach is less susceptible. 
Indeed, as we mentioned in the Introduction, the heaviest object for which matter-wave interference has been observed is a large molecule with mass $\sim \SI{4e-23}{\kg}$~\cite{Fein2019}, still much too small to generate any appreciable gravitational field. 

Another proposal to which it is worth comparing our own is that contained in the recent work~\cite{Matsumura2022}. In that proposal, maps that can generate entanglement are distinguished from maps that cannot, even though no entanglement is actually generated during the experiment. In spite of this apparent similarity --- the absence of generated entanglement --- there are at least three crucial differences between~\cite{Matsumura2022} and our work.
%Also in that proposal no entanglement generation is strictly required; however, the similarities between~\cite{Matsumura2022} and this work terminate here. 

The most important one from the conceptual standpoint, and certainly the most relevant for practical applications, is that although they contain no entanglement, the initial states of~\cite{Matsumura2022} are still highly delocalised states of relatively large masses. Such states are precisely what is required in entanglement-based proposals, and preparing them is well known to be exceptionally challenging from an experimental point of view. In this sense, the experimental difficulty associated with the protocol in~\cite{Matsumura2022} is comparable to that of entanglement-based tests. The initial states required by our proposal, on the contrary, are simple coherent states, arising naturally as (displaced) ground states of harmonic oscillators. As such, they are more easily obtained than the delocalised states of~\cite{Matsumura2022}.

The aim of the proposal in~\cite{Matsumura2022} is to certify that the dynamics induced by gravity is not a non-entangling map --- i.e.\ it \emph{can} generate entanglement. This stands in contrast with our framework, which aims to detect the non-LOCCness of the dynamics. Importantly, the class of non-entangling operations is strictly larger than the class of LOCC maps. That is, there are non-entangling map that are not LOCC~\cite{Bennett1999} --- in fact, the two classes lead to very different entanglement manipulation rates~\cite{HorodeckiBound, DNE-distillable}. Our result is threfore strictly more powerful as it allows us to discriminate whether gravity is of non-LOCC form (and therefore non-classical), even if it would not have the ability to generate entanglement. Furthermore, our set-up allows to test for even finer graduations, as different classes of operations, and thus different underlying physical models of gravity, will lead to their own fidelity thresholds, which may then be tested in our experiment. In this sense, our test is more refined not only than that in~\cite{Matsumura2022}, but also than experiments that test merely the generation of entanglement.

Finally, the non-entangling maps excluded by the witness found in~\cite{Matsumura2022} are \emph{only} those that satisfy the population-preservation hypothesis. While this is certainly a reasonable assumption in the context of that work, we would like to remark that our non-LOCC inequalities are much more general: their violation allows us to automatically exclude \emph{all} LOCC processes, without any restriction.

In summary, the majority of the existing literature %has concentrated 
has focused on utilising widely spread states, such as the aforementioned squeezed Gaussian and Schr\"odinger-cat states, for their ability to boost the rate of entanglement generation in entanglement-based tests~\cite{Pedernales2023}. However, besides the already discussed challenges in their preparation, these states also represent an increased dissipation channel, as their delocalisation not only enhances the interaction between the objects of interest but also the interaction of each of the bodies with their environment. In contrast, in our proposed test, the interacting bodies remain at all times in separable coherent states. We consider this another significant advantage of our protocol over previous ones, as this eliminates the %challenges 
obstacles related to %linked to 
state preparation and also removes the decoherence associated with highly delocalised systems. 

%Naturally, our proposal is not without challenges. The primary %hurdles
%\tcb{ones} lie in achieving the required long coherence times, demanding high-quality factors and low environmental temperatures, which themselves represent important experimental challenges. %However, we view these hurdles as a different set of experimental challenges. Thus, 
%While comparing %the difficulty of these challenges 
%\tcb{these} with those in entanglement-based tests may be elusive, \tcb{as the two frameworks ultimately present their own unique sets of challenges,} we believe that our proposal opens a new experimental avenue with its unique set of challenges, and we think that this is inherently enriching and positive, as it allows the exploration of diverse paths toward a shared goal .

Naturally, our proposal is not without challenges. The primary ones lie in achieving the necessary long coherence times, high-quality factors, and low environmental temperatures, which themselves require important experimental progress. While comparing these difficulties to those associated with entanglement-based tests may be elusive, as the two frameworks ultimately present their own unique sets of challenges, we believe that our proposal opens a new experimental avenue. This is bound to be inherently enriching and positive, as it allows the exploration of diverse paths toward a shared goal.

%On the other hand, our proposal requires that the oscillators be cooled down close to the ground state (as a coherent state is nothing but a displaced ground state). Although doing so involves a different set of experimental challenges, we believe that it may prove easier. For example, the heaviest oscillator for which a state of a few ($<11$) phonons has been prepared has a mass of $\SI{10}{\kg}$~\cite{ligo-mirrors}. (It corresponds to the degree of freedom associated with the differential motion between the centres of mass of the two pairs of mirrors making up LIGO's interferometer.)

\subsection{Conclusions} \label{subsec_conclusions}

The weight of modern physics lies on two well-tested albeit ill-matching pillars, namely, general relativity and quantum mechanics. While a marriage between the two has been sought right since their birth at the dawn of the XX century, the truth is that physicists still await an invitation to the wedding. Arguably, one of the aspects that are delaying this is the lack of \tcr{access to experiments where the source of a measurable gravitational field exhibits quantum properties such as superposition of its center of mass}. Typically, this has been considered to be an experimentally inaccessible regime with state-of-the-art technology; a perception that might be changing with the advent of quantum technologies, and in particular with that of quantum optomechanics, which is showing an ever-improving quantum control of massive systems. While a direct test of the quantisation of the gravitational field remains a daunting task, indirect ways of testing quantum aspects of gravity have recently stirred the physics community. In particular, significant interest has been raised towards experiments that could, in principle, test whether the gravitational interaction between masses with quantum mechanical degrees of freedom can be reproduced by using only local operations and classical communication (LOCC). To the best of our knowledge, all of the existing proposals in this respect rely on looking for the generation of entanglement mediated by gravity, a phenomenon that is impossible to reproduce if the interaction is modelled by an LOCC channel. 

Here we have shown, in stark contrast, that entanglement is not a necessary ingredient for this type of test. To achieve this we have established a general LOCC inequality, i.e.\ an upper bound on the fidelity that an LOCC protocol can achieve when trying to reproduce certain classes of unitary evolutions. This can then be tested for a subset of states, although, under such an evolution, these same states never get entangled. Finding final output state fidelities above the computed threshold would thus violate the LOCC inequality, indicating that the evolution is not of an LOCC type. This has immediate implications for the experimental efforts to implement this type of tests. Entanglement-based tests of gravity require the generation of large spatial superpositions of massive systems, a task of significant experimental challenge. Our result shows that this is not necessary, thereby removing one important obstacle towards this goal. We have provided a detailed analysis of the prospects for an experimental implementation of our ideas and shown that, while technology would still require significant improvement, these requirements are not harder to meet than those of entanglement-based tests. In fact, while our experiment would necessitate similarly low dissipation rates, it does not require the generation of coherently delocalised quantum states, which is extremely challenging experimentally, nor the detection of entanglement, which, for gravitational experiments, is expected to be extremely small and fragile to all sorts of decoherence sources, and thus very challenging to detect.

In short, our result represents a conceptual shift with respect to the established notions regarding tests of quantum aspects of gravity with quantum technologies, and with it opens the door to a novel class of experiments.

\medskip
\paragraph*{Acknowledgements.} We thank Gerardo Adesso, Jonathan Oppenheim, as well as an anonymous referee of QIP 2023 for useful comments. LL was partly supported by the Alexander von Humboldt Foundation. JSP and MBP are supported in part by the ERC via Synergy grant HyperQ (grant no.\ 856432) and the DFG via QuantERA project Lemaqume.

\bibliography{biblio}

\appendix

\section{Detour: classical thresholds for teleportation as a special case of Theorem~\ref{general_bound_thm}} \label{subsec_telep_threshold}

The usefulness of our methods, and in particular of Theorem~\ref{general_bound_thm}, is best illustrated by applying it to the simplest possible case, that in which there are only two parties ($n=2$) and the $A_2$ and $A'_1$ sub-systems are trivial, $A_1\simeq A'_2$ are isomorphic, and moreover $U_{A\to A'} = \id_{A_1\to A'_2}$ is the identity channel. Although this setting is not immediately relevant to the application considered here, it has been the subject of extensive investigation in the context of classical simulation of teleportation. Indeed, the task at hand consists in simulating a perfect teleportation channel $A_1\to A'_2$ on an ensemble of pure states $\mathcal{E}$ using only operations that are LOCC with respect to the splitting $1:2$, that is, using measure-and-prepare (MP) channels $A_1\to A'_2$, i.e.\ channels of the form $\Lambda(X_{A_1}) = \sum_\ell \Tr\big[X_{A_1} E_\ell^{A_1}\big] \rho^{(\ell)}_{A'_2}$, where $(E_\ell)_\ell$ is a (discrete) POVM, and $\rho^{(\ell)}_{A'_2}$ are states. The quantity 
\bb
P_2 &= F_{\cls}(\mathcal{E}, \id_{A_1\to A'_2}) \\
&= \sup_{\Lambda\in \mathrm{MP}(A_1\to A'_2)} \sum_\alpha p_\alpha \Tr \left[ \Lambda(\psi_\alpha) \psi_\alpha \right] ,
\label{classical_fidelity_teleportation_def}
\ee
i.e.\ the maximal average fidelity obtainable by such a simulation, is also known in this context as the \deff{classical threshold}. This quantity has garnered a lot of attention as a benchmark for quantum experiments, and has been computed in a variety of scenarios: (a)~for Haar-random ensembles of qubit and qudit states~\cite{Massar1995, Werner1998, Horodecki-teleportation, Bruss1999}; (b)~for Gaussian ensembles of coherent states~\cite{Hammerer2005, BUCCO}; (c)~for ensembles of general Gaussian states~\cite{Adesso2008, Owari2008, jensen2011quantum, Chiribella2014}; and finally (d)~for ensembles of continuous-variable states invariant under some group action~\cite{Calsamiglia2009, Yang2014}. With our methods it is possible to recover several of these results swiftly and with a much more transparent proof. Applied to this special case, our Theorem~\ref{general_bound_thm} gives the bound
\bb
F_{\cls}(\mathcal{E}, \id_{A_1\to A'_2}) \leq \inf\left\{ \kappa: \sum_\alpha p_{\!\alpha}\, \psi_{\!\alpha}^{\phantom{'}}\! \otimes \psi'_{\!\alpha}\! \leq \kappa \xi\!\otimes\! \id\right\} .
\label{classical_fidelity_teleportation_bound}
\ee
Below we show how to recover from this two classic results.

\medskip
\paragraph{Haar-random qudit states.} Let $\EE_{\mathrm{H},d}$ be the ensemble composed of all qudit states with the Haar measure. Then it is known that~\cite{Massar1995, Werner1998, Horodecki-teleportation, Bruss1999}
\bb
F_{\cls}\left(\EE_{\mathrm{H},d}, \id_{A_1\to A'_2}\right) = \frac{2}{d+1}\, .
\label{F_cls_telep_Haar_random}
\ee
While a lower bound on the left-hand side is easy to find by making an appropriate ansatz for the measure-and-prepare channel $\Lambda$ in~\eqref{classical_fidelity_teleportation_def}, a matching upper bound can be derived from~\eqref{classical_fidelity_teleportation_bound} by remembering from~\eqref{R_AA'_Gamma_swap} that
\bb
\int \!\dd\psi\ \psi\otimes\psi = \frac{2 S}{d(d+1)}\, ,
\label{noisy_max_ent_state}
\ee
where $S$ is the projector onto the symmetric subspace, and then by making the ansatz $\xi = \id/d$ and $\kappa = 2/(d+1)$. Note that~\eqref{F_cls_telep_Haar_random} implies that the bound on the classical simulation fidelity of the swap operation we found in~\eqref{swap_bound} is in fact tight: indeed, one can always simulate a swap by teleporting each state on the other side separately.

\medskip
\paragraph{Ensembles of coherent states.} Let $p$ be a phase-invariant probability distribution on the complex plane $\C$, i.e.\ $p(e^{i\varphi}\alpha)=p(\alpha)$ for all $\varphi\in \R$ and $\alpha\in \C$.
%~\footnote{That is, such that $p(e^{i\varphi}\alpha)=p(\alpha)$ for all $\varphi\in \R$.} 
Then one can consider the ensemble of coherent states $\EE_p = \{p(\alpha), \ketbra{\alpha} \}_{\alpha\in \C}$. A notable ensemble of this kind is the Gaussian ensemble $p(\alpha) = p_\lambda(\alpha) = \tfrac{\lambda}{\pi}\, e^{-\lambda |\alpha|^2}$, where $\lambda>0$, which will play a key role later in the paper. In~\cite{Hammerer2005} (see also~\cite[\S~8.1.2]{BUCCO}), it was shown that the corresponding classical threshold evaluates to
\bb
F_{\cls}\left(\EE_{p_\lambda}, \id_{A_1\to A'_2}\right) = \frac{1+\lambda}{2+\lambda}\, ,
\label{classical_fidelity_Hammerer}
\ee
converging in particular to $1/2$ as $\lambda\to 0$ and thus the ensemble becomes very spread out. Once again, establishing a lower bound on $F_{\cls}\left(\EE_{p_\lambda}, \id_{A_1\to A'_2}\right)$ is relatively easy --- a heterodyne measurement is the right ansatz --- and the difficult part is to prove an upper bound. The main drawbacks of the proof of the upper bound in~\cite{Hammerer2005} are that (i)~it is somewhat lengthy --- it occupies 6~pages in~\cite[pp.~232--237]{BUCCO} --- and (ii)~it cannot be easily generalised to arbitrary phase-invariant distributions. We can solve both of these issues with the help of our Theorem~\ref{general_bound_thm}.

\begin{prop} \label{teleportation_benchmark_coherent_prop}
For any phase-invariant probability distribution $p$ on $\C$, it holds that
\begin{align}
F_{\cls}\left(\EE_p, \id_{A_1\to A'_2}\right) &\leq \inf_q \sup_{N\in \N} \mu_N \sum_{k=0}^N \frac{1}{q_k} \binom{N}{k} \label{teleportation_benchmark_coherent_1} \\
&\leq \sup_{N\in \N} \mu_N \sum_{k=0}^N \frac{1}{\nu_k} \binom{N}{k} \, ,
\label{teleportation_benchmark_coherent_2}
\end{align}
where the infimum is over all probability distributions $q$ on $\N$, and
\begin{align}
\mu_N &\coloneqq \frac{1}{N!} \int\! \dd^2\!\alpha\ p(\alpha)\, e^{-2|\alpha|^2} |\alpha|^{2N}\, , \label{mu_N} \\
\nu_k &\coloneqq \frac{1}{k!} \int\! \dd^2\!\alpha\ p(\alpha)\, e^{-|\alpha|^2} |\alpha|^{2k}\, . \label{nu_k}
\end{align}
In particular, for $p(\alpha) = p_\lambda(\alpha) = \tfrac{\lambda}{\pi}\, e^{-\lambda |\alpha|^2}$ we recover~\eqref{classical_fidelity_Hammerer}.
\end{prop}

\begin{proof}
We start by computing the operator appearing in~\eqref{classical_fidelity_teleportation_bound}, which in this case takes the form
\bb
&\int\hspace{-1.0ex} \dd^2\!\alpha\, p(\alpha) \ketbra{\alpha}\otimes \ketbra{\alpha} \\
&\eqt{1} \!\sum_{h,k,h',k'=0}^\infty \!\!\! \frac{\ketbraa{h}{k}\!\otimes\! \ketbraa{h'}{k'}}{\sqrt{h!\,k!\,h'!\,k'!}} \!\int\hspace{-1.0ex} \dd^2\!\alpha\, p(\alpha)\, e^{-2|\alpha|^2} \alpha^{h+h'} (\alpha^*)^{k+k'} \\
&\eqt{2} \sum_{N=0}^\infty \sum_{h,k=0}^N \frac{\ketbraa{h}{k}\!\otimes\! \ketbraa{N\!-\!h}{N\!-\!k}}{\sqrt{h!\,k!\,(N\!-\!h)!\,(N\!-\!k)!}} \!\int\hspace{-1.0ex} \dd^2\!\alpha\, p(\alpha)\, e^{-2|\alpha|^2} |\alpha|^{2N} \\
&\eqt{3} \sum_{N=0}^\infty \mu_N\! \sum_{h,k=0}^N \ketbraa{h}{k}\!\otimes\! \ketbraa{N\!-\!h}{N\!-\!k} \sqrt{\binom{N}{h}\binom{N}{k}} \\
&\eqt{4} \sum_{N=0}^\infty 2^N\mu_N \ketbra{\phi_N}\, .
\ee
Here, in~1 we used~\eqref{coherent}; in~2 we noticed that $\int\! \dd^2\!\alpha\, p(\alpha)\, e^{-2|\alpha|^2} \alpha^{h+h'} (\alpha^*)^{k+k'}=0$ unless $h+h'=k+k'\eqqcolon N$, due to the phase invariance of $p$; in~3 we employed the definition of $\mu_N$, given by~\eqref{mu_N}; finally, in~4 we defined the states
\bb
\ket{\phi_N} \coloneqq 2^{-N/2} \sum_{k=0}^N \sqrt{\binom{N}{k}} \ket{k}\otimes \ket{N-k}\, .
\ee
Note that the $\ket{\phi_N}$ are orthonormal states, because each of them belongs to a different subspace of fixed total photon number $N$, whose corresponding projector we denote by
\bb
\Pi_N \coloneqq \sum_{k=0}^N \ketbra{k}\otimes \ketbra{N\!-\!k}\, .
\ee

Now, to prove~\eqref{teleportation_benchmark_coherent_1}--\eqref{teleportation_benchmark_coherent_2} we resort once more to~\eqref{classical_fidelity_teleportation_bound}. Setting $\xi \coloneqq \sum_{k=0}^\infty q_k \ketbra{k}$, we see that the operator inequality
\bb
\sum_{N=0}^\infty 2^N\mu_N \ketbra{\phi_N} \leq \kappa\, \xi\otimes \id
\ee
is satisfied if and only if
\bb
2^N\mu_N \ketbra{\phi_N} &\leq \kappa\, \Pi_N (\xi\otimes \id) \Pi_N \\
&= \kappa \sum_{k=0}^N q_k \ketbra{k}\!\otimes\! \ketbra{N\!-\!k}
\ee
for all $N\in \N$. Remembering that $b \ketbra{\psi} \leq A$ if and only if $\psi\in \supp(A)$ and $b \braket{\psi|A^{-1}|\psi}\leq 1$, where the inverse of $A$ is taken on the support, it is immediate to deduce from this that the minimal $\kappa$ for which this can hold is given by the right-hand side of~\eqref{teleportation_benchmark_coherent_1}. The inequality in~\eqref{teleportation_benchmark_coherent_2}, instead, is obtained by making the ansatz $q_k=\nu_k$, where $\nu_k$ is given by~\eqref{nu_k}.

Finally, for the Gaussian case a straightforward calculation shows that $\mu_N = \lambda (2+\lambda)^{-N-1}$ and $\nu_k = \lambda (1+\lambda)^{-k-1}$, so that $\mu_N \sum_{k=0}^N \frac{1}{\nu_k} \binom{N}{k} = (1+\lambda)/(2+\lambda)$ for all $N\in \N$, recovering the Hammerer et al.'s result~\eqref{classical_fidelity_Hammerer}~\cite{Hammerer2005}.
\end{proof}

\section{Proof of Theorem~\ref{general_bound_thm} with general infinite-dimensional ensembles} \label{infinite_dim_app}

Throughout this appendix, we explain how to modify the proof of Theorem~\ref{general_bound_thm} to adapt it to the general setting where $\EE = \{ \dd\mu(\alpha), \psi_\alpha\}$ is a general ensemble of possibly infinite-dimensional states $\psi_\alpha = \ketbra{\psi_\alpha}$, where $\ket{\psi_\alpha}\in \HH$ and the only assumption on $\HH$ is that it is separable as a Hilbert space. Here, $\mu$ is a generic probability measure on some measure space, and accordingly~\eqref{Fs} takes the form
\bb
F_{\cls}(\mathcal{E}, U) \coloneqq \sup_{\Lambda\in \locc(A\to A')} \int \dd\mu(\alpha)\ \Tr \left[ \Lambda(\psi_\alpha) \psi'_\alpha \right] ,
\label{Fs_general}
\ee
while~\eqref{R_gamma} becomes
\bb
R_{AA'} \coloneqq \int \dd\mu(\alpha)\ (\psi_\alpha^*)_{A} \otimes (\psi'_\alpha)_{A'}\, . \label{R_gamma_general}
\ee
Clearly, it suffices to prove~\eqref{general_bound}, as~\eqref{handier_general_bound} follows exactly as in the proof of Theorem~\ref{general_bound_thm} presented in the main text. To achieve this, we modify~\eqref{general_bound_proof} as follows: let us take a generic LOCC operation $\Lambda\in\locc$ and consider $\kappa>0$ such that $R_{AA'}^{\Gamma_J} \leq \kappa\, \xi_A\otimes \id_{A'}$. Then
\begin{align}
&F_{\cls}(\mathcal{E}, U) \nonumber \\
&\quad = \int \dd\mu(\alpha)\ \braket{\psi'_\alpha| \Lambda(\psi_\alpha) |\psi'_\alpha} \nonumber \\
&\quad\eqt{1} \int \dd\mu(\alpha)\ \lim_{d\to\infty} \braket{(\psi_\alpha^*)_{A} (\psi'_\alpha)_{A'} \big| D_{\Lambda,d}^{AA'} \big| (\psi_\alpha^*)_{A} (\psi'_\alpha)_{A'}} \nonumber \\
&\quad \leqt{2} \liminf_{d\to\infty} \int \dd\mu(\alpha)\ \braket{(\psi_\alpha^*)_{A} (\psi'_\alpha)_{A'} \big| D_{\Lambda,d}^{AA'} \big| (\psi_\alpha^*)_{A} (\psi'_\alpha)_{A'}} \label{general_bound_calculation_infinite_dim} \\
&\quad\eqt{3} \liminf_{d\to\infty} \Tr \left[ R_{AA'} D_{\Lambda,d}^{AA'} \right] \nonumber \\
&\quad\leqt{4} \liminf_{d\to\infty} \kappa \nonumber \\
&\quad = \kappa\, .
\end{align}
The most delicate step in the above calculation is~1. There, we introduced the un-normalised Choi--Jamio{\l}kowski state $D_{\Lambda,d}^{AA'}$ defined over the first $d$ basis vectors, formally given by~\eqref{dynamical_matrix} and~\eqref{max_ent}, but where now $\{\ket{1},\ket{2},\ldots \}$ is a possibly infinite Hilbert basis of $\HH_A$. Calling $\Pi_d\coloneqq \sum_{\ell=1}^d \ketbra{\ell}$, one observes that~\eqref{max_ent_trick} can be modified so as to give $M\otimes \id \ket{\Phi_d} = \id\otimes \big(\Pi_d M^\intercal\big) \ket{\Phi_d}$ for any bounded operator $M$. Then, as in~\eqref{calculation_dynamical_matrix} we compute
\begin{align}
&\braket{\psi_\alpha^* \psi'_\alpha | D_{\Lambda,d} |\psi_\alpha^* \psi'_\alpha} \nonumber \\
&\qquad = d \Tr \left[\psi_\alpha^*\otimes \psi'_\alpha\, (I\otimes \Lambda)(\Phi_d) \right] \nonumber \\
&\qquad = d \Tr \left[\id \otimes \psi'_\alpha\, (I\otimes \Lambda)\left((\psi_\alpha^* \otimes \id)\, \Phi_d\right) \right] \nonumber \\
&\qquad = d \Tr \left[\id \otimes \psi'_\alpha\, (I\otimes \Lambda)\left((\id\otimes (\Pi_d \psi_\alpha))\, \Phi_d\right) \right] \label{calculation_dynamical_matrix} \\
&\qquad = d \Tr \left[\psi'_\alpha\,\Lambda\left(\Tr_1\, [ (\id\otimes (\Pi_d\psi_\alpha))\, \Phi_d ] \right) \right] \nonumber \\
&\qquad = \Tr \left[\psi'_\alpha\, \Lambda( \Pi_d \psi_\alpha \Pi_d) \right] \nonumber \\
&\qquad = \braket{\psi'_\alpha| \Lambda(\Pi_d \psi_\alpha \Pi_d)|\psi'_\alpha} , \nonumber
\end{align}
which implies in particular that
\bb
\lim_{d\to\infty} \braket{\psi_\alpha^* \psi'_\alpha | D_{\Lambda,d} |\psi_\alpha^* \psi'_\alpha} = \braket{\psi'_\alpha| \Lambda(\psi_\alpha)|\psi'_\alpha}
\ee
due to the strong convergence of $\Pi_d \psi_\alpha \Pi_d$ to $\psi_\alpha$ as $d\to\infty$ for each $\alpha$. This completes the justification of~1 in~\eqref{general_bound_calculation_infinite_dim}. Continuing, in~2 we applied Fatou's lemma, in~3 we used~\eqref{R_gamma_general}, and finally~4 follows as in~\eqref{general_bound_proof}. This completes the proof of Theorem~\ref{general_bound_thm} in the general case.

\section{Taylor expansion of the Hamiltonian of gravitationally interacting harmonic oscillators} \label{Taylor_app}

Consider a system of $n$ particles, each bound to a one-dimensional harmonic oscillator (Figure~\ref{system_oscillators_fig}). We assume that these particles interact with each other gravitationally, and that the oscillation amplitude is much smaller than the distance between the oscillator centres. In this appendix, we employ this approximation to perform a Taylor expansion of the total Hamiltonian.

In what follows, we denote the position of the $j^\text{th}$ centre with $\vec{R}_j^0\in \R^3$, the angular frequency of the harmonic oscillator with $\omega_j$, the mass of the particle with $m_j$, the direction of the oscillation line with $\hat{n}_j\in \R^3$, where $\|\hat{n}_j\|=1$, and the position of the particle along that line with $x_j\in \R$. The position of the $j^\text{th}$ particle is thus $\vec{R}_j = \vec{R}_j^0 + x_j \hat{n}_j$. Remembering that we set $\vec{d}_{jk}\coloneqq \vec{R}_k^0 - \vec{R}_j^0$, we see that the gravitational interaction between particles $i$ and $j$ contributes
\bb
-\frac{Gm_jm_k}{\left\| \vec{R}_j - \vec{R}_k\right\|}=-\frac{Gm_jm_k}{\left\|\vec{d}_{jk} - x_j \hat{n}_j + x_k \hat{n}_k\right\|}
\ee
to the total energy.

We start by computing
\begin{align}
&\left\|\vec{d}_{jk} - x_j \hat{n}_j + x_k \hat{n}_k\right\|^{-1} \\ &\qquad \eqt{1} \frac{1}{d_{jk}} \left\| \hat{d}_{jk} - \frac{x_j}{d_{jk}}\, \hat{n}_j + \frac{x_k}{d_{jk}}\, \hat{n}_k \right\|^{-1} \nonumber\\
&\qquad = \frac{1}{d_{jk}} \bigg( 1 - \frac{2}{d_{jk}} \left( x_j\, \hat{n}_j\cdot \hat{d}_{jk} - x_k\, \hat{n}_k\cdot \hat{d}_{jk} \right) \nonumber\\
&\hspace{12.39ex} + \frac{1}{d_{jk}^2} \left( x_j^2+x_k^2 - 2x_jx_k\, \hat{n}_j\cdot \hat{n}_k \right) \bigg)^{-1/2} \nonumber\\
&\qquad \eqt{2} \frac{1}{d_{jk}} \bigg( 1 - \frac{2}{d_{jk}} \left( x_j \cos\theta_{jk} + x_k \cos\theta_{kj} \right) \nonumber\\
&\hspace{12.39ex} + \frac{1}{d_{jk}^2} \left( x_j^2+x_k^2 - 2x_jx_k \cos\varphi_{jk} \right) \bigg)^{-1/2} \nonumber\\
&\qquad \eqt{3} \frac{1}{d_{jk}} \bigg( 1 + \frac{1}{d_{jk}} \left( x_j \cos\theta_{jk} + x_k \cos\theta_{kj} \right) \label{expansion_jnverse_distance} \\
&\hspace{12.39ex} - \frac{1}{2d_{jk}^2} \left( x_j^2+x_k^2 - 2x_jx_k \cos\varphi_{jk} \right) \nonumber\\
&\hspace{12.39ex} + \frac{3}{8} \left(\frac{2}{d_{jk}}\right)^2 \left( x_j \cos\theta_{jk} + x_k \cos\theta_{kj} \right)^2 \nonumber\\
&\hspace{12.39ex} + \text{third-order terms} \bigg) \nonumber\\
&\qquad \eqt{4} -\frac{1}{2d_{jk}^3} \bigg( \left(1-3\cos^2\theta_{jk}\right) x_j^2 + \left(1-3\cos^2\theta_{kj}\right) x_k^2 \nonumber\\
&\hspace{14.5ex} - 2 \left( \cos\varphi_{jk} + 3 \cos\theta_{jk} \cos\theta_{kj} \right) x_j x_k \bigg) \nonumber\\
&\hspace{6.6ex} + \text{constant} + \text{linear terms} + \text{third-order terms.} \nonumber
\end{align}
Here, in~1 we introduced the notation $d_{jk}\coloneqq \big\|\vec{d}_{jk}\big\|$ and $\hat{d}_{jk}\coloneqq \frac{\vec{d}_{jk}}{d_{jk}}$, in~2 we defined the angles $\theta_{jk},\theta_{kj},\varphi_{jk}\in [0,\pi]$ via the relations~\eqref{angles}, in~3 we Taylor expanded
\bb
(1+z)^{-1/2} = 1 - \frac12 z + \frac38 z^2 + O\left(|z|^3\right)\qquad (z\to 0)\, ,
\ee
and finally in~4 it is understood that a linear term is a generic addend $c_j x_j$, while a third-order term is of the form $c_{jk\ell} x_j x_k x_\ell$. 

The Hamiltonian of the system therefore can be expressed as
\begin{align}
H &\eqt{5} \sum_j \left( \frac12 m_j \omega_j^2 x_j^2 + \frac{p_j^2}{2m_j} \right) - \sum_{j<k} \frac{G m_j m_k}{\left\|\vec{d}_{jk} - x_j \hat{n}_j + x_k \hat{n}_k\right\|} \nonumber \\
&\eqt{6} \sum_j \left( \frac12 m_j \omega_j^2 x_j^2 + \frac{p_j^2}{2m_j} \right) \nonumber \\
&\quad + \frac12 \sum_{j<k} \frac{G m_j m_k}{d_{jk}^3} \nonumber \\
&\hspace{10ex} \times \bigg( \!\left(1\!-\!3\cos^2\theta_{jk}\right) x_j^2 + \left(1\!-\!3\cos^2\theta_{kj}\right) x_k^2 \nonumber \\
&\hspace{14ex} - 2 \left( \cos\varphi_{jk}\! +\! 3 \cos\theta_{jk} \cos\theta_{kj} \right) x_j x_k \bigg) \nonumber \\
&\quad + \text{constant} + \text{linear terms} + \text{third-order terms} \nonumber \\
&= \sum_j \left( \frac12 m_j \omega_j^2 x_j^2 + \frac{p_j^2}{2m_j} \right) \label{Hamiltonian_expanded} \\
&\quad + \frac12 \sum_j \left( \sum_{k:\,k\neq j} \frac{G m_j m_k}{d_{jk}^3} \left(1\!-\! 3\cos^2\theta_{jk}\right) \right) x_j^2 \nonumber \\
&\quad - \frac12 \sum_{j\neq k} \frac{G m_j m_k}{d_{jk}^3} \left( \cos\varphi_{jk} \!+\! 3 \cos\theta_{jk} \cos\theta_{kj} \right) x_j x_k \nonumber \\
&\quad + \text{constant} + \text{linear terms} + \text{third-order terms} \nonumber \\[.5ex]
&\eqt{7} \sum_j \hbar \omega_j \frac{\bar{x}_j^2+\bar{p_j}^2}{2} \nonumber \\[-1.5ex]
&\quad + \frac12 \sum_j \left( \sum_{k:\,k\neq j} \frac{G \hbar m_k}{d_{jk}^3\omega_j} \left(1-3\cos^2\theta_{jk}\right) \right) \bar{x}_j^2 \nonumber \\
&\quad - \frac12 \sum_{j\neq k} \frac{G \hbar \sqrt{m_j m_k}}{d_{jk}^3 \sqrt{\omega_j\omega_k}} \left( \cos\varphi_{jk} + 3 \cos\theta_{jk} \cos\theta_{kj} \right) \bar{x}_j \bar{x}_k \nonumber \\
&\quad + \text{constant} + \text{linear terms} + \text{third-order terms,} \nonumber 
\end{align}
where~5 is just~\eqref{system_oscillators_H_tot}, in~6 we plugged~\eqref{expansion_jnverse_distance}, and finally in~7 we employed the dimensionless variables $\bar{x}_j,\bar{p}_j$ defined by~\eqref{dimensionless}. The quadratic term in~\eqref{Hamiltonian_expanded}, specialised to the case where $\omega_j\equiv \omega$ for all $j$, yields~\eqref{g}.

Thus, to complete the motivation of~\eqref{effective_Hamiltonian} it only remains to justify why we can neglect the constant and linear terms appearing in~\eqref{Hamiltonian_expanded}. The constant terms clearly commute with everything else and amount to an overall phase, so they can safely be ignored. What about the linear terms? As it turns out, they can be ignored, too, because the unitary evolution corresponding to a sum of linear and a quadratic Hamiltonian can be expanded as
\bb
e^{-\frac{it}{\hbar} \left( H_{\mathrm{lin}} + H_{\mathrm{q}} \right)} = e^{-\frac{it}{\hbar} \,\widetilde{H}_{\mathrm{lin}}}\, e^{-\frac{it}{\hbar}\, H_{\mathrm{q}}}\, ,
\label{getting_rid_H_linear}
\ee
where $\widetilde{H}_{\mathrm{lin}}\neq H_{\mathrm{lin}}$ is a new linear Hamiltonian, but the $H_\mathrm{q}$ on the right-hand side is the same as on the left-hand side. The above observation can be proved in many ways, but the most immediate is perhaps via the Baker--Campbell--Hausdorff formula, which dictates that $e^{-\frac{it}{\hbar} \left( H_{\mathrm{lin}} + H_{\mathrm{q}} \right)} e^{\frac{it}{\hbar}\, H_{\mathrm{q}}}$ can be re-written as an exponential of $-\frac{it}{\hbar} H_{\mathrm{lin}}$ plus a complicated expression involving nested commutators of $-\frac{it}{\hbar} \left( H_{\mathrm{lin}} + H_{\mathrm{q}} \right)$ and $\frac{it}{\hbar}\, H_{\mathrm{q}}$; since there is only one quadratic term, which naturally commutes with itself, and the canonical commutation relations~\eqref{CCR} lower the degree of any expression involving commutators by $2$, a little thought reveals that only linear terms can survive any chain of nested commutators of this sort. Thus, $e^{-\frac{it}{\hbar} \left( H_{\mathrm{lin}} + H_{\mathrm{q}} \right)} e^{\frac{it}{\hbar}\, H_{\mathrm{q}}}$ must be the exponential of a linear operator.
%e.g.\ by completing the square, re-writing $H_{\mathrm{lin}} + H_{\mathrm{q}}$, up to constants, as a shifted purely quadratic term, noting that a translation in the canonical operators can 
Now, since $\widetilde{H}_{\mathrm{lin}} = \sum_j c_j x_j$ and $x_j$ acts only on the $j^\text{th}$ oscillator, we see that $e^{-\frac{it}{\hbar} \,\widetilde{H}_{\mathrm{lin}}} = \bigotimes_j V_j$ is in fact a product unitary. Thanks to Lemma~\ref{local_unitaries_lemma}, this has no effect on the LOCC simulation fidelity.

%\section{Proof of Theorems~\ref{general_symplectic_bound_thm},~\ref{recap_thm}, and~\ref{recap_small_times_thm}} \label{sec_proofs}

\section{Proof of Theorem~\ref{general_symplectic_bound_thm}} \label{subsec_proof_general_symplectic_bound}

This appendix is devoted to the presentation of a complete mathematical proof of Theorem~\ref{general_symplectic_bound_thm}. To apply Theorem~\ref{general_bound_thm} to the Gaussian coherent state ensemble described in the statement of Theorem~\ref{general_symplectic_bound_thm}, we first need to compute the operator $R_{AA'}$ given by~\eqref{R_gamma} for the case of interest. This is given by
\bb
R_{AA'} \coloneqq \int\!\dd\alpha\ p_\lambda(\alpha)\ketbra{\alpha^*}_A \otimes \left( U_S\ketbra{\alpha} U_S^\dag\right)_{A'}\, ,
\label{R_AA'_special}
\ee
where the systems $A,A'$ are composed of $n$ modes each. Note that
\begin{align}
R_A &= \Tr_{A'} R_{AA'} \nonumber \\
&= \int_{\C^n} \!\dd^n\!\alpha\ p_\lambda(\alpha) \ketbra{\alpha^*} \nonumber \\
&= \bigotimes_{j=1}^n \int_\C \!\dd^2\!\alpha_j\ \frac{\lambda}{\pi}\, e^{-\lambda |\alpha_j|^2} \ketbra{\alpha_j}_{A_j} \label{R_A_special}
 \\
&= \bigotimes_{j=1}^n \left(\tau_{1/\lambda}\right)_{A_j} , \nonumber
\end{align}
where
\bb
\tau_\nu\coloneqq \frac{1}{\nu\!+\!1}\sum_{n=0}^\infty \left( \frac{\nu}{\nu\!+\!1}\right)^n \ketbra{n} = \int \!\dd^2 z\ \frac{1}{\pi \nu}\, e^{-|z|^2/\nu} \ketbra{z}
\label{thermal}
\ee
is the thermal state with mean excitation number $\nu\geq 0$, expressed in the second line of the above equation by means of its $P$-representation~\cite[Eq.~(4.5.36)]{BARNETT-RADMORE}. 
This suggests that we should take a thermal state as an ansatz for $\xi_A$ in~\eqref{general_bound}. In the interest of rigour, we will choose $\xi_A = \big(\tau_{1/\lambda_\mu}^{\otimes n}\big)_A$ to be a thermal state with a slighly larger mean photon number than $1/\lambda$. Here, $\mu>0$ is a parameter that we will take to $0$ later on, and $\lambda_\mu$ is defined by
\bb
\lambda_\mu \coloneqq \frac{\lambda-2\mu}{1+\mu}\, .
\label{lambda_mu}
\ee
The precise functional dependence of $\lambda_\mu$ on $\mu$ is immaterial, as long as $\lambda_\mu > \lambda$ for all $\mu>0$ and $\lim_{\mu\to 0^+} \lambda_\mu = \lambda$. The specific choice in~\eqref{lambda_mu} is made for convenience. 

For every fixed $J\subseteq [n]$, we then obtain that
\begin{align}
&F_{\cls}(\mathcal{E}_\lambda, U_S) \nonumber \\
&\ \leqt{1} \underset{\mu\to 0}{\mathrm{liminf}}\ d_{\max} \left( R_{AA'}^{\Gamma_{\!J}} \,\Big\|\, \big(\tau_{1/\lambda_\mu}^{\otimes n}\big)_A \otimes \id_{A'}\right) \nonumber\\
&\ \eqt{2} \underset{\mu\to 0}{\mathrm{liminf}}\ \lambda_{\max} \left( \Big(\tau_{1/\lambda_\mu}^{-\frac12}\Big)^{\otimes n}_A\, R_{AA'}^{\Gamma_{\!J}} \, \Big(\tau_{1/\lambda_\mu}^{-\frac12}\Big)^{\otimes n}_A \right) \label{symplectic_proof_eq1} \\
&\ \eqt{3} \underset{\mu\to 0}{\mathrm{liminf}}\ \lambda_{\max} \left( \big(Q_\mu^{AA'}\big)^{\Gamma_{\!A_{\!J^c}}\Gamma_{\!A'_{\!J}}} \right)\, . \nonumber
%\lambda_{\max}\left(\left(\big(R_A^{\Gamma_{\!J}}\big)^{-\frac12}\! \otimes \id_{A'}\right) R_{AA'}^{\Gamma_{\!J}} \left(\big(R_A^{\Gamma_{\!J}}\big)^{-\frac12}\! \otimes \id_{A'}\right) \right) \nonumber \\
%&\ \eqt{2} \min_{J\subseteq [n]} \lambda_{\max}\left(\left( R_A^{-\frac12}\! \otimes \id_{A'}\right)^{\Gamma_{\!J}} R_{AA'}^{\Gamma_{\!J}} \left( R_A^{-\frac12}\! \otimes \id_{A'}\right)^{\Gamma_{\!J}} \right) \label{symplectic_proof_eq1} \\
%&\ \eqt{3} \min_{J\subseteq [n]} \lambda_{\max}\left( \left(\left( R_A^{-\frac12} \otimes \id_{A'}\right) R_{AA'} \left( R_A^{-\frac12} \otimes \id_{A'}\right) \right)^{\Gamma_{\!J}} \right) \nonumber \\
%&\ \eqt{4} \min_{J\subseteq [n]} \lambda_{\max}\left( Q_{AA'}^{\Gamma_{\!A_{\!J^c}}\Gamma_{\!A'_{\!J}}} \right) . \nonumber
\end{align}
The above steps are justified as follows: in~1 we used Theorem~\ref{general_bound_thm}, and in particular~\eqref{general_bound}, employing as ansatz $\xi_A=\big(\tau_{1/\lambda_\mu}^{\otimes n}\big)_A$ (see~\eqref{thermal}); in~2 we leveraged~\eqref{d_max_lambda_max} ---  we will soon see that what is enclosed in brackets is indeed a bounded operator --- and omitted a tensor product with $\id_{A'}$ for brevity; finally, in~3 we introduced the operator
\bb
Q_\mu^{AA'} \coloneqq&\, \left( \Big(\tau_{1/\lambda_\mu}^{-\frac12}\Big)^{\otimes n}_A\, R_{AA'}^{\Gamma_{\!J}} \, \Big(\tau_{1/\lambda_\mu}^{-\frac12}\Big)^{\otimes n}_A \right)^{\Gamma_{\!A_{\!J^c}}\Gamma_{\!A'_{\!J}}} \\
=&\, \Big(\tau_{1/\lambda_\mu}^{-\frac12}\Big)^{\otimes n}_A\, R_{AA'}^{\Gamma_A} \, \Big(\tau_{1/\lambda_\mu}^{-\frac12}\Big)^{\otimes n}_A  
\label{Q_AA'}\, .
\ee
where we chose the local Fock basis as the one with respect to which we are taking the transposition, and noted that a thermal state is diagonal in this basis.

Continuing, we compute
\begin{widetext}
\begin{align}
Q_\mu^{AA'} \eqt{4}& \int\! \dd^n\!\alpha\, p_\lambda(\alpha) \left(\tau_{1/\lambda_\mu}^{-\frac12}\right)^{\otimes n}_A \ketbra{\alpha}_A \left(\tau_{1/\lambda_\mu}^{-\frac12}\right)^{\otimes n}_A \otimes \left( U_S\ketbra{\alpha} U_S^\dag\right)_{A'} \nonumber \\
\eqt{5}&\ U_S^{A'} \left( \bigotimes_{j=1}^n \int_\C \!\dd^2\!\alpha_j\, \frac{\lambda}{\pi}\, e^{-\lambda |\alpha_j|^2} \left( \tau_{1/\lambda_\mu}^{-\frac12} \ketbra{\alpha_j} \tau_{1/\lambda_\mu}^{-\frac12}\right)_{A_j} \otimes \ketbra{\alpha_j}_{A'_j} \right) \Big(U_S^{A'}\Big)^\dag \nonumber \\
\eqt{6}&\ U_S^{A'} \left( \bigotimes_{j=1}^n \int_\C \!\dd^2\!\alpha_j\, \frac{\lambda(1\!+\!\lambda_\mu)}{\pi \lambda_\mu}\, e^{-(\lambda-\lambda_\mu) |\alpha_j|^2}\, \Ketbra{\sqrt{1\!+\!\lambda_\mu}\, \alpha_j}_{A_j} \otimes \ketbra{\alpha_j}_{A'_j} \right) \Big(U_S^{A'}\Big)^\dag \nonumber \\
\eqt{7}&\ U_S^{A'} \left( \bigotimes_{j=1}^n  U_{\varphi_\mu}^{A_j A'_j} \left(\int_\C \!\dd^2\!\alpha_j\, \frac{\lambda(1\!+\!\lambda_\mu)}{\pi \lambda_\mu}\, e^{-(\lambda-\lambda_\mu) |\alpha_j|^2} \ketbra{0}_{A_j} \otimes \Ketbra{\sqrt{2\!+\!\lambda_\mu}\, \alpha_j}_{A'_j} \right) \bigg(U_{\varphi_\mu}^{A_j A'_j}\bigg)^\dag \right) \Big(U_S^{A'}\Big)^\dag \label{symplectic_proof_eq2} \\
\eqt{8}&\ U_S^{A'} \left( \bigotimes_{j=1}^n  U_{\varphi_\mu}^{A_j A'_j} \left(\int_\C \!\frac{\dd^2\beta_j}{2\!+\!\lambda_\mu}\, \frac{\lambda(1\!+\!\lambda_\mu)}{\pi \lambda_\mu}\, e^{-\frac{\lambda-\lambda_\mu}{2+\lambda_\mu} |\beta_j|^2} \ketbra{0}_{A_j} \otimes \Ketbra{\beta_j}_{A'_j} \right) \bigg(U_{\varphi_\mu}^{A_j A'_j}\bigg)^\dag \right) \Big(U_S^{A'}\Big)^\dag \nonumber \\
\eqt{9}&\ \frac{1}{\mu^n}\left(\frac{(1\!+\!\lambda_\mu) \lambda \mu}{\lambda_\mu(\lambda\!-\!\lambda_\mu)}\right)^n U_S^{A'} \left( \bigotimes_{j=1}^n  U_{\varphi_\mu}^{A_j A'_j} \right) \left(\ketbra{0}_A\otimes \left(\tau_{1/\mu}^{\otimes n}\right)_{A'} \right) \left( \bigotimes_{j=1}^n  U_{\varphi_\mu}^{A_j A'_j} \right)^\dag \Big(U_S^{A'}\Big)^\dag \nonumber \\
\eqt{10}&\ \frac{1}{\mu^n}\left(\frac{1\!+\!\lambda}{2\!+\!\lambda} + o(1)\right)^n U_S^{A'} \omega_\mu^{AA'} \Big(U_S^{A'}\Big)^\dag\, . \nonumber
\end{align}
\end{widetext}
The above steps can be justified as follows. In~4 we used~\eqref{R_AA'_special} and~\eqref{Q_AA'}. In~5 we decomposed every operator except $U_S$ mode-wise. In~6 we noticed that 
\bb
\tau_{1/\zeta}^{-\frac12}\ket{\beta} = \sqrt{\frac{1+\zeta}{\zeta}}\, e^{\zeta |\beta|^2/2} \Ket{\sqrt{1+\zeta}\, \beta}\, ,
\ee
for all $\beta\in \C$ and $\zeta>0$, as can be verified directly by writing everything out in Fock basis. In~7 we set $\varphi_\mu\coloneqq \arccos\frac{1}{\sqrt{2+\lambda_\mu}}$, introduced the beam splitter unitary acting on two modes defined by
\bb
U_\varphi \coloneqq e^{\varphi (a^\dag b - a b^\dag)}\, ,
\label{BS}
\ee
where $a,b$ are the annihilation operators of the first and second mode, respectively, and recalled that
\bb
U_\varphi \ketbra{\beta, \gamma} = \ket{\beta \cos\varphi + \gamma \sin\varphi,\, -\beta\sin\varphi +\gamma \cos\varphi} .
\label{BS_coherent_states}
\ee
for all $\beta,\gamma\in \C$. Incidentally, $U_\varphi$ is the Gaussian unitary (cf.~\eqref{U_S}) corresponding to a rotation matrix --- which is in particular symplectic. Namely,
\bb
U_\varphi = U_{\mathrm{R}(-\varphi)},\qquad \mathrm{R}(\varphi) \coloneqq \begin{pmatrix} \cos\varphi & -\sin\varphi \\ \sin\varphi & \cos\varphi \end{pmatrix} .
\label{BS_as_symplectic}
\ee
In~9 we changed variable, setting $\beta_j\coloneqq \sqrt{2+\lambda_\mu}\, \alpha_j$, and leveraged~\eqref{thermal}. Finally, in~10 we noted that $\lim_{\mu\to 0^+}\frac{(1+\lambda_\mu) \lambda \mu}{\lambda_\mu(\lambda-\lambda_\mu)} = \frac{1 + \lambda}{2 + \lambda}$, and defined the state
\bb
\omega_\mu^{AA'} \coloneqq \left( \bigotimes_{j=1}^n  U_{\varphi_\mu}^{A_j A'_j} \right) \left(\ketbra{0}_A\otimes \left(\tau_{1/\mu}^{\otimes n}\right)_{A'} \right) \left( \bigotimes_{j=1}^n  U_{\varphi_\mu}^{A_j A'_j} \right)^\dag .
\ee
Now, from~\eqref{symplectic_proof_eq1} we continue by writing
\bb
&F_{\cls}(\mathcal{E}_\lambda, U_S) \\
&\leq \underset{\mu\to 0}{\mathrm{liminf}}\ \lambda_{\max} \left( \big(Q_\mu^{AA'}\big)^{\Gamma_{\!A_{\!J^c}}\Gamma_{\!A'_{\!J}}} \right) \\
&\eqt{11} \left( \frac{1\!+\!\lambda}{2\!+\!\lambda} \right)^n\! \underset{\mu\to 0}{\mathrm{liminf}}\ \frac{1}{\mu^n}\, \lambda_{\max}\left( \left( U_S^{A'} \omega_\mu^{AA'} {U_S^{A'}}^\dag \right)^{\Gamma_{\!A_{\!J^c}}\Gamma_{\!A'_{\!J}}} \right)
\label{summarised_bound}
\ee
where 11~comes from~\eqref{symplectic_proof_eq2}.

We now need to evaluate the rightmost side of~\eqref{summarised_bound}. To do so, it is instrumental to observe that the operator whose maximal eigenvalue we need to calculate is the partial transpose of a centred Gaussian state, and thus it is itself a centred Gaussian operator with trace one. Denoting with $W$ its quantum covariance matrix, we can thus write
\bb
\left( U_S^{A'} \omega_\mu^{AA'} {U_S^{A'}}^\dag \right)^{\Gamma_{\!A_{\!J^c}}\Gamma_{\!A'_{\!J}}} &= \mathrm{G}_1\left[W_{J,\mu},0\right] ,
\ee
where the right-hand side is defined by~\eqref{Gaussian_operator}. We can compute the operator norm of the above operator thanks to~\eqref{operator_norm_QCM}, provided that we first compute $W_{J,\mu}$. Start by observing that $\omega_\mu^{AA'}$ is a centred Gaussian state, too, for all $\mu>0$. We will denote its covariance matrix by $V\big[ \omega_\mu^{AA'}\big]$. We have that
\bb
&V\big[\omega_\mu^{AA'}\big] \\
&\eqt{12} \begin{pmatrix} (\cos \varphi_\mu)\id & (\sin\varphi_\mu) \id \\ - (\sin\varphi_\mu) \id & (\cos\varphi_\mu) \id \end{pmatrix} \begin{pmatrix} \id & 0 \\ 0 & \left(\frac{2}{\mu} +1 \right) \id \end{pmatrix} \\
&\quad \times \begin{pmatrix} (\cos \varphi_\mu)\id & -(\sin\varphi_\mu) \id \\ (\sin\varphi_\mu) \id & (\cos\varphi_\mu) \id \end{pmatrix} \\
&= \frac{1}{(2\!+\!\lambda_\mu)\mu} \begin{pmatrix} \left(2(1\!+\!\lambda_\mu) + \mu(2\!+\!\lambda_\mu)\right)\id & 2\sqrt{1\!+\!\lambda_\mu}\, \id \\ 2\sqrt{1\!+\!\lambda_\mu}\, \id & \left(2+(2\!+\!\lambda_\mu)\mu\right)\id \end{pmatrix}\! ,
\label{V_omega_mu}
\ee
where~12 is because
\bb
V\left[ U_\theta^{\phantom{\dag}}\, \rho\, U_\theta^\dag \right] = \begin{pmatrix} \cos\theta & \sin\theta \\ -\sin\theta & \cos\theta \end{pmatrix} V[\rho] \begin{pmatrix} \cos\theta & -\sin\theta \\ \sin\theta & \cos\theta \end{pmatrix}
\ee
holds in general for all two-mode states $\rho$ and all angles $\theta\in \R$, with $U_\theta$ being as usual the beam splitter unitary~\eqref{BS}. This follows immediately by combining~\eqref{BS_as_symplectic} and~\eqref{QCM_transformation_symplectic}. Hence,
\bb
&W_{J,\mu} \\
&\ = V\left[\left( U_S^{A'} \omega_\mu^{AA'} {U_S^{A'}}^\dag \right)^{\Gamma_{\!A_{\!J^c}} \Gamma_{\!A'_{\!J}}}\right] \\
&\ \eqt{13} \begin{pmatrix} \Sigma_{J^c}^A\!\! & \\ & \!\Sigma_{J}^{A'}\! \end{pmatrix} \begin{pmatrix} \id_{A}\! & 0\! \\ \!0 & \!S_{A'} \end{pmatrix} V\big[ \omega_\mu^{AA'}\big] \begin{pmatrix} \id_{A} & 0 \\ 0 & S_{A'}^{\intercal} \end{pmatrix} \begin{pmatrix} \Sigma_{J^c}^A & \\ & \Sigma_{J}^{A'} \end{pmatrix} \\
&\ = \begin{pmatrix} \Sigma_{J^c}^A\!\! & \\ & \!\Sigma_{J}^{A'} S_{A'}\! \end{pmatrix} V\big[ \omega_\mu^{AA'}\big]  \begin{pmatrix} \Sigma_{J^c}^A & \\ & S_{A'}^{\intercal} \Sigma_{J}^{A'} \end{pmatrix} ,
\label{W}
\ee
where in~13 we used~\eqref{QCM_transformation_symplectic} and~\eqref{partial_transpose_Gaussian}, introducing the matrices
\bb
\Sigma_{J^c}^A &= \bigoplus_{j\in J} \id_{A_j} \oplus \bigoplus_{k\in J^c} \begin{pmatrix} 1 & \\ & -1 \end{pmatrix}_{A_{k}}\, ,\\
\Sigma_J^{A'} &= \bigoplus_{j\in J} \begin{pmatrix} 1 & \\ & -1 \end{pmatrix}_{\!A'_j} \oplus \bigoplus_{k\in J^c} \id_{\!A'_{k}} \, ,
\ee
which represent the partial transpositions with respect to $A_{J^c}$ and $A'_J$. 

Now, including the factor $\mu^{-n}$ coming from~\eqref{summarised_bound}, let us write
\begin{align}
&\frac{1}{\mu^n}\, \lambda_{\max} \left(\left( U_S^{A'} \omega_\mu^{AA'} {U_S^{A'}}^\dag \right)^{\Gamma_{\!A_{\!J^c}}\Gamma_{\!A'_{\!J}}} \right) \nonumber \\
&\quad = \frac{1}{\mu^n}\, \lambda_{\max}\left( \mathrm{G}_1\!\left[ W_{J,\mu},\,0\right] \right) \nonumber \\
&\quad \eqt{14} \frac{1}{\mu^n} \prod_{\ell=1}^{2n} \frac{2}{1+\nu_\ell\!\left(W_{J,\mu}\right)} \nonumber \\
&\quad \eqt{15} \frac{2^{2n}}{\mu^n} \sqrt{\prod_{\ell=1}^{4n} \frac{1}{1+\left|\spec_\ell\!\left(W_{J,\mu}\Omega_{AA'}\right) \right|}} \label{operator_norm_computation} \\
&\quad \eqt{16} \frac{2^{2n}}{\mu^n} \frac{1}{\sqrt{\det (W_{J,\mu} \Omega_{AA'})}} \sqrt{\prod_{\ell=1}^{4n} \frac{1}{1+\left|\spec_\ell\!\left(\Omega_{AA'}W_{J,\mu}^{-1}\right)\right|}} \nonumber \\
&\quad \eqt{17} \frac{2^{2n}}{\mu^n} \frac{1}{\left( \frac{2}{\mu}+1 \right)^n} \sqrt{\prod_{\ell=1}^{4n} \frac{1}{1+\left|\spec_\ell\!\left(\Omega_{AA'}W_{J,\mu}^{-1}\right)\right|}}\, . \nonumber
%&\quad = \frac{1}{\left( 2+\mu \right)^n} \prod_{\ell=1}^{2n} \frac{2}{1+\nu_\ell\!\left(W_{J,\mu}^{-1}\right)}\, .
\end{align}
The above steps are justified as follows: 14~comes from~\eqref{operator_norm_QCM}; in~15 we leveraged~\eqref{symplectic_spectrum_eigenvalues} and introduced the notation $\spec_\ell(M)$ for the $\ell^{\text{th}}$ eigenvalue of a matrix $M$ (the order is not important here); in~16 we observed that the inverse of the eigenvalues of an invertible matrix coincide with the eigenvalues of the inverse matrix; finally, 17~is a consequence of the fact that
\bb
\det (W_{J,\mu} \Omega_{AA'}) = \det W_{J,\mu} = \left( \frac{2}{\mu}+1\right)^{2n} ,
\ee
as becomes clear by looking at~\eqref{V_omega_mu} and~\eqref{W} and recalling that $\det S=1$ as $S$ is symplectic.

Plugging~\eqref{operator_norm_computation} into~\eqref{summarised_bound} yields
\bb
F_{\cls}(\mathcal{E}_\lambda, U_S) \leq\ & \left( \frac{2(1\!+\!\lambda)}{2\!+\!\lambda} \right)^n \\
& \times \underset{\mu\to 0}{\mathrm{liminf}}\ \sqrt{\prod_{\ell=1}^{4n} \frac{1}{1+\left|\spec_\ell\!\left(\Omega_{AA'}W_{J,\mu}^{-1}\right)\right|}}\, .
\label{symplectic_proof_eq3}
\ee
We will next show how to evaluate the limit in~\eqref{symplectic_proof_eq3}. Proceeding formally from~\eqref{V_omega_mu}, one can see that
\bb
&V\!\left[\omega_\mu^{AA'}\right]^{-1} \\
&= \begin{pmatrix} (\cos \varphi_\mu)\id & (\sin\varphi_\mu) \id \\ - (\sin\varphi_\mu) \id & (\cos\varphi_\mu) \id \end{pmatrix} \begin{pmatrix} \id & 0 \\ 0 & \frac{\mu}{2+\mu}\, \id \end{pmatrix} \\
&\qquad \times \begin{pmatrix} (\cos \varphi_\mu)\id & -(\sin\varphi_\mu) \id \\ (\sin\varphi_\mu) \id & (\cos\varphi_\mu) \id \end{pmatrix}\, .
\ee
Thus,
\bb
&\lim_{\mu\to 0} V\!\left[\omega_\mu^{AA'}\right]^{-1}\! \\
&\qquad = \begin{pmatrix} (\cos \varphi_0)\id & (\sin\varphi_0) \id \\ - (\sin\varphi_0) \id & (\cos\varphi_0) \id \end{pmatrix} \begin{pmatrix} \id & 0 \\ 0 & 0 \end{pmatrix} \\
&\qquad \qquad \times \begin{pmatrix} (\cos \varphi_0)\id & -(\sin\varphi_0) \id \\ (\sin\varphi_0) \id & (\cos\varphi_0) \id \end{pmatrix} \\
&\qquad = \frac{1}{2+\lambda} \begin{pmatrix} \id & -\sqrt{1+\lambda}\, \id \\[0.4ex] -\sqrt{1+\lambda}\, \id & (1+\lambda)\, \id \end{pmatrix} ,
\label{limit_V_mu}
\ee
where in the last line we remembered that $\varphi_\mu\coloneqq \arccos\frac{1}{\sqrt{2+\lambda_\mu}}$. Therefore, from~\eqref{W} we see that
\bb
&W_{J,0}^{-1} \\
&\coloneqq \lim_{\mu\to 0} W_{J,\mu}^{-1} \\
&\eqt{18} \begin{pmatrix} \Sigma_{J^c}^A & \\ & \Sigma_{J}^{A'} S_{A'}^{-\intercal} \end{pmatrix} V\!\left[ \omega_\mu^{AA'}\right]^{-1} \begin{pmatrix} \Sigma_{J^c}^A & \\ & S_{A'}^{-1} \Sigma_{J}^{A'} \end{pmatrix} \\
&\eqt{19} \frac{1}{2\!+\!\lambda} \begin{pmatrix} \Sigma_{J^c}^A & \\ & \!\!\!\Sigma_{J}^{A'}\! S_{A'}^{-\intercal}\!\! \end{pmatrix} \begin{pmatrix} \id & -\sqrt{1\!+\!\lambda}\, \id \\[0.5ex] -\sqrt{1\!+\!\lambda}\, \id & (1\!+\!\lambda)\, \id \end{pmatrix} \begin{pmatrix} \Sigma_{J^c}^A & \\ & \!\!\!S_{A'}^{-1} \Sigma_{J}^{A'}\!\! \end{pmatrix} \\
&= \frac{1}{2\!+\!\lambda} \begin{pmatrix} \Sigma_{J^c}^A & \\ & \!\!\!\Sigma_{J}^{A'}\! S_{A'}^{-\intercal}\!\! \end{pmatrix} \begin{pmatrix} \id \\[0.5ex] -\sqrt{1+\lambda}\, \id \end{pmatrix} \\
&\hspace{4ex} \times \begin{pmatrix} \id & -\sqrt{1+\lambda}\, \id \end{pmatrix} \begin{pmatrix} \Sigma_{J^c}^A & \\ & \!\!\!S_{A'}^{-1} \Sigma_{J}^{A'}\!\! \end{pmatrix} \\
&\overset{20}{\eqqcolon} \frac{1}{2\!+\!\lambda}\, L_{J}^{\phantom{\intercal\!}} L_{J}^\intercal . %\hspace{40ex}
\label{W_0}
\ee
Here, in~18 we recalled~\eqref{W}, in~19 we used~\eqref{limit_V_mu}, and in~20 we defined
\bb
L_{J} \coloneqq \begin{pmatrix} \Sigma_{J^c}^A & \\ & \!\!\!\Sigma_{J}^{A'}\! S_{A'}^{-\intercal}\!\! \end{pmatrix} \begin{pmatrix} \id \\[0.5ex] -\sqrt{1+\lambda}\, \id \end{pmatrix} = \begin{pmatrix} \Sigma_{J^c}^A \\[0.7ex] -\sqrt{1+\lambda}\, \Sigma_{J}^{A'}\! S_{A'}^{-\intercal} \end{pmatrix} .
\label{L_J}
\ee
Thanks to~\eqref{W_0} and using the continuity of the spectrum~\cite[Chapter~II, \S~5.1]{KATO} we can take the limit $\mu\to 0$ explicitly in~\eqref{symplectic_proof_eq3}, obtaining that
\bb
F_{\cls}(\mathcal{E}_\lambda\!, U_S) \!\leq\! \left( \frac{2(1\!+\!\lambda)}{2\!+\!\lambda} \right)^n\!\! \sqrt{\prod_{\ell=1}^{4n} \frac{1}{1\!+\!\left|\spec_\ell\!\left(\Omega_{AA'} W_{\!J,0}^{-1}\right)\right|}} .
\label{summarised_bound_2}
\ee
Now,
\bb
&\spec\left( \Omega_{AA'} W^{-1}_{J,0} \right) \\
&\quad \eqt{21} \spec \left( \frac{1}{2+\lambda}\, \Omega_{AA'} L_J^{\phantom{\intercal\!}} L_J^\intercal \right) \\
&\quad \eqt{22} \spec \left( \frac{1}{2+\lambda}\, L_J^\intercal \Omega_{AA'} L_J^{\phantom{\intercal\!}} \right) \cup \{ \underbrace{0,\ldots, 0}_{\text{$2n$ times}}\} \\
&\quad \eqt{23} \spec \left( \frac{1}{2+\lambda} \left( \Omega_J - (1+\lambda)\, S^{-1}\Omega_J S^{-\intercal} \right) \right) \cup \{ \underbrace{0,\ldots, 0}_{\text{$2n$ times}}\} \\
&\quad \eqt{24} \frac{1}{2+\lambda} \spec\left(iZ_{\lambda,S,J}\right) \cup \{ \underbrace{0,\ldots, 0}_{\text{$2n$ times}}\}\, .
\label{W_0_symp_spectrum}
\ee
Here, in~21 we plugged in the expression~\eqref{W_0} for $W_{J,0}^{-1}$. In~22, instead, we recalled that for two matrices $M,N$ of sizes $h\times k$ and $k\times h$, respectively, if $h\geq k$ then
\bb
\spec(MN) = \spec(NM) \cup \{ \underbrace{0,\ldots,0}_{\text{$h-k$ times}} \}\, .
\ee
In~23 we calculated
\bb
&L_J^\intercal \Omega_{AA'} L_J \\
&\ = \begin{pmatrix} \Sigma_{J^c}^A & -\sqrt{1+\lambda}\, S_{A'}^{-1} \Sigma_{J}^{A'}\! \end{pmatrix} \begin{pmatrix} \Omega_A & \\ & \!\Omega_{A'}\!\! \end{pmatrix} \begin{pmatrix} \Sigma_{J^c}^A \\[0.7ex] -\sqrt{1+\lambda}\, \Sigma_{J}^{A'}\! S_{A'}^{-\intercal} \end{pmatrix} \\
&\ = \Sigma_{J^c} \Omega \Sigma_{J^c} + (1+\lambda) S^{-1} \Sigma_{J} \Omega \Sigma_{J} S^{-\intercal} \\
&\ = \Omega_J - (1+\lambda) S^{-1} \Omega_J S^{-\intercal} ,
\ee
where in the last line we observed, according to~\eqref{z_i_Omega_J}, that $\Sigma_{J^c} \Omega \Sigma_{J^c} = \Omega_J$ and analogously $\Sigma_{J} \Omega \Sigma_{J} = - \Omega_J$. Finally, in~24 we introduced the $2n\times 2n$ Hermitian matrix
\bb
Z_{\lambda,S,J} \coloneqq (1+\lambda) S^{-1} i\Omega_J S^{-\intercal} - i\Omega_J\, .
\label{Z_lambda_S_J}
\ee
From the above calculation, we see that the non-zero eigenvalues of $\Omega_{AA'} W_{J,0}^{-1}$ coincide in modulus with those of $Z_{\lambda,S,J}$, which we denoted in the statement of Theorem~\ref{general_symplectic_bound_thm} with $z_\ell(\lambda,S,J)$. Continuing from~\eqref{summarised_bound_2}, we thus obtain that
\begin{align}
F_{\cls}(\mathcal{E}_\lambda, U_S) &\leq \left( \frac{2(1\!+\!\lambda)}{2\!+\!\lambda} \right)^n \sqrt{\prod_{\ell=1}^{2n} \frac{1}{1+\frac{1}{2+\lambda}\left|z_\ell (\lambda, S,J)\right|}} \nonumber \\
&= \left( \frac{2(1\!+\!\lambda)}{2\!+\!\lambda} \right)^n \sqrt{\prod_{i=1}^{2n} \frac{2\!+\!\lambda}{2\!+\!\lambda\!+\!\left|z_\ell (\lambda, S,J)\right|}} \\
&= \frac{2^n(1\!+\!\lambda)^n}{\prod_{\ell=1}^{2n} \sqrt{2+\lambda+\left|z_\ell(\lambda,S,J)\right|}}\, . \nonumber
\end{align}
To conclude the proof, it suffices to minimise over $J\subseteq [n]$. Due to the fact that exchanging $J$ with $J^c$ amounts to the substitution $Z_{\lambda,S,J} \mapsto Z_{\lambda,S,J^c} = - Z_{\lambda,S,J}$, and the above upper bound only depends on the modulus of the eigenvalues, we see that restricting to subsets $J$ with $|J|\leq \floor{n/2}$ suffices. The same conclusion can be reached by appealing to Remark~\ref{nontrivial_J_rem}.

\section{Proof of Theorem~\ref{recap_thm}} \label{subsec_proof_recap_thm}

Most of the heavy lifting needed to prove Theorem~\ref{recap_thm} has been already done while establishing Theorem~\ref{general_symplectic_bound_thm}. Before we do the rest and provide a complete proof of Theorem~\ref{recap_thm}, it is worth to state a technical lemma that will play an important role in the argument. In it, we formalise a very simple version of the rotating-wave approximation.

\begin{lemma} \label{rotating_wave_lemma}
Let $X$ be a Hermitian matrix, and let $Z$ be any matrix of the same size as $X$. Then
\begin{equation}
\lim_{\kappa\to\infty} e^{-i\kappa X} e^{i(\kappa X + Z)} = e^{i\mathcal{P}_X(Z)}\, ,
\label{rotating_wave}
\end{equation}
where $\mathcal{P}_X(Z) \coloneqq \sum_j P_j Z P_j$ is the `$X$-pinched' operator $Z$, defined via the eigenprojectors $P_j$ of $X$.
\end{lemma}

\begin{proof}
Let the spectral decomposition of $X$ be $X=\sum_j \lambda_j P_j$, where each $\lambda_j$ is a distinct eigenvalue of $X$, so that $\lambda_j\neq \lambda_{j'}$ for $j\neq j'$. From the discussion in~\cite[Chapter~II, \S~2.1~\S~2.3]{KATO}, we see that for sufficiently small $\epsilon$ the spectral decomposition of $X+\epsilon Z$ is of the form $X+\epsilon Z = \sum_j \sum_k \lambda_{jk}(\epsilon) P_{jk}(\epsilon)$. Here, $\lambda_{jk}(\epsilon) = \lambda_j + \epsilon z_{jk}+ o(\epsilon)$ and $P_{jk}(\epsilon) = P_{jk} + o(1)$ as $\epsilon\to 0$, where $P_{jk}$ (respectively, $z_{jk}$) are the $k^\text{th}$ eigenprojector (respectively, the $k^\text{th}$ eigenvalue) of the $j^\text{th}$ block of $\mathcal{P}_X(Z)$ (cf.~\cite[Chapter~II, Eq.~(2.38)--(2.40)]{KATO}). Note that $\sum_k P_{jk} = P_j + o(1)$ (cf.~\cite[Chapter~II, Eq.~(2.39)]{KATO}), so that $P_{j'} P_{jk} = \delta_{j,j'} P_{jk}$. Putting all together,
\bb
&e^{-i\kappa X} e^{i\kappa (X + Z/\kappa)} \\
&\quad = \left( \sumno_{j'} e^{-i\kappa \lambda_{j'}} P_{j'} \right) \left( \sumno_j \sumno_k e^{i\kappa \lambda_{jk}(1/\kappa)} P_{jk}(1/\kappa) \right) \\
&\quad = \sum_{j,j'} \sum_k e^{-i\kappa\lambda_{j'} + i \kappa \lambda_{jk}(1/\kappa)} P_{j'} P_{jk}(1/\kappa) \\
&\quad = \sum_{j,j'} \sum_k e^{i\kappa (\lambda_j - \lambda_{j'}) + i z_{jk} + o(1)} \left( \delta_{j,j'} P_{jk} + o(1) \right) \\
&\quad \tends{}{\kappa\to\infty} \sum_j \sum_k e^{i z_{jk}} P_{jk} \\
&\quad = e^{i \mathcal{P}_X(Z)} ,
\ee
which completes the proof.
\end{proof}

We now move on to the proof of Theorem~\ref{recap_thm}. Since the effective Hamiltonian $H_{\mathrm{eff}}$ is given by $H_{\mathrm{eff}} = \frac{\hbar}{2} \bar{r}^\intercal (\omega \id_{2n} + \widetilde{g}) \bar{r}$ (see~\eqref{effective_Hamiltonian}), due to~\eqref{quadratic<-->symplectic} we have that
\bb
e^{-\frac{it}{\hbar}\, H_{\mathrm{eff}}} = U_S\, ,\qquad S = e^{\Omega(\omega \id_{2n} + \widetilde{g}) t} ,
\ee
where $\Omega$ is as usual the standard symplectic form defined in~\eqref{CCR}. Since we are estimating the LOCC simulation fidelity, thanks to Lemma~\ref{local_unitaries_lemma} we can freely apply a local unitary. We thus make the substitution
\bb
e^{-\frac{it}{\hbar} H_{\mathrm{eff}}} &\mapsto \left( \bigotimes\nolimits_j e^{i \frac{\hbar \omega}{2} \left(\bar{x}_j^2 + \bar{p}_j^2 \right)}\right) e^{-\frac{it}{\hbar}\, H_{\mathrm{eff}}}\, .
\label{recap_thm_proof_eq1}
\ee
Since
\bb
\bigotimes\nolimits_j e^{i \frac{\hbar \omega}{2} \left(\bar{x}_j^2 + \bar{p}_j^2 \right)} = e^{i \frac{\hbar \omega}{2} \sum_j \left(\bar{x}_j^2 + \bar{p}_j^2 \right)} = e^{i \frac{\omega t}{2} \bar{r}^\intercal \bar{r}}
\ee
and the correspondence $S\mapsto U_S$ is a group isomorphism (cf.\ discussion below~\eqref{quadratic<-->symplectic}), at the level of symplectic matrices the substitution~\eqref{recap_thm_proof_eq1} amounts to
\bb
S \mapsto e^{- \omega t \Omega} e^{\Omega(\omega \id_{2n} + \widetilde{g}) t}\, .
\ee
At this point, since the right-hand side is a symplectic matrix we could directly apply Theorem~\ref{general_symplectic_bound_thm}. Although this is certainly possible, here we want to leverage assumption~(III') above to simplify the resulting expression. To do so, we observe that because of assumption~(III'), and in particular of~\eqref{III'}, the oscillation frequency $\omega$ is much larger than the entries of $\widetilde{g}$. Taking formally the limit $\omega\to\infty$ amounts to making the rotating-wave approximation. We can do so rigorously by appealing to Lemma~\ref{rotating_wave_lemma}, which entails that the effective dynamics is well approximated by the symplectic matrix
\bb
S_{\mathrm{eff}} \coloneqq \lim_{\omega\to\infty} e^{- \omega t \Omega} e^{\Omega(\omega \id_{2n} + \widetilde{g}) t} = e^{\mathcal{P}_{i\Omega}(\Omega\widetilde{g}t)}\, ,
\ee
where
\bb
\mathcal{P}_{i\Omega}(Z) \coloneqq \frac12 \left(P_+ Z P_+ + P_- Z P_-\right) = \frac12 \left( Z + \Omega Z \Omega^\intercal\right)
\ee
is the $i\Omega$-pinching operation, and $P_\pm$ are the spectral projections corresponding to the eigenvalues $\pm 1$ of $i\Omega$, so that $i\Omega = P_+ - P_-$. It is easy to verify that
\bb
S_{\mathrm{eff}} = \mathcal{P}_{i\Omega}\big(\Omega \widetilde{g}\big) = \frac12 \left( \Omega \widetilde{g} + \widetilde{g} \Omega \right) = \frac12 \begin{pmatrix} 0 & g \\ -g & 0 \end{pmatrix} .
\ee
Therefore,
\bb
e^{\mathcal{P}_{i\Omega}(\Omega\widetilde{g}t)} = \begin{pmatrix} \cos\left(gt/2\right) & \sin\left(gt/2\right) \\ -\sin\left(gt/2\right) & \cos\left(gt/2\right) \end{pmatrix}
\ee
is indeed an orthogonal symplectic matrix.

Note that the right-hand side of~\eqref{general_symplectic_bound} is a continuous function of $S$, so that we can take the limit $\omega\to \infty$ on $S$ first and then compute the upper bound on $F_{\cls}$ via~\eqref{general_symplectic_bound}. Therefore, in the limit where assumptions~(I),~(II), and~(III') are met, and therefore in particular the above rotating-wave approximation can be made, we can compute an upper bound to $F_{\cls} \left( \EE_\lambda, e^{-\frac{it}{\hbar} H_{\mathrm{tot}}} \right)$ by evaluating the right-hand side of~\eqref{general_symplectic_bound} for $S=S_{\mathrm{eff}}$ given by~\eqref{S_eff}.

To perform this calculation effectively, it is useful to introduce the $2n\times 2n$ complex matrix
\bb
V_0 \coloneqq \frac{1}{\sqrt2} \begin{pmatrix} \id_n & \id_n \\ i\id_n & -i\id_n \end{pmatrix} ,
\label{V_0}
\ee
which satisfies
\bb
\Omega_J = \begin{pmatrix} 0 & \Xi_J \\ -\Xi_J & 0 \end{pmatrix} = V_0 \left(i\Xi_J \oplus (-i\Xi_J) \right) V_0^\dag
\label{Omega_J_diagonalised}
\ee
as well as
\bb
S_{\mathrm{eff}} &= \begin{pmatrix} \cos(g t/2) & \sin(gt/2) \\ -\sin(gt/2) & \cos(gt/2) \end{pmatrix} = V_0 \left( e^{i\frac{gt}{2}}\!\! \oplus e^{- i\frac{gt}{2}} \right) V_0^\dag ,
\label{S_eff_diagonalised}
\ee
where we use the shorthand notation $X \oplus Y \coloneqq \lsmatrix X & 0 \\ 0 & Y \rsmatrix$. Hence,
\begin{align}
&V_0^\dag \left((1+\lambda) S_{\mathrm{eff}}^{-1} i\Omega_J S_{\mathrm{eff}}^{-\intercal} - i\Omega_J\right) V_0 \nonumber \\
&\ \eqt{1} V_0^\dag \left((1+\lambda) S_{\mathrm{eff}}^{-1} i\Omega_J S_{\mathrm{eff}} - i\Omega_J\right) V_0 \nonumber \\
&\ = (1+\lambda) \left(V_0^\dag S_{\mathrm{eff}} V_0\right)^{-1}\! \left( V_0^\dag i\Omega_J V_0\right) V_0^\dag S_{\mathrm{eff}} V_0 - V_0^\dag i\Omega_J V_0 \nonumber \\
&\ \eqt{2} \left( -(1+\lambda) e^{-i\frac{gt}{2}} \Xi_J e^{i\frac{gt}{2}} + \Xi_J \right) \\
&\qquad \oplus \left( (1+\lambda) e^{i\frac{gt}{2}} \Xi_J e^{-i\frac{gt}{2}} - \Xi_J \right) \nonumber \\
&\ \eqt{3} \left(-\left( (1+\lambda) e^{i\frac{gt}{2}} \Xi_J e^{-i\frac{gt}{2}} - \Xi_J \right)^\intercal\right) \nonumber \\
&\qquad \oplus \left( (1+\lambda) e^{i\frac{gt}{2}} \Xi_J e^{-i\frac{gt}{2}} - \Xi_J \right) . \nonumber
\end{align}
Here, in~1 we used the fact that $S_{\mathrm{eff}}$ is an orthogonal matrix, 2~follows from~\eqref{Omega_J_diagonalised} and~\eqref{S_eff_diagonalised}, and in~3 we noticed that $g$ is real symmetric.
Remembering that $\{w_\ell(\lambda, gt, J) \}_{\ell=1,\ldots, n}$ is the spectrum of $(1+\lambda) e^{i\frac{gt}{2}} \Xi_J e^{-i\frac{gt}{2}} - \Xi_J$, the above calculation, together with the observation that a matrix and its transpose have the same spectrum, entails that
\bb
&\spec\left((1+\lambda) S_{\mathrm{eff}}^{-1} i\Omega_J S_{\mathrm{eff}}^{-\intercal} - i\Omega_J\right) \\
&\qquad = \{w_\ell(\lambda, gt, J)\}_{\ell=1,\ldots, n} \cup \{ - w_\ell(\lambda, gt, J)\}_{\ell=1,\ldots, n}\, .
\ee
Computing the right-hand side of~\eqref{general_symplectic_bound} using this information yields immediately the rightmost side of~\eqref{f_simplified}, thereby concluding the proof.

\section{Proof of Theorem~\ref{recap_small_times_thm}} \label{subsec_proof_recap_small_times}

We will now argue that if in addition to assumptions~(I),~(II), and~(III'), also assumption~(IV) is met, then we can expand the right-hand side of~\eqref{general_symplectic_bound} according to~\eqref{F_cls_expansion}, with $\Delta$ scaling as in~\eqref{Delta_big_O_estimate}, with explicit estimates reported in~\eqref{Delta_UB}--\eqref{Delta_UB_universal}. This will prove the first claim in Theorem~\ref{recap_small_times_thm}. Let us start by noting that the matrix in~\eqref{rotated_Xi_J} can be expanded as
\bb
(1+\lambda)\, e^{\frac{it}{2} g}\, \Xi_J\, e^{- \frac{it}{2} g} - \Xi_J &= \lambda \Xi_J + \, \frac{it}{2}\, [g,\Xi_J] + T\, ,
\label{expansion_rotated_Xi_J}
\ee
where $T$ is formally defined by~\eqref{expansion_rotated_Xi_J}, and contains terms of order $\lambda \frac{Gmt}{d^3\omega}$ or $\left(\frac{Gmt}{d^3\omega}\right)^2$. Although we do not indicate this explicitly, $T$ depends on $\lambda$, $gt$, and $J$. To estimate its magnitude, It is useful to introduce the Hilbert--Schmidt norm, defined for an $N\times N$ matrix $X$ by
\bb
\|X\|_2\coloneqq \sqrt{\Tr X^\dag X} = \sqrt{\sumno_{i,j=1}^N |X_{ij}|^2}\, .
\label{Hilbert_Schmidt_norm}
\ee
Let us now present a simple auxiliary lemma.
%, proved in Appendix~\ref{adjoint_perturbation_app}.

\begin{lemma} \label{adjoint_perturbation_lemma}
For all matrices $X,Y$, it holds that
\begin{align}
\left\| e^{X} Y e^{-X}\! - Y \right\|_2 &\leq \left(e^{2\|X\|_\infty}\! - 1 \right) \|Y\|_2\, , \label{adjoint_perturbation_1st} \\
\left\| e^{X} Y e^{-X}\! - Y - [X,Y] \right\|_2 &\leq \left(e^{2\|X\|_\infty}\! - 1 - 2\|X\|_\infty \right) \|Y\|_2\, . \label{adjoint_perturbation_2nd}
\end{align}
\end{lemma}

\begin{proof}
We will use several properties of two matrix norms, namely the Hilbert--Schmidt norm~\eqref{Hilbert_Schmidt_norm} and the operator norm~\eqref{operator_norm}. For a general introduction to the topic, see the monograph by Bhatia~\cite[\S~IV.2]{BHATIA-MATRIX}. We will make use of only two special cases of H\"older's inequality~\cite[Corollary~IV.2.6]{BHATIA-MATRIX}, i.e.
\begin{align}
\left\| XY \right\|_2 &\leq \min\left\{\|X\|_\infty \|Y\|_2,\, \|X\|_2\|Y\|_\infty\right\} , \label{Hoelder_2_norm} \\
\left\| XY \right\|_\infty &\leq \|X\|_\infty \|Y\|_\infty\, , \label{Hoelder_operator_norm}
\end{align}
valid for all pairs of matrices $X,Y$. As an immediate application of~\eqref{Hoelder_operator_norm} as well as of the triangle inequality, note that the operator norm of the matrix-valued function
\bb
\Delta_N(X) \coloneqq e^X - \sum_{k=0}^{N-1} \frac{X^k}{k!} = \sum_{k=N}^\infty \frac{X^k}{k!}\, ,
\label{Delta_N}
\ee
where $N\in \N_+$ is a positive integer, satisfies that
\bb
\left\|\Delta_N(X)\right\|_\infty \leq \sum_{k=N}^\infty \frac{\|X\|_\infty^k}{k!} = \delta_N(\|X\|_\infty)\, ,
\label{operator_norm_Delta_N}
\ee
where on the right-hand side we have the scalar version of $\Delta_N$, i.e.\ $\delta_N(x) \coloneqq e^x - \sum_{k=0}^{N-1} \frac{x^k}{k!}$. We are now ready to delve into the proof of~\eqref{adjoint_perturbation_1st}, for which we write
\begin{align}
&\left\| e^{X} Y e^{-X} - Y\right\|_2 \nonumber \\
&\qquad = \left\| \left(\id + \Delta_1(X)\right) Y \left(\id + \Delta_1(-X)\right) - Y \right\|_2 \nonumber \\
&\qquad \leqt{1} \Big(\left\|\Delta_1(X)\right\|_\infty\! + \left\| \Delta_1(-X) \right\|_\infty\! \nonumber \\
&\qquad\qquad + \left\|\Delta_1(X)\right\|_\infty\! \left\| \Delta_1(-X) \right\|_\infty \Big)\, \|Y\|_2 \\
&\qquad \leqt{2} \left( 2\delta_1\left(\|X\|_\infty\right) + \delta_1\left(\|X\|_\infty\right)^2 \right) \|Y\|_2 \nonumber \\
&\qquad = \left( \big( \delta_1\left( \|X\|_\infty\right) + 1\big)^2 - 1 \right) \|Y\|_2 \nonumber \\ 
&\qquad = \left( e^{2\|X\|_\infty} - 1 \right) \|Y\|_2\, . \nonumber
\end{align}
Here, 1~comes from repeated applications of~\eqref{Hoelder_2_norm} and~\eqref{Hoelder_operator_norm}, and in~2 we used~\eqref{operator_norm_Delta_N}.

The proof of~\eqref{adjoint_perturbation_2nd} is entirely analogous. We have that
\bb
&\left\| e^{X} Y e^{-X} - Y - [X,Y]\right\|_2 \\
&\ = \left\| \left(\id + X + \Delta_2(X)\right) Y \left(\id - X + \Delta_2(-X)\right) - Y - [X,Y] \right\|_2 \\
&\ = \left\| \left(X + \Delta_2(X)\right) Y \left(- X + \Delta_2(-X)\right) + \Delta_2(X) Y + Y \Delta_2(-X) \right\|_2 \\
&\ = \left\| \Delta_1(X) Y \Delta_1(-X) + \Delta_2(X) Y + Y \Delta_2(-X) \right\|_2 \\
&\ \leqt{3} \left(\left\|\Delta_1(X)\right\|_\infty \left\|\Delta_1(-X)\right\|_\infty + 2\|\Delta_2(X)\|_\infty \right) \|Y\|_2 \\
&\ \leqt{4} \left( \delta_1\left(\|X\|_\infty\right)^2 + 2 \delta_2\left(\|X\|_\infty\right)\right) \|Y\|_2 \\
&\ = \delta_2\left(2\|X\|_\infty\right) \|Y\|_2\, ,
\ee
which is precisely~\eqref{adjoint_perturbation_2nd}. In the above calculation, 3~comes from~\eqref{Hoelder_2_norm}--\eqref{Hoelder_operator_norm}, while 4~is because of~\eqref{operator_norm_Delta_N}.
\end{proof}

With Lemma~\ref{adjoint_perturbation_lemma} at hand, we observe that
\bb
\|T\|_2 &= \left\| (1+\lambda)\, e^{\frac{it}{2} g}\, \Xi_J\, e^{- \frac{it}{2} g} - \Xi_J - \lambda \Xi_J - \, \frac{it}{2}\, [g,\Xi_J] \right\|_2 \\
&\leqt{1} \lambda \left\| e^{\frac{it}{2} g}\, \Xi_J\, e^{- \frac{it}{2} g} - \Xi_J \right\|_2 \\
&\quad + \left\| e^{\frac{it}{2} g}\, \Xi_J\, e^{- \frac{it}{2} g} - \Xi_J - \frac{it}{2}\, [g,\Xi_J] \right\|_2 \\
&\leqt{2} \sqrt{n} \left(\lambda \left( e^{t\|g\|_\infty} - 1 \right) + e^{t\|g\|_\infty} - 1 - t\|g\|_\infty\right) .
\label{first_estimate_Delta}
\ee
Here, 1~is simply the triangle inequality, while in~2 we applied Lemma~\ref{adjoint_perturbation_lemma} together with the observation that $\|\Xi_J\|_2 = \sqrt{n}$ because $\Xi_J$ has only $n$ non-zero entries, all of modulus $1$ (see~\eqref{Xi_J} and~\eqref{Hilbert_Schmidt_norm}).

Now, let
\bb
\left\{ w'_\ell(\lambda, gt, J) \right\}_{\ell=1}^n \coloneqq \spec\left( \lambda \Xi_J + \, \frac{it}{2}\, [g,\Xi_J] \right)
\label{w_prime}
\ee
denote the spectrum of the matrix representing the first-order term on the right-hand side of~\eqref{expansion_rotated_Xi_J}. We will omit for the time being the dependence on $\lambda$, $gt$, and $J$, and write simply $w'_\ell$. Now, since the remainder term $T$ is small, we expect $w'_\ell \approx w_\ell$ upon re-ordering, where $w_\ell$ is the spectrum of the left-hand side of~\eqref{expansion_rotated_Xi_J}. This intuition is in fact correct, and we can make it rigorous by using a classic result in perturbation theory known as the Hoffman--Wielandt theorem~\cite[Theorem~6.3.5]{HJ1}. It states that (up to re-ordering $w'_\ell$) we can make
\bb
\sqrt{\sumno_{\ell=1}^n |w_\ell - w'_\ell|^2} \leq \|T\|_2\, .
\label{Hoffman_Wielandt}
\ee
In our case this is easily seen to imply that
\bb
&\sumno_{\ell=1}^n |w_\ell - w'_\ell| \\
&\ \leqt{3} \sqrt{n} \sqrt{\sumno_{\ell=1}^n |w_\ell - w'_\ell|^2} \\
&\ \leqt{4} n \left(\lambda \left( e^{t\|g\|_\infty} - 1 \right) + e^{t\|g\|_\infty} - 1 - t\|g\|_\infty\right) ,
\label{1_norm_w_vs_w_prime}
\ee
where 3~is just the Cauchy--Schwarz inequality, and~4 comes from the Hoffman--Wielandt theorem together with~\eqref{first_estimate_Delta}.

This whole endeavour of estimating the difference between $w_\ell$ and $w'_\ell$ is rewarding because --- as it turns out --- we can easily compute analytically the $w'_\ell$. Let us see how this is done. Decompose the space as $\C^n = \C^{|J|} \oplus \C^{|J^c|}$, and set for brevity $m\coloneqq |J|$ and $K = K(J) \coloneqq g^{J,J^c}$. Using the fact that $\Xi_J = \lsmatrix \id_m & 0 \\ 0 & -\id_{n-m} \rsmatrix$ with respect to this decomposition, we notice that
\bb
\frac{i}{2} [g, \Xi_J] = \begin{pmatrix} 0_{m} & -i K \\[1ex] i K^\intercal & 0_{n-m} \end{pmatrix} .
\label{commutator_g_Xi_J}
\ee
Thus,
\bb
\lambda \Xi_J + \, \frac{it}{2}\, [g,\Xi_J] = \begin{pmatrix} \lambda \id_m & -it K \\[1ex] it K^\intercal & - \lambda \id_{n-m} \end{pmatrix} .
\ee
The eigenvalues of the above matrix can be computed straight away. In fact, for a generic $x\in \R$ we can evaluate its characteristic polynomial as
\begin{align}
&\det \left( \begin{pmatrix} \lambda \id_m & -it K \\[1ex] it K^\intercal & - \lambda \id_{n-m} \end{pmatrix} - x\id_n \right) \nonumber \\
&\ \eqt{5} (-\lambda - x)^{n-m} \det \left( (\lambda-x)\id_{m} - \frac{t^2}{-\lambda - x}\, K^\intercal K \right) \label{recap_small_times_proof_char_pol} \\
&\ \eqt{6} (-1)^{n-m} (\lambda+x)^{n-2m} \det\left( (\lambda^2-x^2)\id_{m} + t^2 K^\intercal K \right) . \nonumber
\end{align}
Here, 5~is because of Schur's formula~\cite[Theorem~1.1]{ZHANG}
\bb
\det  \begin{pmatrix} A & B \\ C & D \end{pmatrix} = \det(D) \det\left( A - B D^{-1} C\right) ,
\ee
valid when $D$ is invertible (to prove the identity of two polynomials, it suffices to do so on a dense set, so we can always assume that $x\neq -\lambda$). In~6 we took a factor $(\lambda+x)^m$ inside the determinant --- note that by assumption $m=|J|\leq n/2$. The $n$ roots of the polynomial in~\eqref{recap_small_times_proof_char_pol}, which are nothing but the numbers $\{w'_\ell\}_{\ell=1,\ldots,n}$ defined by~\eqref{w_prime}, are easy to find: there is $x=-\lambda$ with multiplicity $n-2m$, and the $2m$ additional roots $x = \pm \sqrt{\lambda^2 + t^2 s_\ell^2}$ (here, $\ell=1,\ldots, m$), where $s_\ell = s_\ell(J)$ are the square roots of the eigenvalues of $K^\intercal K = \left( g^{J,J^c}\right)^\intercal g^{J,J^c}\!$, i.e.\ the singular values of $K=g^{J,J^c}$. Thus, up to re-ordering
\bb
\{w'_\ell\}_{\ell=1,\ldots, n} = \left\{\pm \sqrt{\lambda^2 + t^2 s_\ell^2} \right\}_{\ell=1,\ldots,m} \cup \{-\lambda\}_{\ell=2m+1,\ldots, n}\, .
\label{w_prime_computed}
\ee

Before we put everything together, we need to look more closely at the functions appearing in the right-hand side of~\eqref{f_simplified}. Define
\bb
f(\lambda,w)\coloneqq \frac{2(1+\lambda)}{2+\lambda+|w|}\, .
\ee
Start by observing that the variation of $f$ in $w$ for fixed $\lambda$ is bounded by
\bb
&\left| f(\lambda,w) - f(\lambda,w')\right| \\
&\qquad = \frac{2(1+\lambda)}{(2+\lambda+|w|)(2+\lambda+|w'|)}\left||w| - |w'|\right| \\
&\qquad \leq \frac{2(1+\lambda)}{(2+\lambda)^2} \left|w-w'\right| \\
&\qquad \leq \frac{|w-w'|}{2}\, .
\label{bounds_f_fixed_lambda}
\ee
On the other hand, it also holds that
\bb
\left| 1 - f(\lambda,w') \right| \leq \frac{|w'|-\lambda}{2}
\label{bound_f_difference_1st}
\ee
and that
\bb
\left| f(\lambda,w') - \left(1-\frac{|w'|-\lambda}{2}\right) \right| &= \left|\frac{|w'|^2 - \lambda^2}{2(2+\lambda+|w'|)}\right| \\
&\leq \frac{|w'|^2 - \lambda^2}{4}\, .
\label{bound_f_difference_2nd}
\ee
as long as $|w'|\geq \lambda$. The last ingredients we need are a couple of elementary lemmata.
%, also proved in Appendix~\ref{adjoint_perturbation_app}.

\begin{lemma} \label{product_difference_lemma}
Let $a_k,b_k\in [0,1]$ ($k=1,\ldots, N$) be $N$ pairs of numbers between $0$ and $1$. Then
\bb
\left|\prod_{k=1}^N a_k - \prod_{k=1}^N b_k \right| \leq \sum_{k=1}^N |a_k-b_k|\, .
\ee
\end{lemma}

\begin{proof}
Setting $b_0 = a_{N+1} = 1$, we have that
\bb
&\left|\prod_{k=1}^N a_k - \prod_{k=1}^N b_k \right| \\
&\quad = \left|\sum_{k=1}^N \left(b_0\ldots b_{k-1} a_{k}\ldots a_{N+1} - b_0\ldots b_{k} a_{k+1}\ldots a_{N+1}\right) \right| \\
&\quad = \left|\sum_{k=1}^N b_0\ldots b_{k-1} \left( a_k - b_{k}\right) a_{k+1}\ldots a_{N+1} \right| \\
&\quad \leq \sum_{k=1}^N |a_k-b_k|\, ,
\ee
as claimed.
\end{proof}

\begin{lemma} \label{delta_k_lemma}
Let $\delta_k,\delta'_k\in [0,1]$ ($k=1,\ldots, N$) be $N$ pairs of numbers between $0$ and $1$. Then
\bb
\left| \prod_{k=1}^N (1-\delta_k) - 1 - \sum_{k=1}^N \delta'_k \right| \leq \sum_{k=1}^N |\delta_k - \delta'_k| + \sum_{1\leq h<k\leq N} \delta_h \delta_k \, .
\ee
\end{lemma}

\begin{proof}
From Lemma~\ref{product_difference_lemma} we infer immediately that
\bb
0\leq 1 - \prod_{k=1}^N (1-\delta_k) \leq \sum_{k=1}^N \delta_k\, .
\label{product_delta_k_LB_1st}
\ee
The above inequality is of course well known, see e.g.~\cite[Chapter~2, \S~58]{INEQUALITIES}. Setting $\delta_0=0$, we have that
\bb
\sum_{k=1}^N \delta_k - \left(1 - \prod_{k=1}^N (1\!-\!\delta_k)\right) &= \sum_{k=1}^N \delta_k \left(1 - \prod_{h=0}^{k-1} (1\!-\!\delta_h) \right) ,
\ee
and hence
\bb
0\leq \sum_{k=1}^N \delta_k - \left(1 - \prod_{k=1}^N (1-\delta_k)\right) &\leq \sum_{k=1}^N \delta_k \sum_{h=1}^{k-1} \delta_h \\
&= \sum_{1\leq h<k\leq N} \delta_h \delta_k\, ,
\label{product_delta_k_LB_2nd}
\ee
where we made use of~\eqref{product_delta_k_LB_1st}. Putting all together, we have that
\bb
&\left| \prod_{k=1}^N (1\!-\!\delta_k) - 1 - \sum_{k=1}^N \delta'_k \right| \\
&\quad \leq \sum_{k=1}^N |\delta_k \!-\! \delta'_k| + \left|\prod_{k=1}^N (1\!-\!\delta_k) - 1 - \sum_{k=1}^N \delta_k \right| \\
&\quad \leq \sum_{k=1}^N |\delta_k \!-\! \delta'_k| + \sum_{1\leq h<k\leq N} \delta_h \delta_k \, ,
\ee
where in the last line we employed~\ref{product_delta_k_LB_2nd}. 
\end{proof}

We now have that
\begingroup
\addtolength{\jot}{.8ex}
\begin{align}
&\left|\frac{2^n(1+\lambda)^n}{\prod_{\ell=1}^n \left(2+\lambda+|w_\ell|\right)} - \left(1- \sum_{\ell=1}^{|J|} \left(\sqrt{\lambda^2+t^2 s_\ell^2} - \lambda\right) \right) \right| \nonumber \\
& \nonumber \\[-2.7ex]
&\ \eqt{7} \left|\prod_{\ell=1}^n f(\lambda, w_\ell) - \left(1- \frac12 \sum_{\ell=1}^{n} \left(|w'_\ell| - \lambda\right) \right) \right| \nonumber \\
&\ \leqt{8} \left|\prod_{\ell=1}^n f(\lambda, w_\ell) - \prod_{\ell=1}^n f(\lambda, w'_\ell) \right| \nonumber \\
&\ \quad + \left| \prod_{\ell=1}^n f(\lambda,w'_\ell) - \left(1- \frac12 \sum_{\ell=1}^{n} \left(|w'_\ell| - \lambda\right) \right) \right| \nonumber \\
&\ \leqt{9} \sum_{\ell=1}^n \left|f(\lambda, w_\ell) - f(\lambda, w'_\ell) \right| \nonumber \\
&\ \quad + \sum_{\ell=1}^n \left| 1 - f(\lambda,w'_\ell) - \frac{|w'_\ell|-\lambda}{2} \right| \nonumber \\
&\ \quad + \sum_{1\leq \ell < k \leq n} \left( 1 - f(\lambda,w'_\ell)\right) \left(1 - f(\lambda,w'_k)\right) \nonumber \\
&\ \leqt{10} \frac12 \sum_{\ell=1}^n \left|w_\ell - w'_\ell\right| + \frac14 \sum_{\ell=1}^n \left(|w'_\ell|^2 - \lambda^2\right) \nonumber \\
&\ \quad + \frac14 \sum_{1\leq \ell < k \leq n} \left(|w'_\ell|-\lambda\right)\left(|w'_k|-\lambda\right) \nonumber \\
&\ \leqt{11} \frac12 \sum_{\ell=1}^n \left|w_\ell - w'_\ell\right| \,+\, \frac{t^2}{2} \sum_{\ell=1}^m s_\ell^2 \,+\, t^2\!\! \sum_{1\leq \ell < k \leq m} s_\ell s_k \nonumber \\
&\ \eqt{12} \frac12 \sum_{\ell=1}^n \left|w_\ell - w'_\ell\right| \,+\, \frac{t^2}{2} \|K\|_1^2 \nonumber \\
&\ \leqt{13} \frac{n}{2} \left(\lambda\! \left( e^{t\|g\|_\infty}\! - 1 \right) + e^{t\|g\|_\infty}\! - 1 - t\|g\|_\infty\!\right) + \frac{n^2}{8}\, t^2 \|g\|_\infty^2\, . \hspace{3ex}
%&\quad \leqt{14} \frac{n}{2} \left(\lambda \left( e^{6(n-1) \gamma t} - 1 \right) + e^{6(n-1) \gamma t} - 1 - 6(n-1) \gamma t \right) \nonumber \\
%&\quad \hspace{2.3ex} + 2 n^3 (\gamma t)^2 \nonumber \\
%&\quad \approx 3n(n-1) \lambda \gamma t + n\left(11n^2 -18n+9\right) (\gamma t)^2\, . \hspace{4ex}
\label{basically_proof_recap_thm}
\end{align}
\endgroup
The above steps are justified as follows: in~7 we recalled the expression of $w'_\ell$, given by~\eqref{w_prime_computed}; 8~is simply the triangle inequality; in~9 we used Lemma~\ref{product_difference_lemma} for the first term --- note that $f(\lambda,w')\in [0,1]$ as long as $|w'|\geq \lambda$, which is the case here --- and Lemma~\ref{delta_k_lemma} for the second; 10~follows from the estimates in~\eqref{bounds_f_fixed_lambda}--\eqref{bound_f_difference_2nd}; in~11 we leveraged once again the expressions in~\eqref{w_prime_computed}, and noticed that $|w'_\ell|-\lambda \leq \sqrt{|w'_\ell|^2-\lambda^2}$; continuing, in~12 we remembered that the trace norm can also be written as the sum of the singular values, which in our case implies that $\|K\|_1 = \sum_{\ell=1}^m s_\ell$; finally, in~13 we employed~\eqref{1_norm_w_vs_w_prime} and observed that
\bb
n \|g\|_\infty \geq \|g\|_1 \geq \left\|\frac{i}{2}[g,\Xi_J]\right\|_1 = \left\| \begin{pmatrix} 0_m & -K \\ K^\intercal & 0_{n-m} \end{pmatrix} \right\|_1 = 2\|K\|_1
\label{norms_g_and_K}
\ee
thanks to the calculation in~\eqref{commutator_g_Xi_J}.

To complete the proof of~\eqref{F_cls_expansion} with the estimate on $\Delta$ given by~\eqref{Delta_UB}, it suffices to apply~\eqref{basically_proof_recap_thm} to the right-hand side of~\eqref{f_simplified}. The universal bound in~\eqref{Delta_UB_universal}, which implies immediately the big-\emph{O} estimate in~\eqref{Delta_big_O_estimate}, can be obtained by using in addition Proposition~\ref{geometrical_prop}, proved in Appendix~\ref{operator_norm_g_app} below. There, by means of a geometric reasoning we show that in a 3-dimensional space it is impossible to pack $n$ oscillators in such a way that one of them is close to all of the others. This in turn implies that $\|g\|_\infty \leq \gamma \min\left\{ 6(n-1),\ C_1 \ln(n-1) + C_2 \right\}$, where $\gamma$ is given by~\eqref{Delta_UB_universal}, and $C_1,C_2$ are some universal constants that do not depend on $n$.

We now move on to~\eqref{F_cls_simple_expansion}. To derive it from~\eqref{F_cls_expansion} using assumption~(IV'), i.e.~\eqref{lambda_effective_0_approximation}, it suffices to observe that
\bb
\sqrt{\lambda^2+t^2 s_\ell^2(J)} = t s_\ell + O(\lambda)\, ,
\ee
in turn implying that
\begin{align}
&1 - \max_{\substack{J\subseteq [n],\\ |J|\leq n/2}} \sum_{\ell=1}^{|J|} \left( \sqrt{\lambda^2 + t^2 s_\ell^2(J)} - \lambda \right) \nonumber \\
&\quad = 1 - t \max_{\substack{J\subseteq [n],\\ |J|\leq n/2}} \sum_{\ell=1}^{|J|} s_\ell(J) + O(n\lambda) \nonumber \\
&\quad = 1 - t \max_{\substack{J\subseteq [n],\\ |J|\leq n/2}} \left\| K(J) \right\|_1  + O(n\lambda) \\
&\quad = 1 - t \max_{\substack{J\subseteq [n],\\ |J|\leq n/2}} \frac14 \left\| [g,\Xi_J] \right\|_1  + O(n\lambda) \nonumber \\
&\quad = 1 - \eta t + O(n\lambda)\, , \nonumber
\end{align}
where we remembered the calculation in~\eqref{norms_g_and_K}, and defined the sensitivity $\eta$ as in~\eqref{sensitivity}. How does the above remainder $O(n\lambda)$ compare to the error terms inherited from~\eqref{Delta_big_O_estimate}, which are of the form $O\big(\lambda n \frac{Gmt}{d^3\omega} \big)$ and $O\Big(\big( n \frac{Gmt}{d^3\omega}\big)^2 \Big)$, provided that $\|g\|_\infty\sim n\frac{Gm}{d^3\omega}$, as in all cases of practical interest such as that in \S~\ref{subsubsec_many_oscillators}? It is not difficult to see that $n\lambda$ is much larger than the former by virtue of~\eqref{lambda_t_both_small}; hence, the total error can be expressed as in~\eqref{F_cls_simple_expansion}. Note that we expect $\eta$ to be of the order of
\bb
\eta \sim n \frac{Gm}{d^3\omega}\, ,
\ee
so that by virtue of~\eqref{lambda_effective_0_approximation} we have that
\bb
\eta t \gg \max\left\{ n\lambda,\, \left( n \frac{Gmt}{d^3\omega}\right)^2 \right\} ,
\ee
confirming the validity of the expansion in~\eqref{F_cls_simple_expansion}. This concludes the proof of Theorem~\ref{recap_small_times_thm}.

%\section{Proof of Lemmata~\ref{adjoint_perturbation_lemma},~\ref{product_difference_lemma}, and~\ref{delta_k_lemma}} \label{adjoint_perturbation_app}

%This appendix is devoted to the proof of the technical lemmata we needed in Section~\ref{subsec_proof_recap_small_times} to arrive at a complete proof of Theorem~\ref{recap_small_times_thm}.

%\begin{manuallemma}{\ref{adjoint_perturbation_lemma}}
%For all matrices $X,Y$, it holds that
%\begin{align}
%\left\| e^{X} Y e^{-X}\! - Y \right\|_2 &\leq \left(e^{2\|X\|_\infty}\! - 1 \right) \|Y\|_2\, , \tag{\ref{adjoint_perturbation_1st}} \\
%\left\| e^{X} Y e^{-X}\! - Y - [X,Y] \right\|_2 &\leq \left(e^{2\|X\|_\infty}\! - 1 - 2\|X\|_\infty \right) \|Y\|_2\, . \tag{\ref{adjoint_perturbation_2nd}}
%\end{align}
%\end{manuallemma}

%We now move on to the proof of Lemma~\ref{product_difference_lemma}, restated below.

%\begin{manuallemma}{\ref{product_difference_lemma}}
%Let $a_k,b_k\in [0,1]$ ($k=1,\ldots, N$) be $N$ pairs of numbers between $0$ and $1$. Then
%\bb \left|\prod_{k=1}^N a_k - \prod_{k=1}^N b_k \right| \leq \sum_{k=1}^N |a_k-b_k|\, . \ee
%\end{manuallemma}

%Finally, we prove Lemma~\ref{delta_k_lemma}.

%\begin{manuallemma}{\ref{delta_k_lemma}}
%Let $\delta_k,\delta'_k\in [0,1]$ ($k=1,\ldots, N$) be $N$ pairs of numbers between $0$ and $1$. Then
%\bb \left| \prod_{k=1}^N (1-\delta_k) - 1 - \sum_{k=1}^N \delta'_k \right| \leq \sum_{k=1}^N |\delta_k - \delta'_k| + \sum_{1\leq h<k\leq N} \delta_h \delta_k \, . \ee
%\end{manuallemma}

\section{The operator norm of the \emph{g} matrix} \label{operator_norm_g_app}

To obtain any concrete estimate from~\eqref{first_estimate_Delta}, we still need to upper bound $\|g\|_\infty$. 

\begin{prop} \label{geometrical_prop}
For any 3-dimensional arrangement of $n$ harmonic oscillators, the operator norm~\eqref{operator_norm} of the matrix $g$ defined by~\eqref{g} (cf.\ Figure~\ref{system_oscillators_fig}) satisfies that
\bb
\|g\|_\infty \leq \gamma \min\left\{ 6(n-1),\ C_1 \ln(n-1) + C_2 \right\} .
\label{geometrical}
\ee
Here, $\gamma = \frac{Gm}{d^3\omega}$, where $d\coloneqq \min_{j\neq \ell} d_{j\ell}$ is the minimal distance between oscillators and $m\coloneqq \max_j m_j$ is the maximal mass, and $C_1,C_2$ are some universal constants that do not depend on $n$. One can choose $C_1=288$ and $C_2=966$.
\end{prop}

\begin{proof}
We are going to use the simple bound
\bb
\|X\|_\infty \leq \max_{i=1,\ldots, N} \sum_{j=1}^N |X_{ij}|\, ,
\label{operator_norm_UB_max_row_norm}
\ee
valid for any $N\times N$ Hermitian (or anyway normal) $X$. To prove~\eqref{operator_norm_UB_max_row_norm}, a key first step is to notice that for normal $X$ we have that $\|X\|_\infty = \max_{\lambda\in \spec(X)}|\lambda|$ coincides with the largest modulus of the eigenvalues of $X$. From this, one can directly apply Gershgorin's theorem~\cite{VARGA} 
%~\cite{Gershgorin} (see also the excellent monograph~\cite{VARGA})
and conclude. Alternatively, it is possible to proceed directly: take an eigenvector $\ket{\psi} = \sum_{i=1}^N \psi_i \ket{i} \in \C^N \setminus\{0\}$ corresponding to the eigenvalue $\lambda$ such that $|\lambda| = \|X\|_\infty$. Choose $i\in \{1,\ldots, N\}$ such that $|\psi_i|=\max_{j=1,\ldots, N} |\psi_j| > 0$.
Then
\bb
|\lambda| |\psi_i| = \left|\sumno_{j}^{} X_{ij} \psi_j\right| \leq \left(\sumno_j |X_{ij}|\right) |\psi_i|\, ,
\ee
which proves~\eqref{operator_norm_UB_max_row_norm}. Applied to the matrix $g$ defined in~\eqref{g}, the estimate~\eqref{operator_norm_UB_max_row_norm} yields, upon straightforward computations, the bound
\bb
\|g\|_\infty &\leq 6\gamma \max_j \sum_{\substack{\ell=1,\ldots, n \\ \ell\neq j}} \left(\frac{d}{d_{j\ell}}\right)^3 ,
\label{g_operator_norm_UB}
\ee
where $\gamma = \frac{Gm}{d^3\omega}$, with $m = \max_j m_j$ being the maximum oscillator mass and $d = \min_{j\neq \ell} d_{j\ell}$ being the minimal distance between oscillator centres. How to proceed from here?

A simple solution is to estimate $\sum_{\ell:\,\ell\neq j} \frac{d^3}{d_{j\ell}^3}\leq n-1$, using the fact that \mbox{none of the terms in the sum exceed} $1$; in this way we arrive at
\bb
\|g\|_\infty \leq 6(n-1)\gamma\, ,
\ee
which reproduces the first bound in~\eqref{geometrical}. This is a good estimate when $n$ is not too large. However, a little thought reveals that for large $n$ it cannot be tight. Indeed, there is no geometrically feasible way of packing a very large number of points together so that the distance between any two of them is larger than a constant \emph{and}, at the same time, one of them is close to \emph{all} of the others. In fact, a simple geometric reasoning shows that for any arrangement of $n$ points $\vec{v}_1,\ldots, \vec{v}_n\in \R^3$ such that $\min_{j\neq \ell}\|\vec{v}_j - \vec{v}_\ell\|\geq 1$, it holds that
\bb
\max_j \sum_{\ell:\,\ell\neq j} \frac{1}{\|\vec{v}_j - \vec{v}_\ell\|^3} \leq 48 \ln(n-1) + 161\, ,
\label{geometrical_estimate}
\ee
i.e.\ the sum on the left-hand side grows at most logarithmically with $n$. Before proving~\eqref{geometrical_estimate}, note that we can plug it in~\eqref{g_operator_norm_UB} and obtain immediately
\bb
\|g\|_\infty &\leq 6\gamma \left( 48 \ln(n-1) + 161 \right) \\
&= \gamma \left(288\ln(n-1) + 966\right) ,
\ee
which concludes the proof.

Let us now prove~\eqref{geometrical_estimate}. Without loss of generality, we can assume that the maximum on the left-hand side is attained for $j=n$. For $k\in \N_+$, consider the spherical hulls of internal radius $k$ and external radius $k+1$ centred on $\vec{v}_n$, formally defined by
\bb
S_k \coloneqq \left\{ \vec{v}\in \R^3:\ k \leq \|\vec{v} - \vec{v}_n\| < k+1 \right\} .
\ee
Call $k_1 < k_2 <\ldots < k_m$ the $m$ distinct positive integers $k$ with the property that $S_k$ contains at least one of the $\vec{v}_\ell$, for some $\ell=1,\ldots, n-1$. Set $R_k \coloneqq S_k \cap \{\vec{v}_\ell \}_{\ell=1,\ldots,m}$. Note that since no $\vec{v}_\ell$ can be closer than $1$ to $\vec{v}_n$, all such points belong to $S_k$ for some $k\in \N_+$. Also, it holds that
\bb
k_a \geq a \quad\forall\ a=1,\ldots, m,\qquad m\leq n-1\, .
\label{geometric_lemma_proof_k_ell_LB}
\ee
Let us try to estimate the maximum cardinality of $R_k$. To this end, call $B_{1/2}(\vec{v}_\ell)$ the open Euclidean ball of radius $1/2$ centred on $\vec{v}_\ell$. By assumption, 
\bb
B_{1/2}(\vec{v}_\ell) \cap B_{1/2}(\vec{v}_{\ell'}) = \emptyset\qquad \forall\ \ell\neq \ell'\, .
\ee
Moreover, simple geometric considerations show that $\bigcup_{\vec{v}_\ell \in R_k} B_{1/2}(\vec{v}_\ell)$ is contained in a spherical hull of internal radius $k-1/2$ and external radius $k+3/2$, for all $k\in \N_+$. With a simple computation we surmise that
\bb
\frac{\pi}{6}\, |R_k| &= \sum_{\vec{v}_\ell \in R_k} \vol\left(B_{1/2}(\vec{v}_\ell)\right) \\
&\leq \frac43 \pi \left( \left(k+\frac32\right)^3 - \left(k-\frac12\right)^3 \right) \\
&= 8\pi \left(k^2+k+\frac{7}{12}\right) ,
\ee
so that
\bb
|R_k| \leq 48 \left(k^2+k+\frac{7}{12}\right) .
\label{geometric_lemma_proof_Rk_cardinality}
\ee
Now, we have that
\begin{align}
\sum_{\ell=1}^{n-1} \|\vec{v}_\ell - \vec{v}_n\|^{-3} &= \sum_{a=1}^m \sum_{\ell:\, \vec{v}_\ell \in R_{k_a}} \|\vec{v}_\ell - \vec{v}_n\|^{-3} \nonumber \\
&\leqt{1} \sum_{a=1}^m \frac{|R_{k_a}|}{k_a^3} \nonumber \\
&\leqt{2} \sum_{a=1}^m \frac{48 \left(k_a^2 + k_a +7/12\right)}{k_a^3} \nonumber \\
&\leqt{3} 48 \sum_{\ell=1}^{n-1} \frac{\ell^2 + \ell + 7/12}{\ell^3} \\
&\leq 48\left( \sum_{\ell=1}^{n-1} \frac{1}{\ell} + \sum_{\ell=1}^{\infty} \frac{1}{\ell^2} + \frac{7}{12} \sum_{\ell=1}^{\infty} \frac{1}{\ell^3} \right) \nonumber \\
&\leqt{4} 48 \left( \ln(n-1) + 1 + \frac{\pi^2}{6} + \frac{7}{12} \zeta(3) \right) \nonumber \\
&< 48 \ln(n-1) + 161\, . \nonumber
\end{align}
Here: 1~holds because $\|\vec{v}_\ell - \vec{v}_n\|\geq k$ whenever $\vec{v}_\ell\in S_k$; in~2 we employed~\eqref{geometric_lemma_proof_Rk_cardinality}; 3~follows due to~\eqref{geometric_lemma_proof_k_ell_LB}, thanks to the fact that $x\mapsto \frac{1}{x^3}\left(x^2+x+7/12\right)$ is a decreasing function; finally, in~4 we used the standard estimate $\sum_{\ell=1}^{n-1} f(\ell) \leq f(1) + \int_1^{n-1} \dd \ell\ f(\ell)$, valid whenever $f:[1,\infty) \to \R$ is non-negative and monotonically non-increasing, and introduced the zeta function $\zeta(s)\coloneqq \sum_{\ell=1}^\infty n^{-s}$ for $s>1$.

\end{proof}

\section{Noise sources}
\label{Noise}

In this appendix, we analyse some of the most relevant sources of noise for the proposal considered in the main text. 
\medskip

\paragraph{Mass fluctuations.}

Assume that a mass $\mu$ suddenly appears at a position $\vec{\delta}$ with respect to one of our oscillators. The gravitational potential induced on the oscillating mass $m$ with position $x \hat{n}$ can be written as
\bb
&-\frac{Gm\mu}{\big\|\vec{\delta} - x \hat{n}\big\|} \\
&\quad = \mathrm{constant} - \frac{Gm\mu}{\delta^2}\, (\hat{\delta}\cdot \hat{n})\, x + \mathrm{higher\ order\ terms.}
\ee
Introducing the normalised coordinate $\bar{x} \coloneqq \sqrt{\frac{\hbar \omega}{m}}\, x$ as in~\eqref{dimensionless} and setting $\cos\theta \coloneqq \hat{\delta}\cdot \hat{n}$, the resulting approximate effective Hamiltonian takes the form
\bb
\Delta H = - \frac{G\mu}{\delta^2}\,\sqrt{\frac{m}{\hbar \omega}}\,(\cos\theta)\, \bar{x}\, .
\ee
If the mass $\mu$ remains for a time $\Delta t$ and then suddenly disappears, the unitary induced by $\Delta H$ will be $e^{-\frac{i\, \Delta t}{\hbar}\, \Delta H} = D_v$, where $D_v$ is a displacement operator (see~\eqref{displacement}) with
\bb
\|v\| = \frac{G\mu}{\delta^2}\,\sqrt{\frac{m}{\hbar \omega}}\,|\cos\theta|\, \Delta t\, .
\label{v_noise}
\ee
Furthermore, since the mass is far away we can assume to first order that the effect will be the same on all oscillators. If a total of $t/\Delta t$ displacement operators of this sort act over a time interval $t$, and they are uncorrelated with each other, the overall effective displacement $v_{\mathrm{tot}}$ will satisfy
\bb
\|v_{\mathrm{tot}}\| = \frac{G\mu}{\delta^2}\,\sqrt{\frac{m t \Delta t}{3\hbar \omega}}\, ,
\label{v_noise_tot}
\ee
where we also used the fact that $\braket{\cos^2\theta} = 1/3$, with the average being over the entire solid angle. Now, if the external mass is detected its effect can be taken into account when computing the action of the unitary $U_S$. However, if it is undetected it will increase the probability of a type-1 error. This will be no longer zero, as predicted by~\eqref{P-I}, but instead will satisfy
\bb
1 - P^{\mathrm{noise}}_1 &= \int \!\frac{\dd^{2n}\!\alpha}{\pi^n}\ p_\lambda(\alpha) \left|\braket{\alpha| U_S^\dag D_{v_{\mathrm{tot}}}^{\otimes n} U_S^{\phantom{\dag}} |\alpha}\right|^2 \\
&= \left|\braket{0|D_{S^{-1}v_{\mathrm{tot}}^{\oplus n}}|0}\right|^2 \\
&= e^{-\frac12 \left\| S^{-1}v_{\mathrm{tot}}^{\oplus n} \right\|^2} \\
&= e^{-\frac{n}{2} \| v_{\mathrm{tot}} \|^2},
\ee
where we leveraged the identity $U_S^\dag D_v U_S^{\phantom{\dag}} = D_{S^{-1}v}$, which follows directly from~\eqref{U_S}, used~\eqref{displacement_vacuum} and~\eqref{coherent_multimode}, and assumed that $S$ (and thus also $S^{-1}$) is an orthogonal symplectic matrix, as established by Theorem~\ref{recap_thm}. For the sake of simplicity, in the above calculation we posited that the action of the noise takes place entirely at the end of the process. While this is not a realistic hypothesis, it suffices for the purpose of obtaining an order-of-magnitude estimate of the noise sensitivity. 

In order for the thought experiment we sketched in \S~\ref{subsec_experiments} not to be compromised by noise, the type-1 error probability in the noisy setting, $P_1^{\mathrm{noise}}$, should still be smaller than the type-2 error probability, in formula $P_1^{\mathrm{noise}} < P_2$. Considering as before $P_2=0.1$ and $n=100$, we obtain the requirement that $\|v_{\mathrm{tot}}\|^2 < \SI{2.1e-3}{}$. Using~\eqref{v_noise_tot}, this entails that
\bb
\frac{G\mu}{\delta^2}\,\sqrt{\frac{m t \Delta t}{\hbar \omega}} < \SI{7.9e-2}{}\, .
\ee
Considering as above $m = \SI{1.58e-10}{\kg}$, $\omega \sim \SI{e-3}{\Hz}$, $t= \SI{185}{s}$, and setting $\Delta t=\SI{1}{\s}$, $\mu=\SI{1}{\kg}$, we obtain an estimate
\bb
\delta < \SI{670}{\m}
\ee
for the minimal distance from the experiment of tolerable, kilogram-sized mass fluctuations that happen on time scales of seconds.

However, there may be a way to further increase this noise resistance by resorting to a simple trick. Namely, since far away masses act on all modes in the same way to first order, i.e.\ $\ket{\alpha_1}\otimes \ldots \otimes \ket{\alpha_n}\mapsto \ket{\alpha_1+\beta}\otimes \ldots \otimes \ket{\alpha_n+\beta}$, it may be possible to ignore this effect by focusing only on the \emph{relative} motion. By doing this, we would effectively restrict our attention to a sort of noise-free subspace. A similar idea has been proposed in a different context in~\cite{Pedernales2022b}.

\medskip
\paragraph{Random collisions with surrounding gas molecules}

We now consider as a noise source for our experiment random collisions with hydrogen gas molecules. Each molecule has mass $m_0 = 2.016\, \mathrm{amu} \approx \SI{3.35e-27}{\kg}$. At background temperature $T$, its quadratic average speed along one fixed axis
\bb
v_0 \coloneqq \braket{v_x^2} = \sqrt{\frac{kT}{m_0}}\, .
\ee
By comparison, what is the quadratic average speed $v$ of one of our oscillating masses? This can be deduced from the fact that $|\alpha|^2\sim 1/\lambda$ from the relation
\bb
\sqrt{\braket{v^2}} = \sqrt{\frac{\braket{p^2}}{m^2}} \sim \sqrt{\frac{\hbar m \omega |\alpha|^2}{m^2}} \sim \sqrt{\frac{\hbar \omega}{\lambda m}}\, .
\ee
Using $\lambda \geq \SI{4.26e-12}{}$ from~\eqref{lambda_range}, we have that $\sqrt{\braket{v^2}} < \SI{1.25e-8}{\m/\s} \ll v_0$ as long as $T\gg \SI{3.79e-20}{\K}$. Therefore, the approximation $\sqrt{\braket{v^2}} \ll v_0$ covers all realistic regimes.

Now, if the mass $m$ is still and a molecule with momentum $p_0 = m_0 v_0 =\sqrt{m_0 kT}$ bounces off it, it transfers to the mass a momentum $2p_0 \cos\theta$, where $\theta\in [0,\pi/2]$ is the angle of incidence. For the sake of getting an estimate, let us substitute $\cos\theta \mapsto \braket{|\cos\theta|} = 1/2$, where the average is taken on a uniform distribution over the entire solid angle. The momentum transfer then becomes simply $p_0$. Taking into account that the natural unit for momentum in a harmonic oscillator with mass $m$ and frequency $\omega$ is $\sqrt{\hbar m\omega}$, the corresponding displacement in phase space is thus
\bb
\Delta \alpha_0 = \frac{p_0}{\sqrt{\hbar m\omega}} = \sqrt{\frac{m_0 kT}{\hbar m\omega}} \sim \SI{5.28e-2}{} \sqrt{T[\SI{}{\K}]}\, ,
\ee
where $T[\SI{}{\K}]$ stands for the temperature measured in Kelvin. A sequence of $N$ random impacts at random times will therefore result in a total displacement in phase space
\bb
\Delta \alpha = \sqrt{N}\, \Delta \alpha_0 \sim \SI{5.28e-2}{} \sqrt{N\,T[\SI{}{\K}]}\, .
\ee
Since at the level of Hilbert space vectors
\bb
\left| \braket{\alpha| \alpha + \Delta \alpha} \right|^2 = e^{-|\Delta\alpha|^2} ,
\ee
we would like to have
\bb
|\Delta \alpha|^2 = N |\Delta \alpha_0|^2 = \frac{N m_0kT}{\hbar m\omega} \ll 1
\ee
in order not to destroy the signature of a gravity that behaves according to quantum mechanics. Naturally, this constrains the amount of molecule impacts we can tolerate. Since the impacts are random events, the only way to reduce their number is reducing the pressure of the gas our experiment is immersed in. What kind of low pressures are we talking about?

Let us assume that $N$ impacts occur over a time $t$ on our mass with side surface area $\pi R^2$ --- remember that we are dealing only with one-dimensional oscillators. Since each impact transfers on average a momentum $2p_0 \braket{\cos\theta} = p_0$, the total momentum transfer per unit surface area is
\bb
\frac{N \cdot p_0}{\pi R^2} = \frac{N \sqrt{m_0 kT}}{\pi R^2}\, .
\ee
Equating the momentum transfer per unit area and \emph{per unit time} with the gas pressure $P$ yields
\bb
N = \frac{\pi R^2 t P}{\sqrt{m_0 kT}}\, . 
\ee
The operationally important condition that $|\Delta\alpha|^2 \ll 1$ then becomes~\footnote{A rigorous calculation for a one-dimensional untrapped particle actually gives
\begin{equation*}
    \frac{\sqrt{2\pi} 16 R^2 t}{9\hbar m\omega}\, P \sqrt{m_0 kT} \ll 1\, ,
\end{equation*}
see e.g.~\cite[Appendix~B]{sinha2022dipoles}.}
\bb
|\Delta\alpha|^2 = \frac{\pi R^2 t}{\hbar m\omega}\, P \sqrt{m_0 kT} \ll 1\, ,
\ee
or with explicit numbers
\bb
\label{vacuum}
P[\SI{}{\Pa}]\, \sqrt{T[\SI{}{\K}]} \ll 10^{-15}\, .
\ee
These conditions are certainly very challenging to produce experimentally. However, they are entirely comparable e.g.\ with those required in~\cite[p.~4]{Bose2017} and indeed vacua at the level of~\eqref{vacuum} have been achieved experimentally in the context of measurements on antiprotons~\cite{sellner2017improved}.

It should be noted that this estimate assumes the limit in which the energy transfer in a single collision $E_r$ exceeds significantly the quantised energy of a single phonon in the trapping potential $\hbar\omega$. However, of interest is also the opposite parameter regime $\eta^2 = \frac{E_r}{\hbar\omega}\ll 1$ in which it becomes possible that the recoil from the background gas collision is taken up by the trapping potential as a whole and not the test mass which remains in the same motional state. More specifically, for a coherent state amplitude $\alpha$ the probability for the test mass take up one phonon of energy $\hbar\omega$ is given by $p = \eta^2(2|\alpha|^2 +1)$. In this situation, well known in ion trap physics as the Lamb--Dicke regime, the effect of collisional decoherence can be reduced further~\cite{wineland1998experimental}.

\medskip
\paragraph{Black-body radiation.}

For the estimate of the impact of momentum diffusion due to black body radiation, it is particularly important to account for the fact that the particle is harmonically bound at frequency $\omega$ and, for typical wavelengths of the black body radiation, deep in the Lamb--Dicke regime. Indeed, the recoil energy for a mass $m$ due to a photon of wavelength $\lambda$ is given by 
\begin{equation}
    E_r = \frac{2\pi^2\hbar^2}{\lambda^2 m}
\end{equation}
while the energy of a single phonon of frequency $\omega$ is given
by 
\begin{equation}
    E_p = \hbar\omega.
\end{equation}
At temperature $T$ the wavelength of maximal intensity is determined by Wien's displacement law to be 
\begin{equation}
    \lambda_p = \frac{b}{T}\, ,
\end{equation}
with $b=\SI{2898}{\um}$, whence the condition of the Lamb--Dicke limit, $E_r\ll E_p$, then yields
\begin{equation}
    T \ll \sqrt{\frac{m \omega b^2}{2\pi^2 \hbar}} \approx \SI{2.53e7}{\K}\, ,
\end{equation}
which can evidently be assumed to be satisfied. Now we estimate the probability for a phonon to be excited during the experiment of duration $t$. A particle of radius $R$ will exchange black-body radiation with its environment with a power $4\pi\sigma R^2 T^4$. Assuming for simplicity that these photons have the wavelength $\lambda_p$, then we find the number $N_t$ of scattered photons in the time interval $[0,t]$ to be given by
\begin{equation}
    N_t = \frac{2t\sigma R^2 T^3 b}{\hbar c}.
\end{equation}
The probability to excite a single phonon given the coherent state $|\alpha\rangle$ is then upper bounded by
\begin{equation}
    p = N_t \frac{E_r}{\hbar\omega} \big(|\alpha|^2 + 1\big) = \frac{4\pi^2\sigma R^2 T^5}{b m c \omega}\big(|\alpha|^2 + 1\big)t\, .
\end{equation}
Requiring this quantity to be much smaller than $1$ and inserting $m = \SI{1.58e-10}{\kg}$, $R=\SI{12.5}{\um}$, 
$\omega \sim \SI{e-3}{\Hz}$, $t= \SI{185}{\s}$ together with $|\alpha|^2=10$ we find the condition 
\begin{equation}
    T \ll %23K
    \SI{11}{\K}\, ,
\end{equation}
which is a relatively moderate requirement and compares favorably with~\cite{Bose2017}.

\medskip
\paragraph{Electric and magnetic stray fields.}

While it is evident that the test mass should be charge neutral to avoid random accelerations due to 
fluctuating electric field, we also need to assess the impact of diamagnetic and dielectric contributions
to electric and magentic field gradients. The acceleration of a diamagnetic particle in a magnetic field 
is given by
\begin{equation}
    a = \frac{\chi V B}{m\mu_0}\frac{\partial B}{\partial x}
\end{equation}
where $\mu_0 = 4\pi \times \SI{e-6}{\newton/\ampere^2}$ and $\chi = 6.9$ for gold. For a random but (over the duration of a single shot) constant choice of $B$ and $\frac{\partial B}{\partial x}$ we find 
\begin{equation}
    \Delta\alpha = \sqrt{\frac{\omega}{2m\hbar}} \frac{\chi V B \frac{\partial B}{\partial x}}{\mu_0}t^2 = \SI{7.79e11}{} \cdot B\frac{\partial B}{\partial x} 
\end{equation}
and for $B = \SI{1}{\nano\tesla}$ (achievable using $\mu$-metal shields~\cite{gohil2020measurements}) the moderate requirement of 
\begin{equation}
    \frac{\partial B}{\partial x} \ll \SI{1.3e-3}{\tesla/\m}\, .
\end{equation}
For the acceleration response of a dielectric body to an electric field gradient we find
\begin{equation}
    a = \frac{\epsilon\epsilon_0 V E }{m} \frac{\partial E}{\partial x}
\end{equation}
and thus for an electric field gradient of $E=\SI{1000}{\volt/\m}$ (which can be achieved in current levitation experiments~\cite{hebestreit2018sensing}) we find
\begin{equation}
    \Delta\alpha = \sqrt{\frac{\omega}{2m\hbar}} \epsilon\epsilon_0 V E \frac{\partial E}{\partial x} = 2962 \frac{\partial E}{\partial x} t^2
\end{equation}
 and the requirement 
\begin{equation}
    \frac{\partial E}{\partial x} \ll \SI{3.4e-4}{\volt/\m^2}\, .
\end{equation}
While necessitating some electric and magnetic field shielding, these represent relatively modest experimental requirements.

\section{Ground-state cooling of a torsion pendulum}
\label{GSCooling}

The idea here would be to produce an effective damping of the oscillator by applying a force proportional to the particle's velocity. For this, we are required to measure the position and velocity of the pendulum in real-time and adjust the input laser power accordingly, to increase or decrease the radiation force on the oscillator. The position is detected by measuring the output light intensity from the cavity. When the pendulum is displaced, the size of the cavity is changed, and thus its resonance frequency. Since the transmissivity of the input mirror depends on the frequency of the cavity, a change in cavity length affects the amount of light that enters the cavity. In turn, the amount of light that comes out of the cavity informs about its size and, therefore, about the position of the pendulum.  In order for this to be a viable cooling method, it is required that the measurement rate is greater than the thermal decoherence rate by ${\Gamma_{\rm meas} > \Gamma_{\rm th}/8}$~\cite{Komori2020}. The measurement rate is the time that it takes to resolve the zero point motion of the oscillator, $\Delta \theta_{\rm zpm}$, and it is given by
\begin{equation}
\Gamma_{\rm meas} = \frac{\Delta \theta^2_{\rm zpm}}{2 S^{\theta \theta}(\omega)},
\end{equation}
where $S^{\theta \theta}(\omega)$ is the spectral density of the angular position of the pendulum. The thermal decoherence rate is given by ${\Gamma_{\rm th} = k_B T_{\rm th}/(\hbar Q)}$, where $k_B$ is Boltzmann's constant and $T_{\rm th}$ the environmental temperature. Thus, we are looking to minimise the power spectral density of rotation angle $S^{\theta \theta} (\omega)$. Once technical noise sources are reduced below thermal noise, the variance of the angle is dominated by thermal torque fluctuations~\cite{Shao2015, Ando2020}. This is characterised by a spectral density
\begin{equation}
S^{\theta \theta}_{\rm th} = |\chi (\omega)|^2 \frac{4 K_B T_{\rm th} I \omega_I^2}{Q_m \omega},
\end{equation}
where $\chi (\omega)$ is the mechanical susceptibility of the oscillator. The requirement ${\Gamma_{\rm meas} > \Gamma_{\rm th}/8}$ is equivalent to requiring that the thermal spectral density is below the spectral density of the standard quantum limit, ${S^{\theta \theta}_{\rm th} < S^{\theta \theta}_{\rm sql}}$, a condition that is achieved provided that
\begin{equation}
\frac{C}{n_{\rm th}} > 1,
\end{equation}
where ${n_{\rm th} \approx k_B T_{\rm th}/ (\hbar \omega_I)}$ is the mean phonon occupation number at temperature $T_{\rm th}$, and  ${C = 2 g^2/(\gamma \kappa)}$ is the cooperativity of the optomechanical system, with ${g^2= P_{\rm circ} a \omega_L / (I L c \omega_I)}$ the squared optomechanical coupling.  Here, we have $P_{\rm circ}$ the optical power circulating inside the cavity, $\omega_L$ the laser frequency, $L$ the length of the cavity when the pendulum is at equilibrium, $a$ is the arm length of the pendulum, and $c$ the speed of light. Putting all things together, we get the following condition for ground-state cooling
\begin{equation}
\frac{\hbar a^2 Q \omega_L P_{\rm cir}}{I L \omega_I c \kappa} > k_B T_{\rm th}.
\end{equation}
In order to gain some intuition of how hard it is to achieve this regime in %the lab
an experimental setting, let us put some numbers~\cite{Ando2020},
\begin{equation}
\begin{array}{lll}
P_{\rm cir} = 
\SI{10}{\watt}\, , && \omega_L = (2 \pi)\,\SI{2.8e14}{\Hz}\, , \\
L = \SI{9e-2}{\m}\, , && \kappa = (2 \pi)\,\SI{0.5e6}{\Hz}\, .
\end{array}
\end{equation}
With these values, we obtain 
\begin{equation}
\frac{I \omega_I}{a^2 Q} T_{\rm th}< \SI{1.6e-9}{\kg \Hz \K}
%1.6 \cdot 10^{-9} \ {\rm Kg\ Hz \ K}
\label{eq:coolingCriteria}
\end{equation}
Let us consider that our pendulum has the shape of a dumbbell consisting of a bar of length $2a$ with two balls of mass $m$, each attached to one end of the bar, as depicted in Fig~\ref{fig:torsion}b. Assuming for simplicity that the weight of the bar is negligible compared to that of the balls, we can approximate $I \approx 2 m a^2$. For calculations, let us take ${m = \SI{e-4}{\kg}}$ and $a=\SI{e-1}{\m}$. Note that this is a mass significantly larger than those considered in earlier estimations. This becomes necessary when considering torsion pendula in order to achieve the required low frequencies. In the rest of our analysis we will show that all the required regimes are satisfied also in this larger mass regime. For our choice of mass we can reach an angular frequency of ${\omega_I = \SI{7e-3}{\Hz}}$ for the torsional mode, provided a state-of-the-art suspension wire with torsion constant $\tau = \SI{e-10}{\newton \m/\radian}$ is used. For these parameters, Eq.~\eqref{eq:coolingCriteria} reduces to
\begin{equation}
\frac{T_{\rm th}}{Q} > \SI{e-3}{\K}\, .
%10^{-3} \ {\rm K}.
\end{equation}
 State-of-the-art torque sensitivities in torsion pendula stand just one order of magnitude above the standard quantum limit~\cite{Ando2020}. Provided that these can reach the standard quantum limit, for a pendulum with a quality factor of $10^5$, ground-state cooling might be possible even at room temperature, while pendula with lower quality factors will require previous cooling of the thermal bath to cryogenic temperatures.

\end{document}